\documentclass[a4paper, 11pt]{article}

\usepackage[T1]{fontenc}
\usepackage[utf8]{inputenc}
\usepackage[italian, english]{babel}

\usepackage{amsmath, amssymb, amsthm, physics, mathtools, mathrsfs, bm}
\usepackage{slashed, dsfont}

\usepackage[final]{microtype}

\usepackage{authblk}
\usepackage[top=3cm, bottom=3cm, left=3cm, right=3cm]{geometry}
\usepackage{tocloft}
\usepackage[sf, bf, explicit]{titlesec}
\usepackage[toc]{appendix}
\usepackage{abstract}

\usepackage{caption}
\usepackage{xcolor}

\usepackage[backend=biber, sortcites=true, style=nature]{biblatex}
\usepackage{xurl}
\usepackage{csquotes}
\usepackage{comment}

\usepackage{tikz}
\usepackage[compat=1.1.0]{tikz-feynman}

\usetikzlibrary{patterns}
\usetikzlibrary{decorations.pathmorphing}
\usetikzlibrary{decorations.markings}
\usetikzlibrary{external}
\usetikzlibrary{positioning, arrows.meta}

\tikzfeynmanset{warn luatex=false}							
\pgfmathsetmacro\MathAxis{height("$\vcenter{}$")}			
\tikzset{with arrow/.style = {
   decoration={
     markings,
     mark=at position 0.5
          with {\arrow[xshift=0.8mm]{Stealth[width=1.1mm,length=1.6mm]}}
     },
   postaction=decorate}}

\bibliography{./Bibliography}

\usepackage{hyperref}

\newcommand{\Z}{\mathbb{Z}}									
\newcommand{\C}{\mathbb{C}}
\newcommand{\R}{\mathbb{R}}
\newcommand{\N}{\mathbb{N}}

\newcommand{\dgauss}[2]{\mathbb{P}_{#1}(\dd #2)}			
\newcommand{\gyr}{\mathfrak{a}}								
\newcommand{\rgyr}{\gyr_\textsc{z}^\textsc{r}}	
\newcommand{\mnorm}[1]			
	{\abs*{#1}_{\! \scriptscriptstyle \mathrm{M}}}			
\newcommand{\mink}{{\scriptscriptstyle{\mathrm{M}}}}		
\newcommand{\mproj}[1]{\mathcal{C}_{#1}}					
\newcommand{\BigO}{\mathcal{O}}								
\newcommand{\cutoff}{\Lambda}								
\newcommand{\id}{\mathds{1}}								

\newcommand{\sing}[2]{{#1}^{#2}}							

\renewcommand{\phi}{\varphi}								

\newcommand{\loc}{\mathcal{L}}								
\newcommand{\kloc}{\mathscr{L}}								
\newcommand{\ren}{\mathcal{R}}								
\newcommand{\kren}{\mathscr{R}}								
\newcommand{\vrt}{\mathrm{V}}								
\newcommand{\edg}{\mathrm{E}}								
\newcommand{\trees}{\mathcal{T}}							
\newcommand{\TExp}[1]{\mathcal{E}^{#1}_%
{\scriptscriptstyle{\mathrm{T}}}}							
\newcommand{\src}{\omega}									
\newcommand{\Src}{\Omega}
\newcommand{\weight}{\mathsf{u}}							

\newcommand{\torus}{{\scriptstyle{\mathbb{T}}}}				

\newcommand{\floor}[1]{\lfloor #1 \rfloor}					
\DeclareMathOperator{\val}{Val}								
\DeclareMathOperator{\supp}{supp}							

\newtheoremstyle{thm}{}{}{}{}{\bfseries\sffamily}{.}{.5em}{}	

\theoremstyle{thm}
\newtheorem{lm}{Lemma}
\newtheorem{defn}{Definition}
\newtheorem{theorem}{Theorem}

\newtheorem{rmk}{Remark}

\title{ \bfseries \sffamily \Large The $g-2$ in the neutral Electroweak model with cutoff: convergent expansion, RG and the Jackiw-Weinberg formula}

\author[1]{Vieri Mastropietro}
\author[2]{Michele Bianchessi}

\affil[1]{Università degli Studi di Roma ``La Sapienza'', Department of Physics, Rome -- Italy}
\affil[2]{Università degli Studi Roma Tre, Department of Mathematics and Physics, Rome -- Italy}

\begin{document}

\maketitle

\begin{abstract}
The prediction of the anomalous gyromagnetic factor of the electron,
started with the evaluation of the electromagnetic contribution by
Schwinger (1948) and 
of the weak contribution by Jackiw and Weinberg (1972), 
is one of the major successes of Quantum Field Theory and the Standard Model. 
The results obtained truncating the
series are in spectacular agreement with experiments. Yet, a mathematical justification and 
an estimate of the truncation error are problematic, being such series diverging and not asymptotic to any QFT.
For a non perturbative result, one has to consider the Standard Model as an effective theory 
valid up to certain energy scales.
In this paper we 
consider the neutral sector of the Electroweak model 
with a momentum cutoff; we rigorously prove that 
the anomalous gyromagnetic factor in the effective regularized theory 
coincides with the Jackiw-Weinberg result, obtained by the truncation of the formal expansion with no cutoffs (whose sum is not expected to exist), up to a regularization-dependent correction
which is subdominant in the weak coupling regime if the cutoff is smaller than the inverse coupling and larger than the boson mass.
The proof is based on a convergent expansions and Renormalization Group (RG) methods; cancellations
based on exact and approximated symmetries are needed to get lowest order dominance.
\end{abstract}

\newpage

\tableofcontents

\newpage

\section{Introduction and main results}  
\label{sec:introduction}
\subsection{The anomalous gyromagnetic factor}  
A spin-$1/2$ fermion with mass $m$ and charge $e$ has an intrinsic magnetic dipole moment $\bm{\mu} = -g (e/2m) \mathbf{S}$, where $\mathbf{S}$ is the spin three-vector and $g$ a dimensionless number known as \emph{gyromagnetic factor}. The quantity
$\gyr \equiv (g-2)/2$ is called \emph{anomalous gyromagnetic factor}; it receives nonzero contributions from all the fundamental interactions $\gyr = \gyr_{\textsc{qed}} + \gyr_{\textup{weak}} + \gyr_{\textsc{qcd}}$ and it can be computed in the Standard Model, the theory providing the more fundamental description of elementary particles.
The current theoretical value of $\gyr$ for the electron, obtained as result of decades of computations, is $\gyr_e = 0.0011596521817877\dots$, while its experimental value is $\gyr_e^{\textsc{exp}} = 0.00115965218059(13)$~\cite{Aoyama_2012}. In the case of the muon, the agreement is also impressive
but a bit less precise~\cite{Aguillard_2023}. Such results provide a very stringent
test for the validity of the Standard Model, and possible differences could give an insight on new physical phenomena.

The first theoretical prediction about the $g$-factor is due to Dirac, who obtained the result $g = 2, \gyr = 0$ in 1928~\cite{Dirac_1930}. Experimental deviations from this value were attributed to quantum radiative contributions due to the interaction with the electromagnetic field.
Their evaluation was problematic due to the appearence
of diverging integrals (``infinities'') in the computations due to the large energy behavior, and
it was finally achieved using the theory of renormalization.
Such a procedure is based on the fact that the parameters appearing in the model (called \emph{bare parameters}), are not the observed ones (called \emph{dressed parameters}), as they also aquire radiative corrections. It turns out that
one can absorb all the divergences in the bare parameters, thus
expressing the anomalous gyromagnetic factor as a series in terms of the dressed ones.
The corrections to $\gyr$ due to the electroweak forces (there are also corrections due to strong forces which are expected to be much smaller)
can be therefore written as a series expansion in $e$ and $\mathsf{g}_{\textup{w}}$ with 
\begin{equation} 
\label{eq:perturbative_expansion_g}
\gyr_{\textsc{ew}} \equiv \gyr_{\textsc{qed}} + \gyr_{\textsc{weak}} =\sum_{n,m} A_{n,m} e^n \mathsf{g}_{\textup{w}}^m
\end{equation}
where $e$ is the electric charge and $\mathsf{g}_{\textup{w}}$ the weak coupling; they can be expressed in terms
of the fine structure constant $\alpha$, that is $\alpha=e^2/(\hbar c) \approx 1/137$ and $\mathsf{g}_{\textup{w}}^2=4\pi \alpha/\sin^2 \theta_W$ with
$\sin^2\theta_W=0,232 \dots$ and $\theta_W$ is the Weinberg angle. The smallness of the couplings motivates the perturbative approach.
Due to perturbative renormalizablity, $A_{n,m}$ is finite at all orders. 
The lowest order electromagnetic contribution $A_{2,0} e^2=\alpha/\pi$
was obtained by Schwinger in 1948~\cite{Schwinger_1948}, while 
 $A_{4,0} e^4=-(\alpha/\pi)^2
0.328 \dots$ was initially computed by Karplus and Kroll in 1950~\cite{Karplus_Kroll_1950} with an error corrected by Petermann and Sommerfeld in 1958~\cite{Petermann_1958, Sommerfeld_1958}. 
$A_{6,0}$  has been computed analytically~\cite{Laporta_1996} and 
$A_{8,0}$, $A_{10,0}$
are known numerically~\cite{Aoyama_2012}. The lowest order contribution coming from
Weak forces is $A_{0,2}\mathsf{g}_{\textup{w}}^2=\gyr_{\textsc{z}, 1} +
\gyr_{\textsc{w}, 1}$, where $\gyr_{\textsc{z}, 1}$ and $\gyr_{\textsc{w}, 1}$ are the \emph{neutral} and the \emph{charged} contribution, respectively mediated by the $Z^0, W^\pm$ bosons~\cite{Weinberg_1964}. This contribution was computed in 1972 by Jackiw and Weinberg~\cite{Jackiw_Weinberg_1972} as well as Altarelli, Cabibbo and Maiani~\cite{Altarelli_1972} within the first perturbative order. In particular,
\begin{equation}
\label{eq:JW_result} 
\gyr_{\textsc{z}, 1} = \frac{m^2}{M^2} \frac{\lambda^2}{4\pi^2} \frac{1- 5 \kappa^2}{3} \left[ 1 + \BigO \left( \frac{m^2}{M^2} \right) \right]
=\bar \gyr_{\textsc{z}, 1}\left[ 1 + \BigO \left( \frac{m^2}{M^2} \right) \right]
\end{equation}
where $\lambda = \mathsf{g}_{\textup{w}} (4 \sin^2 \theta_W - 1)/(4 \cos \theta_W), \kappa =  (4\sin^2 \theta_W - 1)^{-1}$ and $M$ is the $Z^0$ mass. This prediction was historically relevant, because it anticipated the experimental discovery of weak neutral currents~\cite{Rousset_1994}. Note that $\gyr_{\textsc{z}, 1}$ is proportional to the square of the ratio between the fermion mass and the $Z^0$ mass, which is a very small quantity ($\approx 10^{-6}$) for the electron; it is still small but larger ($\approx 10^{-3}$) for the muon. By summing up the first few terms of~\eqref{eq:perturbative_expansion_g} and adding the QCD contribution, which is computed numerically, one gets the theoretical value mentioned above.

Despite the accuracy of the result obtained truncating the series expansions, the procedure appears mathematically unclear. 
The series are not convergent and even not asymptotic to any QFT, see e.\ g.~\cite{Fredenhagen_2007}; the Higgs sector, necessary to make the Standard Model renormalizable, is trivial, that is, 
there is no way to define the theory in the removed cutoff limit, unless the limiting theory is non-interacting (see~\cite{Frohlich_1982, Aizenman_Copin_2021}). 

A possible rigorous approach consists in considering the Standard Model as an effective theory valid up to certain energy scales and therefore requiring
an energy cutoff $\cutoff$ to be defined.
One expects that $\gyr$ (and all the physical quantities) are independent on the regularization $\cutoff$ 
and agree with the perturbative computations at lowest order,
up to an error which is small in a range of values of the couplings 
including their physical values.
There are however no results putting such expectations on a quantitative, rigorous footing. The main difficulty is that one has to quantitatively fuflill two opposite requirements: from one side, $\cutoff$ must large enough so that no effects are seen, as no sign of cutoff is experimentally observed; on the other hand, the theory is well-defined only up to an energy scale
that decreases if the coupling size increases.
In the present paper we present the first rigorous result
on the anomalous gyromagnetic factor, focusing on the $Z^0$ contributions of Weak forces
and proving upper and lower bounds for it. 
An announcement of a weaker version of this result  
with no techical details is in~\cite{Mastropietro_2024}.

\subsection{The regularized neutral Electroweak model}
\label{subsec:model_definition}

The neutral sector of the Electroweak model (see e.\ g.~\cite[Equation (14)]{Jackiw_Weinberg_1972} with $W=A=0$) describes a fermion $\psi$ with mass $m$ interacting with a boson $\mathcal{Z}_\mu$ with mass $M$. The theory is \emph{chiral}, in the sense that the $\mathcal{Z}_\mu$ boson couples in different ways with the left-handed and the right-handed Weyl components of $\psi$.
We consider an Euclidean regularization of the neutral sector of the Electroweak model obtained via an ultraviolet cutoff $\cutoff$, to be kept fixed, and an infrared finite-volume cutoff $L$, to be removed at the end.
The model is defined by its expectations, expressed by the functional integral
\begin{equation}
\expval{F} = \frac{\int \dgauss{\sing{g}{\le N}}{\psi}  \int \dgauss{v}{\mathcal{Z}} 
\exp[\sum_{s = \pm} \int_{\mathcal{M}_L} \!\! \dd[4]{x} \, \psi^+_s \sigma^\mu_s \psi^-_s i\lambda_s \mathcal{Z}_\mu] F(A,\psi)}{\int \dgauss{\sing{g}{\le N}}{\psi}  \int \dgauss{v}{\mathcal{Z}} 
\exp[\sum_{s = \pm} \int_{\mathcal{M}_L} \!\! \dd[4]{x} \psi^+_s \sigma^\mu_s \psi^-_s  i\lambda_s \mathcal{Z}_\mu]},
\end{equation}
where
\begin{itemize}
\item Spacetime is represented by the four-dimensional torus $\mathcal{M}_L \equiv [-L/2, L/2]^4$ with periodic boundary conditions and a finite volume $L^4$. The spacetime metric is \emph{Euclidean}, so the metric tensor is simply $\mathrm{diag}(1, 1, 1, 1)$. A natural norm on $\mathcal{M}_L$ that is compatible with periodic boundary conditions is given by $\abs{a}_L \equiv (L/\pi)(\sum_{j=1}^4 \sin^2(\pi a_j/L))^{1/2}$. For every $a \in \R^4$, we also define the vector $(a)_\torus = [L/(2\pi)](\sin(2\pi a_1/L), \dots, \sin(2\pi a_4/L))$.
\item The ultraviolet regularization is realized through the function $\chi_N(k) \equiv \chi_0(2^{-N}k)$, where $\chi_0$ is a rotationally invariant smooth function on $\R^4$ such that $\chi_0(k)=0$ for $\abs{k} \ge 2$ and $\chi_0(k)=1$ for $\abs{k} \le 1$. We assume that $\chi_0$ has Gevrey class $2$. The cutoff scale associated with $\chi_N$ is $\cutoff \equiv 2^N$.
\item Given the set $\mathcal{D}_{L, N} = (2\pi \Z/L)^4 \cap \supp(\chi_N)$ and a function $\hat{\mathcal{Z}}_{k \mu} \colon \mathcal{D}_{L, N} \to \R^{4 \abs*{\mathcal{D}_{L, N}}}$, we define the Gaussian measure
\begin{equation}
\label{eq:bosonic_gaussian_measure}
\dgauss{v}{\mathcal{Z}} \equiv \mathcal{N}_\mathcal{Z} \prod_{k \in \mathcal{D}_{L, N}} \prod_{\mu = 0}^3 \dd{\hat{\mathcal{Z}}_{k \mu}} \, \exp \left( - \frac{1}{2} \frac{1}{L^4} \sum_{k \in \mathcal{D}_{L, N}} \!\!\! \hat{\mathcal{Z}}_{k \mu} \left[\hat{v}^{-1}(k) \right]^{\mu \nu} \! \hat{\mathcal{Z}}_{-k, \nu} \right),
\end{equation}
where
\begin{equation}
\label{eq:free_boson_propagator}
v^{\mu \nu}(x-y) =\frac{1}{L^4} \sum_{k \in \mathcal{D}_{L}} \frac{\delta_{\mu \nu} \chi_N(k)}{\abs*{k}^2 + M^2} \, e^{i k \cdot (x-y)} \equiv \frac{\delta_{\mu \nu}}{L^4} \sum_{k \in \mathcal{D}_{L}} e^{i k \cdot (x-y)} \, \hat{v}(k),
\end{equation}
see e.\ g.~\cite{Bauerschmidt_2019}. The constant $\mathcal{N}_\mathcal{Z}$ is chosen so that $\int \dgauss{v}{\mathcal{Z}} = 1$.
\item Given the finite-dimensional Grassmann algebra generated by the anticommuting variables $\lbrace \hat{\psi}^\pm_{k \alpha} \rbrace_{\alpha=1, \dots 4, k \in \mathcal{D}_{L, N}}$, we define the Grassmann integral as
$\int \dd{\hat{\psi}^\pm_{k \alpha}} = 0, \int \dd{\hat{\psi}^\pm_{k \alpha}} \, \psi^{\pm}_{k \alpha} = 1$, see e.\ g.~\cite{Feldman_2002}, and we let
\begin{equation}
\label{eq:fermionic_gaussian_measure}
\dgauss{\sing{g}{\le N}}{\psi} \equiv \mathcal{N}_\psi \prod_{k \in \mathcal{D}_{L, N}} \prod_{\alpha = 1}^4 \dd{\hat{\psi}^+_{k \alpha}} \dd{\hat{\psi}^-_{k \alpha}} \, \exp \left( - \frac{1}{L^4} \sum_{k \in \mathcal{D}_{L, N}} \!\!\! \hat{\psi}^+_{k \alpha} \left[(\sing{\hat{g}}{\le N})^{-1}(k) \right]^{\alpha \beta} \! \hat{\psi}^-_{k \beta} \right)
\end{equation} 
with
\begin{equation}
\label{eq:free_fermion_propagator}
\sing{g}{\le N}(x) \equiv \frac{1}{L^4} \sum_{k \in \mathcal{D}_{L}} \frac{\chi_N(k)}{i k_\mu \gamma^\mu_N + \mathsf{m}_N} \, e^{i k \cdot x} \equiv \frac{1}{L^4} \sum_{k \in \mathcal{D}_{L}} e^{i k \cdot x} \, \sing{\hat{g}}{\le N}(k).
\end{equation}
As before, the constant $\mathcal{N}_\psi$ is chosen so that $\int \dgauss{\sing{g}{\le N}}{\psi} = 1$; the $4 \times 4$ complex matrices $\lbrace \gamma^\mu_N \rbrace_{\mu = 0}^3$, $\mathsf{m}_N$ take the form
\begin{equation}
\label{eq:gamma_matrices_and_mass_matrix}
\gamma^\mu_N \equiv 
\begin{pmatrix}
0					&	Z^+_N \sigma_+^\mu \\
Z^-_N \sigma_-^\mu 	&	0
\end{pmatrix},
\qquad
\mathsf{m}_N \equiv 
\begin{pmatrix}
m^+_N	&	0	\\
0		&	m^-_N
\end{pmatrix},
\end{equation}
where $\sigma_{\pm} \equiv (\id_{2 \times 2}, \mp i \sigma_1, \mp i \sigma_2, \mp i \sigma_3)$ and $\sigma_1, \sigma_2, \sigma_3$ are the standard Pauli matrices. In the following sections, we will often use the standard \emph{Euclidean gamma matrices}
\begin{equation}
\label{eq:standard_gamma_matrices}
\gamma^\mu \equiv 
\begin{pmatrix}
0					&	\sigma_+^\mu \\
\sigma_-^\mu 		&	0
\end{pmatrix},
\qquad
\gamma^5 \equiv 
\begin{pmatrix}
\id_{2 \times 2}	&	0 \\
0 					&	-\id_{2 \times 2}
\end{pmatrix}.
\end{equation}
together with the Feynman slash notation $\slashed{k} \equiv \gamma^\mu k_\mu$. The gamma matrices satisfy the anticommutation relations $\lbrace \gamma^\mu, \gamma^\nu \rbrace = 2 \delta^{\mu \nu}$.
\item We define $
\psi_\alpha^\epsilon(x) \equiv L^{-4} \sum_{k \in \mathcal{D}_{L, N}} \hat{\psi}_{k \alpha}^\epsilon \, e^{i \epsilon k \cdot x}$, thought as a Grassmann-valued smooth periodic function on $\mathcal{M}_L$. We will often use the notation
\begin{equation}
\psi^-(x) \equiv 
\begin{pmatrix}
\psi^-_1 \\ \psi^-_2 \\ \psi^-_3 \\ \psi^-_4
\end{pmatrix}(x)
\equiv 
\begin{pmatrix}
\psi^-_- \\[2pt] \psi^-_+
\end{pmatrix}(x),
\qquad
\psi^+(x) \equiv
\begin{pmatrix}
\psi_1^+ \\ \psi^+_2 \\ \psi^+_3 \\ \psi^+_4
\end{pmatrix}^{\!\!\mathsf{T}}(x)
\equiv 
\begin{pmatrix}
\psi^+_+		\\	\psi^+_-
\end{pmatrix}^{\!\!\mathsf{T}}(x),
\end{equation}
where the $\mathsf{T}$ superscript denotes transposition (so $\psi^+(x)$ is thought as a row vector) and the two-dimensional Grassmann fields $\lbrace \psi^-_s(x) \rbrace_{s = \pm}$ are equal to the left-handed ($s = -$) and right-handed ($s = +$) Weyl components of the Dirac field $\psi^-(x)$. Throughout the rest of the paper, we use Greek indices to denote the components of the Dirac fields $\psi^+, \psi^-$. Given a pair of components $\psi^+_\alpha, \psi^-_\beta$, the indices $\alpha, \beta$ are said to have \emph{opposite chiralities} if $\alpha, \beta \in \lbrace 1, 2 \rbrace$ or $\alpha, \beta \in \lbrace 3, 4 \rbrace$. Conversely, $\alpha, \beta$ are said to have \emph{similar chiralities} if $\alpha \in \lbrace 1, 2 \rbrace, \beta \in \lbrace 3, 4 \rbrace$ or $\alpha \in \lbrace 3, 4 \rbrace, \beta \in \lbrace 1, 2 \rbrace$.
\item Finally, the constants $\lambda_+, \lambda_-$ are related with the physical parameters as $2\lambda \equiv \lambda_+ + \lambda_-, 2\lambda \kappa \equiv \lambda_- - \lambda_+$.
\end{itemize}

The connected correlation functions of our theory are obtained from the generating functional
\begin{multline}
\label{eq:generating_functional}
W[J, \eta] \equiv \log \int \dgauss{\sing{g}{\le N}}{\psi} \, \dgauss{v}{\mathcal{Z}} \, \times \\
\times \exp \left[\sum_{s = \pm} \int_{\mathcal{M}_L} \!\! \dd[4]{x} \, [\psi^+_s \sigma^\mu_s \psi^-_s (i\lambda_s \mathcal{Z}_\mu + Z^{J, s}_N J_\mu) + \eta^+_s \psi^-_{-s} + \psi^+_s \eta^-_{-s}](x) \right].
\end{multline}
By letting
\begin{equation}
\label{eq:schwinger_functions}
G_{\underline{\alpha}, \underline{\beta}, \underline{\mu}}^{r,s}(\underline{x}, \underline{y}, \underline{z}) \equiv \frac{\partial}{\partial \eta^-_{\underline{\beta}}(\underline{y})} \frac{\partial}{\partial \eta^+_{\underline{\alpha}}(\underline{x})} \frac{\partial}{\partial J_{\underline{\mu}}(\underline{z})} W[J, \eta] \eval_{J, \eta = 0},
\end{equation}
the \emph{two-point function} and the \emph{vertex function} are given by
\begin{equation}
S_{\alpha \beta}(x, y) \equiv G_{\alpha, \beta}^{1, 0}(x, y), \qquad (\mathfrak{I}^\mu)_{\alpha \beta}(x, y, z) \equiv G_{\alpha, \beta, \mu}^{2, 1}(x, y, z).
\end{equation}
Since the system is translationally invariant, we define
\begin{equation}\label{ct}
\hat{S}(k) \equiv \int_{\mathcal{M}_L} \dd[4]{x} e^{-i k \cdot x} \, S(x, 0), \qquad
\hat{\mathfrak{I}}^\mu(p', p) \equiv \int_{\mathcal{M}_L^2} \dd[4]{x} \dd[4]{y} e^{-i p' \cdot x + i p \cdot y} \, \mathfrak{I}^\mu(x, y, 0).
\end{equation}
Finally, the \emph{amputated vertex function} in Fourier space is 
\begin{equation}
\hat{\Gamma}^\mu(p', p) \equiv [\hat{S}(p')]^{-1} \cdot \hat{\mathfrak{I}}^\mu(p', p) \cdot [\hat{S}(p)]^{-1}(p).
\end{equation}
After integrating out the boson fields, the integral~\eqref{eq:generating_functional} can be written as
\begin{equation}
\label{eq:generating_functional_fermionic}
W[J, \eta] = \log \int \dgauss{\sing{g}{\le N}}{\psi} \, \exp( \sing{V}{N}[\psi, J, \eta] )
\end{equation}
with
\begin{multline}
\label{eq:ultraviolet_fermionic_potential}
\sing{V}{N}[\psi, J, \eta] \equiv - \frac{\lambda^2}{2} \int_{\mathcal{M}_L^2} \dd[4]{x} \dd[4]{y} \, (\psi^+ \Upsilon^\mu \psi^-)(x) \, v(x - y) \, (\psi^+ \Upsilon^{\mu} \psi^-)(y) \, + \\
+ \sum_{s = \pm} \int_{\mathcal{M}_L} \dd[4]{x} (Z^{J, s}_N \psi^+_s \sigma^\mu_s J_\mu \psi^-_s + \eta^+_s \psi^-_{-s} + \psi^+_s \eta^-_{-s})(x)
\end{multline}
and $\lambda \equiv (\lambda_+ + \lambda_-)/2, \Upsilon^\mu \equiv \gamma^\mu(1 - \kappa \gamma^5),\kappa \equiv (\lambda_+ - \lambda_-)/(2\lambda_+ + 2\lambda_-)$.
Whenever there is no risk of ambiguity, we adopt the condensed notation $(J, \eta) \equiv \src$, it being understood that $\src = (\src_1, \dots, \src_{12}) = (J_0, \dots, \eta^-_4)$.

If possible, the bare parameters $Z^s_N, Z^{J, s}_N, m^s_N$ have to be chosen so that the dominant parts of the correlation functions describe a fermion whose charge and field strength renormalization are equal to $1$ and whose mass is equal to $m$. In particular, the two-point function and the vertex function must satisfy, in the limit $L\to\infty$, the \emph{renormalization conditions}
\begin{align} 
\label{eq:conditions_correlations}
\begin{split}
\hat{S}(k) 
& =\frac{1}{i \slashed{k}+m} \cdot \left( 1 + \BigO \left[ \left( \frac{\max \lbrace m, \abs*{k}\rbrace}{\cutoff} \right)^\theta \right] \right), \\ 
\hat{\Gamma}^\mu(p', p) 
& = \gamma^\mu + \BigO \left[ \left( \frac{\max \lbrace m, \abs{p}, \abs*{p'} \rbrace}{\cutoff} \right)^\theta \right]
\end{split}
\end{align} 
for some $\theta>0$. 

\subsection{The regularized gyromagnetic factor}
\label{subsec:model_definition1}

The regularized gyromagnetic factor is given by
\begin{equation}
\label{eq:g-2:def}
\rgyr \equiv \sum_{\ell = 0}^K \frac{(im)^\ell}{\ell!} \mathcal{A}^{(\ell)}(0)
\end{equation}
with
\begin{equation}
\label{eq:noncovariant_gfunction}
\mathcal{A}(z) \equiv \frac{z \epsilon_{abc} \tr[\gamma^5 (1 + \gamma^0) (\partial_{p'} - \partial_p)^a \hat{\Gamma}^b(p_z, p_z) \, \gamma^c] + 2 \tr[\gamma_a \hat{\Gamma}^a(p_z, p_z)]}{6 \tr[(1 + \gamma^0) \hat{\Gamma}^0(p_z, p_z)]} - 1
\end{equation}
where $a, b, c \in \lbrace 1, 2, 3 \rbrace$ and $p_z \equiv (z, 0, 0, 0)$. If the formal $L, N, K \to\infty$ limit is taken, this definition formally coincides with the standard one (see e.\ g.~\cite[Equation (6.37)]{Peskin_1995}). More details about this definition are provided in Appendix~\ref{app:g-2_definition}. 
Our main result is the following
\begin{theorem}
\label{th:main_theorem}
Suppose that $K=4$ and $M > 10m$. If
\begin{equation}
\lambda^2 \le  (C_0)^{-1} (M^2/\cutoff^2) \log^{-2}(M/m) \log^{-2}(\cutoff/M),
\end{equation}
the limit $L\to
\infty$ of the correlations (\ref{ct}) exists, it verifies (\ref{eq:conditions_correlations})
for a suitable choice of 
$Z^s_N, Z^{J, s}_N, m^s_N$ and is such that the regularized anomalous gyromagnetic factor~\eqref{eq:g-2:def} satisfies
\begin{equation}
\label{eq:final_result_remainder}
\rgyr = 
\bar{\gyr}_{\textsc{z},1}(1+R_\lambda), \qquad 
\abs*{R_\lambda} \le  C_1 \frac{m^2}{M^2} +C_2 \frac{M^2}{\cutoff^2} + C_3 \frac{\lambda^2\cutoff^2}{M^2} \cdot \log^4 \left(\frac{M}{m}\right) \log^4 \left( \frac{\cutoff}{M} \right),
\end{equation}
where $\bar{\gyr}_{\textsc{z},1}$ is given by (\ref{eq:JW_result}) and 
$C_0, C_1, C_2, C_3$ are positive constants that do not depend on $m, M, \cutoff, \lambda$.
\end{theorem}
The result rigorously shows that the anomalous gyromagnetic factor in the effective regularized theory coincides with the truncation of the perturbative expansion  with cutoff removed (whose sum does not exist), up to regularizaton-dependent corrections which are small in a suitable region of the parameter space. Therefore, this result allows to estimate the error due to truncation in the formal expansion, even if it is non convergent. The physical regime corresponds to $m/M$ and $\lambda$ small with 
$m/M \ll \lambda$ (in reality, $\lambda\sim 10^{-2}$ and $m/M\sim 10^{-6}$ in the case of electrons). Hence, the first term appearing in the remainder $R_\lambda$ is surely small; further corrections due to higher values of $K$ would be dominated by this one. The second term comes from the difference between the lowest order contribution of the regularized theory and the one without cutoff, and it is small if $\Lambda$ is chosen to be much higher than the boson mass $M$. 
The third term takes into account higher orders in $\lambda$ and it is small provided that $\lambda$ is sufficiently small, that is $\lambda \ll \BigO( (M/\Lambda)\log^{-1}(M/m))$. A crucial aspect of our result is the fact that this term depends logarithmically on $M/m$: a power law dependence would imply a completely unrealistic smallness condition such as $\lambda \ll \BigO(m/\cutoff)$. This is achieved by proving suitable cancellations in the expansion, implying that all the higher orders contributions are $\BigO(m^2/M^2)$. The choice $K=4$ is enough to get upper and lower bounds in terms of the Jackiw-Weimberg formula~\eqref{eq:JW_result}.

\subsection{Sketch of the proof}

After integrating out the bosons, the amputated vertex function appearing in~\eqref{eq:noncovariant_gfunction} is expressed by Grassmann integrals which can be analyzed by means of rigorous Renormalization Group methods. The amputated vertex function is written as a series in the coupling $\lambda$ and the running coupling constants, that is
\begin{equation}
\label{eq:vertex_splitting}
\hat{\Gamma}^\mu(p', p)=
\hat{\Gamma}_a^\mu(p', p)+\hat{\Gamma}^\mu_b(p', p),
\end{equation}
where the dominant part, $\hat{\Gamma}_a^\mu(p', p)$, only depends from the running coupling constants; $\hat{\Gamma}^\mu_b(p', p)$ is subdominant and can be expressed as a series in $\lambda$, namely
\begin{equation}
\label{eq:subdominant_vertex_expansion}
\hat{\Gamma}^\mu_b(p', p)=\sum_{n=1}^{+\infty} \lambda^{2 n} \hat{\Gamma}_n^\mu(p', p; \lambda).
\end{equation}
Note that the above is not a \emph{power} series in $\lambda$, as the running coupling constants are also functions of $\lambda$ (that is, $\hat{\Gamma}_n^\mu(p', p; \lambda)$ depend from $\lambda$).
It turns out that, by fixing the bare constants so that the renormalization conditions~\eqref{eq:conditions_correlations} are verified, the contribution to $\rgyr$ coming from $\hat{\Gamma}_a^\mu$ is vanishing. The contribution coming from $\lambda^2 \hat{\Gamma}_1^\mu(p', p; \lambda)$ is expressed by several renormalized graphs whose sum reconstructs the Jackiw-Weinberg result, up to corrections which are included in the first
and second term of~\eqref{eq:final_result_remainder}. In particular, the dominant part of $\rgyr$ turns out be of order $\lambda^2 m^2/M^2$. The higher orders ($\lambda^{2 n} \hat{\Gamma}_n^\mu(p', p; \lambda)$, with $n\ge 2$) can also be written as sums
over renormalized Feynman graphs. In order to prove convergence, however, it is convenient
\emph{not} to bound each graph separately, but to control their sum using suitable determinant bounds which exploit cancellations caused by the anticommutativity of fermionic variables. A dimensional bound for $\hat{\Gamma}_n^\mu(p', p; \lambda)$ would be
$\BigO(C^n (\Lambda^2/M^2)^n)$; this estimate is enough for convergence, but not for ensuring the lowest order dominance, as it does not contain the small $m^2/M^2$ factor. An improvement of the bound for the contributions coming from $\hat{\Gamma}^\mu_b(p', p)$ follows from certain cancellations due to exact or approximate symmetries; the last ones are indeed violated by the presence of the fermion mass, which is however small. As a consequence, the renormalized scaling dimension is smaller than or equal to $-2$ and this produces the desidered $m^2/M^2$ gain up to some harmless logarithmic corrections.

\subsection{Open problems and remarks}

The regularized neutral Electroweak model considered in this paper shares some similarities with Condensed Matter models that admit an emerging Quantum Field Theory description, like graphene, Hall systems or Weyl semimetals. In such systems, the lattice of ions among which the conduction elctrons are hopping provides a physical realization of the ultraviolet cutoff. A basic property of these Condensed Matter models is the universality of certain transport coefficients, that is, the independence of such quantities from the microscopic details. In particular, universality is present in the optical conductivity of graphene ($\sigma=(e^2/h) \cdot 1/(2\pi)$) or in the Hall conductance ($\sigma_n = n e^2/h$, where $n \in \N$). These coefficients are strictly related with quantum anomalies of the emerging QFT description.

The transport coefficients, togheter the anomalous gyromagnetic factor of the electron, are measured with very high precision and their theoretical values are used to determine the fine structure constant. Note that, in the case of the transport coefficients, the theoretical prediction coincides with the lowest order computation in perturbation theory, while the anomalous gyromagnetic factor is a series with non vanishing coefficients at any order. The Renormalization Group approach used in this paper provides an explanation of this fact: the crucial point is that transport coefficients (as most of the correlation functions that are actually measured in Condensed Matter experiments) are usually associated to marginal or relevant terms, while the anomalous gyromagnetic factor is associated to an irrelevant one.

The universality of the transport coefficients in presence of a weak short range many-body interaction has been proven in $3+1$ dimensions~\cite{Mastropietro_2007, Giuliani_Mastropietro_Porta_2019, Mastropietro_2024_Vanishing} or in $2+1$ dimensions~\cite{Giuliani_Mastropietro_Porta_2019, Giuliani_Mastropietro_Porta_2017} in the case of graphene. The transport coefficients are decomposed as sums of two terms as in~\eqref{eq:vertex_splitting}. In this case, however, only the first term (involving marginal couplings) gives a nonzero contribution; the second one, involving irrelevant terms, is exactly zero. No need of improvment of estimates with respect to power counting is necessary in this case.
From this fact follows the universality of transport coefficients or the related phenomenon of the non-renormalization of quantum anomalies in presence of a finite cutoff. 
The anomalous gyromagnetic factor, on the contrary, only gets contributions from higher orders in the renormalized expansion, which involve irrelevant interactons. 

There are several natural extensions of the present work, which could be possibly treated with the methods introduced in this paper and their generalizations. 
A natural question is to improve the smallness condition on the coupling,
$\lambda \ll \BigO( (M/\Lambda)\log^{-1}(M/m))$.
In the case of renormalizable models, one could hope to replace such condition with $\lambda \ll \BigO(\log^{-1}(\Lambda/M)\log^{-1}(M/m))$, which would allow to reach much higher cutoffs. The model considered here is perturbatively non renormalizable due to the presence of anomalies, so this improvment is not expected here. However, after including the quark contributions with realistic values of the charges, the anomalies are expected to cancel and perturbative renormalizablity is recovered, so one can hope to get such a result in this generalized framework.
As a starting point, one could consider a vector model like massive QED$_4$, and even in this case the construction up to exponentially high cutoff is unknown.
As a preliminary step, one would like to repeat the analysis presented here with lattice cutoff (which is necessary for the anomaly cancellation), where the symmetries at the basis of the above mentioned cancellations are different.
The issue of QED$_4$ contributions to the gyromagnetic factor with vanishing photon mass is also a natural problem to be considered, see e.\ g.~\cite{Kaplan_2021}.

The paper is structured in the following way. In Section~\ref{sec:RG_analysis} we developed a Renormalization Group analysis and a tree expansion for the correlations, incorporating exact and emeging symmetries. In Section~\ref{sec:renormalized_expansion_bounds} we prove bounds for the coefficient of the expansion for the vertex function, proving its convergence using determinant bounds and cluster expansion formulas for truncared expectations. Finally, in Section~\ref{sec:g-2_calculation} we prove the decomposition~\eqref{eq:subdominant_vertex_expansion} and the relative bounds with improvements due to cancellations, leading to the proof of the main result.

\section{Renormalization Group analysis}
\label{sec:RG_analysis}

\subsection{Multiscale decomposition}
\label{subsec:rg_step_and_localization}
The Grassmann integral occurring in~\eqref{eq:generating_functional_fermionic}
is perfomed iteratively using the Gaussian convolution identity; namely, if $g = g_1 + g_2$, we have
\begin{equation}
\label{eq:measure_additivity}
\int \dgauss{g_1}{\psi} \, F(\psi)=\int \dgauss{g_1}{\psi_1} \int \dgauss{g_1}{\psi_2} \, F(\psi_1+\psi_2)
\end{equation}
for every functional $F$. This allows to represent
the generating functional~\eqref{eq:generating_functional_fermionic} as
\begin{equation}
\label{eq:generating_functional_on_scale_h}
W[\src] = \log \int \dgauss{\sing{g}{\le h}}{\psi} \exp(\mathscr{V}^{h}[\psi, \src]),
\end{equation}
where $h \in \lbrace h^\star, \dots, N \rbrace$, $h^\star = \lfloor \log_2 m \rfloor$, $\chi_h(k) \equiv \chi_0(2^{-h} k)$ and
\begin{equation}
\label{eq:renormalized_propagator}
\sing{\hat{g}}{\le h}(k) \equiv 
\chi_h(k) \cdot
\begin{pmatrix}
m_h^+						&	i Z_h^+ \sigma^\mu_+ k_\mu \\
i Z_h^- \sigma^\mu_- k_\mu	&	m_h^- \\
\end{pmatrix}^{\!-1}
\equiv \frac{\chi_h(k)}{i k_\mu \gamma^\mu_h + \mathsf{m}_h}.
\end{equation}
The \emph{effective potential} on scale $h$ takes the form
\begin{equation}
\label{eq:veff_kernels}
\sing{\mathscr{V}}{h}[\psi, \src] = \sum_{A, B} \int \dd{\underline{x}} \dd{\underline{y}} \, \sing{W}{h}(A, \underline{x}; B, \underline{y}) \Psi(A, \underline{x}) \Src(B, \underline{y}),
\end{equation}
where $\sing{W}{h}(A, \underline{x}; B, \underline{y})$ are smooth, periodic, translationally invariant kernels and the \emph{field monomials} $\Psi(A, \underline{x}) \Src(B, \underline{y})$ contain at most three derivatives on each field $\psi^\pm$. Here, we adopt the same notation introduced in~\cite[Equation (4.2)]{Giuliani_Mastropietro_Rychkov_2021}: given a pair of \emph{index sets} $A^\pm = \lbrace A_1^\pm, \dots, A_s^\pm \rbrace$ with $A^\pm_\ell = (\alpha_\ell^\pm, \underline{\delta}_\ell^\pm) \subseteq \lbrace 1, \dots, 4 \rbrace \times \N^4$, we let
\begin{equation}
\label{eq:field_monomials}
\Psi(A, \underline{x}) \equiv \prod_{a = 1}^{\abs*{A^+}} \partial_{\underline{\delta}_a^+} \psi^+_{\alpha_a^+}(x_a^+) \prod_{b = 1}^{\abs*{A^-}} \partial_{\underline{\delta}_b^-} \psi^-_{\alpha_b^-}(x_b^-),
\end{equation}
where $A \equiv A^+ \sqcup A^-$, $\underline{x} = (x_1^+, \dots, x_s^-)$ and $\partial_{\underline{\delta}_\ell^\pm} = \prod_{r = 1}^4 (\partial_r)^{\delta_\ell^{\pm, r}}$ with $\abs*{\underline{\delta}_\ell^\pm} \equiv \sum_{r = 1}^4 \abs*{\delta_\ell^{\pm, r}} \le 3$. In the same way, we let
\begin{equation}
\label{eq:source_field_monomials}
\Src(B, \underline{y}) \equiv \prod_{a = 1}^{\abs*{B_J}} J^{\mu_a}(x_a) \prod_{b = 1}^{\abs*{B_{\eta}^+}} \eta^+_{\alpha^+_b}(y_b^+) \prod_{c = 1}^{\abs*{B_{\eta}^-}} \eta^-_{\alpha^-_c}(y_c^-),
\end{equation}
where the index set $B = B_J \sqcup B^+_\eta \sqcup B^-_\eta$ is constructed as before. 

The representation~\eqref{eq:generating_functional_on_scale_h} coincides with~\eqref{eq:generating_functional_fermionic} when $h = N$, so it will be fully determined once we show how to pass from an integer $h$ to $h-1$.
\begin{rmk} 
As it will be clear from Theorem~\ref{th:renormalization_bound}, to each kernel $W(A, \underline{x}; B, \underline{y})$ is associated a dimension $D=4-3\ell/2-n_J$, where $\ell, n_J$ respectively denote total number of $\psi$ and $\eta$ fields and the number of $J$ fields. One has to introduce the renormalization operator in order to control the behavior of the kernels with positive or vanishing dimension. 
\end{rmk}

\subsection{Localization}
\label{subsec:localization}

We write~\eqref{eq:generating_functional_on_scale_h} by decomposing $\sing{\mathscr{V}}{h}[\psi, \omega]$ as
\begin{equation}
\label{eq:generating_functional_RL}
W[\src] = \log \int \dgauss{\sing{g}{\le h}}{\psi}
\exp(\loc \sing{\mathscr{V}}{h}[\psi, \src] + \ren \sing{\mathscr{V}}{h}[\psi, \src]),
\end{equation}
where $\ren=\id-\loc$. Depending on the structure of the kernel $W(A, \underline{x}; B, \underline{y})$, the \emph{localization operator} $\loc$ is defined as 
\begin{equation}
\label{eq:localization_definition}
\loc \left[ W(A, \underline{x}; B, \underline{y}) \, \Psi(A, \underline{x})\Src(B, \underline{y}) \right] \equiv
\begin{cases}
(\loc_{0, x} + \loc_{1, x} + \bar{\loc}_{2, x})\left[ W_{\alpha \beta}(x, y) \psi^+_\alpha(x) \psi^-_\beta(y) \right] \\[5pt]
(\loc_{0, z} + \bar{\loc}_{1, z}) \left[ W_{\alpha \beta \mu}(x, y, z) \psi^+_\alpha(x) \psi^-_\beta(y) J_\mu(z) \right] \\[5pt]
(\loc_{0, x} + \bar{\loc}_{1, x}) \left[ W_{\alpha \beta \underline{\delta}}(x, y) \partial_{\underline{\delta}_1} \psi^+_\alpha(x) \partial_{\underline{\delta}_2} \psi^-_\beta(y) \eval_{\underline{\delta} = 1} \right] \\[5pt]
\bar{\loc}_{0, x} \left[ W_{\alpha \beta \underline{\delta}}(x, y) \partial_{\underline{\delta}_1} \psi^+_\alpha(x) \partial_{\underline{\delta}_2} \psi^-_\beta(y) \eval_{\underline{\delta} = 2} \right] \\[5pt]
\bar{\loc}_{0, z} \left[ W_{\alpha \beta \mu \underline{\delta}}(x, y, z) \partial_{\underline{\delta}_1} \psi^+_\alpha(x) \partial_{\underline{\delta}_2} \psi^-_\beta(y) J_\mu(z) \eval_{\underline{\delta} = 1} \right] \\[5pt]
0	\qquad	\text{In any other case}
\end{cases}
\end{equation}
where $\abs*{\underline{\delta}} \equiv \abs*{\underline{\delta}_1} + \abs*{\underline{\delta}_2}$ and
\begin{itemize}
\item If $x_0$ is a spacetime point among $\underline{x}, \underline{y}$, the operator $\loc_{\ell, x_0}$ acts as
\begin{equation}
\label{eq:taylor_localization_definition}
\loc_{\ell, x_0} \left[ W(A, \underline{x}; B, \underline{y}) \, \Psi(A, \underline{x})\Src(B, \underline{y}) \right] \equiv  W(A, \underline{x}; B, \underline{y}) \, \frac{1}{\ell!} \frac{\dd^\ell}{\dd t^\ell}\Psi(A, \underline{x}_t)\Src(B, \underline{y}_t) \eval_{t = 0},
\end{equation}
where the spacetime points $\underline{x}_t, \underline{y}_t$ are obtained by replacing every point $z$ occurring among $\underline{x}, \underline{y}$ with the interpolating point $I_t(x_0, z) \equiv x_0 + t(z - x_0)_\torus$.
\item $\bar{\loc}_{\ell, x_0}$ acts on a kernel with spinor indices $\alpha, \beta$ as
\begin{equation}
\bar{\loc}_{\ell, x_0}= 
\begin{cases}
\loc_{\ell, x_0}			& \textup{if } \alpha, \beta \textup{ have similar chiralities} \\
\mathcal{P}\loc_{\ell, x_0}	& \textup{if } \alpha, \beta \textup{ have opposite chiralities}
\end{cases}
\end{equation}
with $\mathcal{P}(\, \cdot \,) \equiv (\, \cdot \,) \vert_{m_N = 0}$.
\end{itemize}
The \emph{renormalization operator} $\ren = \id - \loc$ can be decomposed as $\ren \equiv \ren^0 + \ren^1$, where $\ren^0 = \id - \loc_{0, x_0} - \loc_{1, x_0} - \cdots - \loc_{R - 1, x_0}$ and $\ren^1$ is either equal to $0$ or to $(\id - \mathcal{P})\loc_{R - 1, x_0}$ (both $R$ and $x_0$ depend on the kernel on which $\ren$ acts). The $\ren^0$ operator will be called \emph{Taylor renormalization}.

An important property of the localization operation is that
\begin{lm}
\label{lm:localization_request}
For every $h \in \lbrace h^\star, \dots, N \rbrace$, $\loc \sing{\mathscr{V}}{h}[\psi, \src]$ takes the form
\begin{equation}
\label{eq:localization_request}
\loc \sing{\mathscr{V}}{h}[\psi, \src] =  \sum_{s = \pm} \int_{\mathcal{M}_L} \dd[4]{x} (-\beta^{m, s}_h \, \psi^+_s \psi^-_{-s} - \beta^s_h \, \psi^+_s \sigma_s^\mu \partial_\mu \psi^-_s + Z^{J, s}_h \psi^+_s \sigma_s^\mu J_\mu \psi^-_s)(x),
\end{equation}
with $\beta^{m, s}_h \vert_{m_N = 0} = m^s_h \vert_{m_N = 0} = 0$.
\end{lm}
\begin{proof}
See Appendix~\ref{app:localization_request_proof}.
\end{proof} 
Note that several terms which should appear in the above expression based on definition~\eqref{eq:localization_definition} are indeed absent due to symmetry properties. In particular, the $\bar{\loc}_{\ell, x_0}$ operators displayed in~\eqref{eq:localization_definition} \emph{have no influence on $\loc \sing{\mathscr{V}}{h}[\psi, \src]$}.

The action of the renormalization operator on the kernel representation~\eqref{eq:veff_kernels} can be understood as follows. Consider a smooth periodic function $F \colon \mathcal{M}_L^2 \to \C$ whose Fourier coefficients take the form $\hat{F}(k_0, k) = f_1(k_0) f_2(k)$, where $f_2$ vanishes as long as the norm of its argument is greater than $2^N$. If $\ren^0_{R, x_0} = \id - \loc_{0, x_0} - \cdots - \loc_{R-1, x_0}$ and $W(x_0, x) = W(0, x - x_0)$ is a smooth, periodic, translationally invariant kernel, then
\begin{equation}
\label{eq:ren_kernel_eq}
\int \dd{x_0} \dd{x} \ren^0_{R, x_0} [W(x_0, x) F(x_0, x)] = \int \dd{x_0} \dd{x} (\kren^0_{R, x_0} W)^{\underline{\alpha}}(x_0, x) \frac{\partial F}{\partial x^{\underline{\alpha}}}(x_0, x)
\end{equation}
with
\begin{equation}
\label{eq:ren_taylor_kernel}
(\kren^0_{R, x_0} W)^{\underline{\alpha}}(x_0, x) \equiv \int_0^1 \dd{t} \frac{(1-t)^{R-1}}{(R-1)!} \int \dd{y} W(0,y) \, y^{\underline{\alpha}}_\torus \, 
\delta_N(x - x_0 - t y_\torus).
\end{equation}
Here, $\underline{\alpha} = (\alpha_1, \dots, \alpha_R) \in \lbrace 0, \dots, 3 \rbrace^R$ is a multi-index: given any $x \in \R^4$, the notations $x^{\underline{\alpha}}$ and $\partial/\partial x^{\underline{\alpha}}$ respectively stand for $\prod_{\ell = 1}^R x^{\alpha_\ell}$ and $\prod_{\ell = 1}^R \partial/\partial x^{\alpha_\ell}$ and sum over $\alpha_1, \dots, \alpha_R$ is understood in both~\eqref{eq:ren_kernel_eq} and~\eqref{eq:ren_taylor_kernel}. Finally, the ``regularized delta'' $\delta_N(x)$ is a periodic function on $\mathcal{M}_L$ whose Fourier coefficients are given by $\hat{\delta}_N(k) \equiv \chi_{N+1}(k)$.

Formula~\eqref{eq:ren_taylor_kernel} is an immediate consequence of the interpolation identity
\begin{multline}
\label{eq:lm:interpolating_formula}
\int \dd{x_0} \dd{x} \ren^0_{R, x_0} [W(x_0, x) F(x_0,x)] = \\ 
= \int \dd{x_0} \dd{x} W(x_0, x) \sum_{\underline{\alpha}} \int_0^1 \dd{t} \frac{(1-t)^{R-1}}{(R-1)!} (x - x_0)^{\underline{\alpha}}_\torus \, \partial_{\underline{\alpha}} F(x_0, I_t(x_0, x)).
\end{multline}
Note that, according to our initial assumption, $\hat{F} \hat{\delta}_N = \hat{F}$. In position space, this equality becomes
\begin{equation}
F(x_0, z) = \int \dd{y} F(x_0, y) \delta_N(y - z) \qquad \forall z \in \mathcal{M}_L,
\end{equation}
so the right hand side of~\eqref{eq:lm:interpolating_formula} can be rewritten as
\begin{equation}
\label{eq:delta_insertion}
\int \dd{x_0} \dd{x} \dd{y} \int_0^1 \dd{t} \delta_N(y - x_0 - t(x - x_0)_\torus) \,
W(0, x-x_0) \frac{(1-t)^{R-1}}{(R-1)!} (x - x_0)^{\underline{\alpha}}_\torus \, \partial_{\underline{\alpha}} F(x_0, y),
\end{equation}
where we used the fact that $W(x_0, x) = W(0, x - x_0)$. After performing the change of variables $x \mapsto x + x_0$ together with the relabelling $x \leftrightarrow y$, the integral~\eqref{eq:delta_insertion} becomes
\begin{equation}
\int \dd{x}_0 \dd{x} \dd{y} \int_0^1 \dd{t} \delta_N(x - x_0 - t y_\torus) \,
W(0, y) \frac{(1-t)^{R-1}}{(R-1)!} y^{\underline{\alpha}}_\torus \, \partial_{\underline{\alpha}} F(x_0, x)
\end{equation}
and by comparing this with~\eqref{eq:ren_kernel_eq} we deduce that $(\kren_{R, x_0}^0 W)^{\underline{\alpha}}$ is indeed given by~\eqref{eq:ren_taylor_kernel}.

The above construction can be applied to the $\loc, \ren^1$ operators as well and it can be generalized in order to deal with kernels having two external fermionic legs and one $J$ source field. By adopting the synthetic notation $\underline{x}$ for either $x_1$ or $(x_1, x_2)$, we have
\begin{align}
\begin{split}
\int \dd{x_0} \dd{\underline{x}} \loc_{\ell, x_0} [W(x_0, \underline{x}) F(x_0, \underline{x})] = \int \dd{x_0} \dd{\underline{x}} (\kloc_{\ell, x_0} W)^{\underline{\alpha}}(x_0, \underline{x}) \partial_{\underline{\alpha}} F(x_0, \underline{x})
\end{split} \\
\begin{split}
\label{eq:renormalization_kernel_action}
\int \dd{x_0} \dd{\underline{x}} \ren^a_{R, x_0} [W(x_0, \underline{x}) F(x_0, \underline{x})] = \int \dd{x_0} \dd{\underline{x}} (\kren^a_{R, x_0} W)^{\underline{\alpha}}(x_0, \underline{x})  \partial_{\underline{\alpha}} F(x_0, \underline{x})
\end{split}
\end{align}
with $a = 0, 1$ and
\begin{subequations}
\begin{align}
\begin{split}
\label{eq:loc_kernel_general}
(\kloc_{\ell, x_0} W)^{\underline{\alpha}}(x_0, \underline{x}) & \equiv \frac{1}{\ell!} \binom{\ell}{\underline{\alpha}} (\underline{x} - \underline{x}_0)^{\underline{\alpha}}_\torus \, \delta_N^{(n)}(\underline{x} - \underline{x}_0) \int \dd{\underline{y}} W(0, \underline{y}),
\end{split}
\\
\begin{split}
\label{eq:ren_taylor_kernel_general}
(\kren_{R, x_0}^0 W)^{\underline{\alpha}}(x_0, \underline{x}) & 
\equiv \binom{R}{\underline{\alpha}} \int_0^1 \dd{t} \frac{(1-t)^{R-1}}{(R-1)!} \int \dd{\underline{y}} W(0, \underline{y}) \, \underline{y}^{\underline{\alpha}}_\torus \, \delta^{(n)}_N(\underline{x} - \underline{x}_0 - t \underline{y}_\torus),
\end{split}
\\
\begin{split}
\label{eq:ren_zero_mass_kernel_general}
(\kren_{R, x_0}^1 W)^{\underline{\alpha}}(x_0, \underline{x}) & \equiv -(\mathscr{S}\kloc_{R-1, x_0} W)^{\underline{\alpha}}(x_0, \underline{x}),
\end{split}
\end{align}
\end{subequations}
where $\mathscr{S}(\, \cdot \,) \equiv (\, \cdot \,) - (\, \cdot \,) \vert_{m_N = 0}$, $\underline{x}_0 \equiv (x_0, \dots, x_0)$ and $n=1$ or $n=2$ depending on whether $W$ carries a $J$ source field or not. This time, $\underline{\alpha}$ is a multi-index of the form $\underline{\alpha} = (\underline{\alpha}_1, \dots, \underline{\alpha}_n)$ with $\underline{\alpha}_\ell \in \lbrace 0, \dots, 3 \rbrace^{\abs*{\underline{\alpha}_\ell}}$. The constraints $\sum_{\ell = 1}^n \abs*{\underline{\alpha}_\ell} = \ell, R, R-1$ are respectively required in~\eqref{eq:loc_kernel_general},~\eqref{eq:ren_taylor_kernel_general} and~\eqref{eq:ren_zero_mass_kernel_general}, whereas
\begin{equation}
\binom{R}{\underline{\alpha}} \equiv \frac{R!}{\abs*{\alpha_1}! \cdots \abs*{\alpha_n}!}, \qquad \underline{x}^{\underline{\alpha}} \equiv \prod_{\ell=1}^n \prod_{s=1}^{\abs*{\alpha_\ell}} x_\ell^{(\alpha_\ell)_s}, \qquad \partial_{\underline{\alpha}} \equiv \prod_{\ell=1}^n \prod_{s=1}^{\abs*{\alpha_\ell}} \frac{\partial}{\partial x_\ell^{(\alpha_\ell)_s}}.
\end{equation}

\begin{rmk}
As it is clear from~\eqref{eq:renormalization_kernel_action} and~\eqref{eq:ren_taylor_kernel_general}, the $\ren^0$ operator produces a gain $R$ in the scaling dimension of the kernel on which it acts (namely, if $W$ has scaling dimension $D$, then $\kren^0 W$ has scaling dimension $D - R$). Since the kernels acting on monomials with two external $\psi$ fields have scaling dimension $1$, a $\ren$ operation which extracts a maximum gain $R = 2$ would be sufficient to make $D - R$ always negative; here we choose instead to get a maximum gain $R = 3$ in order to obtain a suitable bound for the anomalous gyromagnetic factor. The extra subtraction necessary to get $R = 3$ could, in principle, require to add some extra running coupling constants. However, this is not the case, as it is apparent from Lemma~\ref{lm:localization_request}. This is true only if we define, in addition to the renormalization $\ren^0$, also the renormalization $\ren^1$. The property~\eqref{eq:localization_request} follows from the exact Lorentz invariance and the approximate chiral invariance, broken by the mass $m$ which is much smaller than $M$.
\end{rmk}

\subsection{Single scale integration}
\label{subsec:single_scale_integration}

Thanks to the localization operator defined by~\eqref{eq:localization_definition}, we can finally show how to rewrite~\eqref{eq:generating_functional_on_scale_h} by lowering the scale from $h$ to $h-1$. At first, we reabsorb the quadratic terms contained into $\loc \sing{\mathscr{V}}{h}[\psi, \src]$ inside the Gaussian measure, thus getting
\begin{equation}
\label{eq:generating_functional_RL_explicit}
W[\src] = \log \int \dgauss{\sing{(g')}{\le h}}{\psi} \, \exp( \sing{V}{h}[\psi, \src]),
\end{equation}
where the functional $\sing{V}{h}$ is defined as
\begin{equation}
\loc \sing{V}{h}[\psi, \src] \equiv \sum_{s = \pm} \int_{\mathcal{M}_L} \dd[4]{x} (Z^{J, s}_h \psi^+_s \sigma_s^\mu J_\mu \psi^-_s)(x), \qquad \ren \sing{V}{h}[\psi, \src] \equiv \ren \sing{\mathscr{V}}{h}[\psi, \src]
\end{equation}
and the \emph{renormalized propagator} is
\begin{equation}
\label{eq:propagator_splitting}
\sing{(\hat{g}')}{\le h}(k) \equiv \frac{\chi_{h-1}(k)}{i k_\mu \gamma^\mu_{h-1} + \mathsf{m}_{h-1}} + \frac{f_{h}(k)}{i k_\mu \tilde{\gamma}^\mu_{h}(k) + \tilde{\mathsf{m}}_{h}(k)} \equiv \sing{\hat{g}}{\le h-1}(k) + \sing{\hat{g}}{h}(k)
\end{equation}
with $f_h(k) \equiv \chi_h(k) - \chi_{h-1}(k)$ and $Z^s_{h-1} \equiv Z^s_h + \beta^s_h, \,\, m^s_{h-1} \equiv m^s_h + \beta^{m,s}_h$. The matrices $\tilde{\gamma}^\mu_h(k), \tilde{\mathsf{m}}_h(k)$ are defined as $\gamma^\mu_h, \mathsf{m}_h$ (see~\eqref{eq:renormalized_propagator}) with $Z^s_h + \beta^s_h \chi_h(k), \, m^s_h + \beta^{m, s}_h \chi_h(k)$ in place of $Z^s_h, m^s_h$. The renormalized potential takes the form
\begin{equation}
\label{eq:veff_kernelsr}
\ren \sing{V}{h}[\psi, \src] = \sum_{A, B} \int \dd{\underline{x}} \dd{\underline{y}} \, (\kren \sing{W}{h})(A, \underline{x}; B, \underline{y}) \Psi(A, \underline{x}) \Src(B, \underline{y}),
\end{equation}
where $\kren$ is either equal to $\id$ or to a suitable combination of $\kren^0_{R, x_0}$ and $\kren^1_{R, x_0}$, where both $R$ and $x_0$ depend on the structure of the kernel on which they act in accordance with definition~\eqref{eq:localization_definition}. We stress that there are at most three derivatives for each fermion field; indeed, the $\ren$ operator respectively produces three and two extra derivatives when acting on $W \psi^+ \psi^-$ and $W \psi^+ \partial \psi^-$ and one extra derivative when acting on $W \psi^+ \partial^2 \psi^-$. Similarly, $\ren$ respectively produces two extra derivatives and one extra derivative when acting on $J \psi^+\psi^-$ and $J \psi^+ \partial \psi^-$. In any other case, $\ren$ behaves as the identity operator.

The running coupling constants $Z^{J, s}_{h-1}, Z^s_{h-1}, m^s_{h-1}$ are entirely determined by $Z^{J, s}_h$, $Z^s_h$, $m^s_h$, $\lambda$ by means of the relations
\begin{equation}
Z^{J, s}_{h-1} = Z^{J, s}_h + \beta^{J, s}_h, \qquad Z^s_{h-1} = Z^s_h + \beta^s_h, \qquad m^s_{h-1} = m^s_h + \beta^{m, s}_h.
\end{equation}
The quantities $\beta^{J, s}_h, \beta^s_h, \beta^{m, s}_h$ are known as \emph{beta functions} on scale $h$ and they depend on $\lambda$ and on the set $\lbrace Z^{J, s}_j, Z^s_j, m^s_j \rbrace_{j = h}^N$.

Thanks to~\eqref{eq:measure_additivity}, we can write
\begin{align}
\label{eq:generating_functional_additivity}
\begin{split}
W[\src] 
& = N_h + \log \int \dgauss{\sing{g}{\le h-1}}{\psi}\int \dgauss{\sing{g}{h}}{\psi}
\exp (\sing{V}{h}[\psi, \src]) \\
& = N_h + \log \int \dgauss{\sing{g}{\le h-1}}{\psi} \exp(\sing{\mathscr{V}}{h-1}[\psi, \src])
\end{split}
\end{align}
where $N_h$ is an inessential constant and 
\begin{equation}
\label{eq:V_h_integration_definition}
\sing{\mathscr{V}}{h-1}[\psi, \src] \equiv \log \int \dgauss{\sing{g}{h}}{\psi} \, \exp ( \sing{V}{h}[\psi, \src] ).
\end{equation}
This construction shows how to unambiguously pass from $h$ to $h-1$ inside~\eqref{eq:generating_functional_on_scale_h}. In the next section, we shall explicitly prove (see~\eqref{eq:sum_over_trees_P}) that $\sing{\mathscr{V}}{h}[\psi, \src]$ has always the same structure for every scale $h$, that is, it admits an expansion like~\eqref{eq:veff_kernels} with no more than three derivatives per fermion field.

\subsection{Tree expansion}
\label{subsec:trees}

The Gallavotti-Nicolò tree expansion (or, in what follows, the \emph{tree expansion} tout-court) is a convenient method to represent the recursive structure of multiple Renormalization Group steps~\cite{Gallavotti_1995, Mastropietro_Nonpert_Renorm, Gentile_2001}. As a starting point, we express~\eqref{eq:V_h_integration_definition} as
\begin{equation}
\label{eq:truncated_expval_potential}
\sing{\mathscr{V}}{h-1}[\psi, \src] = \sum_{n = 1}^{+\infty} \frac{1}{n!} \TExp{h}(\underbrace{\sing{V}{h}[\, \cdot + \psi, \omega], \dots, \sing{V}{h}[\, \cdot + \psi, \omega]}_{n \text{ times}}),
\end{equation}
where the \emph{truncated expectation value} of any set of $n$ functionals $\lbrace \mathscr{O}_1, \dots, \mathscr{O}_n \rbrace$ is
\begin{equation}
\TExp{h}(\mathscr{O}_1, \dots, \mathscr{O}_n) \equiv \frac{\partial^n}{\partial t_1 \cdots \partial t_n} \log \int \dgauss{\sing{g}{h}}{\psi} \, \exp(t_1 \mathscr{O}_1 + \cdots + t_n \mathscr{O}_n) \eval_{t_1 = 0, \dots, t_n = 0}
\end{equation}
(in the following sections, we will often write $\TExp{h}(\lbrace \mathscr{O}_\ell \rbrace_{\ell=1}^n)$ in place of $\TExp{h}(\mathscr{O}_1, \dots, \mathscr{O}_n)$). If $h = N$, we depict~\eqref{eq:truncated_expval_potential} in the following way,
\begin{equation}
\label{eq:gallavotti_N-1}
\sing{\mathscr{V}}{N-1} = 
\begin{tikzpicture}[baseline={(0, 0.2 - \MathAxis pt)}]
\draw (-0.8, 0) -- (0, 0);
\fill[black] (0,0) circle (1.5pt);
\node (N-1) at (0, -0.3) {\tiny $N$};
\draw (0,0) -- (0.8, 0);
\fill[black] (0.8,0) circle (1.5pt);
\node (N) at (0.8, -0.3) {\tiny $N+1$};
\end{tikzpicture} 
\,\, + \,\,
\begin{tikzpicture}[baseline={(0, 0.2 - \MathAxis pt)}]
\draw (-0.8, 0) -- (0, 0);
\draw (0, 0) -- (0.8, 0.4);
\draw (0, 0) -- (0.8, -0.4);
\fill[black] (0,0) circle (1.5pt);
\fill[black] (0.8, 0.4) circle (1.5pt);
\fill[black] (0.8, -0.4) circle (1.5pt);
\node (N-1) at (0, -0.3) {\tiny $N$};
\node (N1) at (0.8, -0.7) {\tiny $N+1$};
\node (N2) at (0.8, 0.7) {\tiny $N+1$};
\end{tikzpicture} 
\,\, + \,\,
\begin{tikzpicture}[baseline={(0, 0.2 - \MathAxis pt)}]
\draw (-0.8, 0) -- (0, 0);
\draw (0, 0) -- (0.8, 0.4);
\draw (0, 0) -- (0.8, 0);
\draw (0, 0) -- (0.8, -0.4);
\fill[black] (0,0) circle (1.5pt);
\fill[black] (0.8, 0.4) circle (1.5pt);
\fill[black] (0.8, -0.4) circle (1.5pt);
\fill[black] (0.8, 0) circle (1.5pt);
\node (N-1) at (0, -0.3) {\tiny $N$};
\node (N1) at (0.8, -0.7) {\tiny $N+1$};
\node (N2) at (1.3, 0) {\tiny $N+1$};
\node (N3) at (0.8, 0.7) {\tiny $N+1$};
\end{tikzpicture}
\,\, + \,\,
\begin{tikzpicture}[baseline={(0, 0.2 - \MathAxis pt)}]
\draw (-0.8, 0) -- (0, 0);
\draw (0, 0) -- (1, 0.6);
\draw (0, 0) -- (1, 0.2);
\draw (0, 0) -- (1, -0.2);
\draw (0, 0) -- (1, -0.6);
\fill[black] (0,0) circle (1.5pt);
\fill[black] (1, 0.6) circle (1.5pt);
\fill[black] (1, 0.2) circle (1.5pt);
\fill[black] (1, -0.6) circle (1.5pt);
\fill[black] (1, -0.2) circle (1.5pt);
\node (N-1) at (0, -0.3) {\tiny $N$};
\node (N1) at (1, 0.9) {\tiny $N+1$};
\node (N2) at (1.5, 0.2) {\tiny $N+1$};
\node (N3) at (1.5, -0.2) {\tiny $N+1$};
\node (N4) at (1, -0.9) {\tiny $N+1$};
\end{tikzpicture}
\,\, + \cdots,
\end{equation}
it being understood that each dot with label $N+1$ represents a $\sing{V}{N}[\psi, \omega]$ factor. If every $\sing{V}{h}$ factor appearing at the right hand side of~\eqref{eq:truncated_expval_potential} is decomposed as $\loc \sing{V}{h} + \ren \sing{V}{h} = \loc \sing{V}{h} + \ren \sing{\mathscr{V}}{h}$, this equality can be recast into
\begin{equation}
\label{eq:truncated_expval_potential_splitting}
\sing{\mathscr{V}}{h-1} = \sum_{n = 1}^{+\infty} \frac{1}{n!} \sum_{\lbrace \mathscr{O}_1, \dots, \mathscr{O}_n \rbrace} \TExp{h}(\mathscr{O}_1, \dots, \mathscr{O}_n),
\end{equation}
where each $\mathscr{O}_\ell$ can be either equal to $\loc \sing{V}{h}[\, \cdot + \psi, \omega]$ or $\ren \sing{\mathscr{V}}{h}[\, \cdot + \psi, \omega]$. The renormalized potential $\ren \sing{\mathscr{V}}{h}[\psi, \omega]$ admits an expansion with the same structure as~\eqref{eq:truncated_expval_potential_splitting}, namely
\begin{equation}
\label{eq:truncated_expval_potential_renormalized}
\ren \sing{\mathscr{V}}{h} = \sum_{n = 1}^{+\infty} \frac{1}{n!} \sum_{\lbrace \mathscr{O}_1, \dots, \mathscr{O}_n \rbrace} \ren \TExp{h+1}(\mathscr{O}_1, \dots, \mathscr{O}_n);
\end{equation}
this time, each $\mathscr{O}_\ell$ can be either equal to $\loc \sing{V}{h+1}[\, \cdot + \psi, \omega]$ or $\ren \sing{\mathscr{V}}{h+1}[\, \cdot + \psi, \omega]$. After plugging~\eqref{eq:truncated_expval_potential_renormalized} into~\eqref{eq:truncated_expval_potential_splitting}, we obtain an expansion for $\sing{\mathscr{V}}{h-1}[\psi, \omega]$ as a function of $\loc \sing{V}{h}$, $\loc \sing{V}{h+1}$, $\ren \sing{\mathscr{V}}{h+1}$. This construction can be iterated by writing $\ren \sing{\mathscr{V}}{h+1}$ in terms of $\loc \sing{V}{h+2}, \ren \sing{\mathscr{V}}{h+2}$, then $\ren \sing{\mathscr{V}}{h+2}$ in terms of $\loc \sing{V}{h+3}, \ren \sing{\mathscr{V}}{h+3}$, and so on; ultimately, $\sing{\mathscr{V}}{h-1}$ will be expressed as a complicated sum of nested truncated expectation values of $\loc \sing{V}{h}$, $\loc \sing{V}{h+1}, \dots, \loc \sing{V}{N-1}$, $\sing{V}{N}$. With a simple generalization of the graphical representation~\eqref{eq:gallavotti_N-1}, each term of this expansion can be associated with a Gallavotti-Nicolò tree:
\begin{defn}[Gallavotti-Nicolò trees]
\label{def:trees}
Let $\tau$ be a finite, rooted, labeled tree graph with at least two vertices. $\tau$ is a Gallavotti-Nicolò tree if its root vertex is simple and if the labels $\lbrace h_v \in \Z \rbrace_{v \in \vrt(\tau)}$ satisfy $h_{v_2} = h_{v_1} + 1$ whenever $v_2$ is a child vertex of $v_1$.
\end{defn}

An example of Gallavotti-Nicolò tree is provided in figure~\ref{fig:tree}. The set of all the Gallavotti-Nicolò trees whose root label is equal to $h$ and whose endpoint labels are less or equal than $N+1$ is denoted by $\trees_{h|N}$. For every $j \in \Z$, the elements of $\trees_{j|j}$ are called \emph{trivial trees}; the set $\trees'_{h|N} \equiv \trees_{h|N} \setminus \trees_{h|h}$ only consists of nontrivial trees. Given some $\tau \in \trees_{h|N}$, the unique vertex $w_0 \in \vrt(\tau)$ with label $h+1$ will be called \emph{first nonroot vertex}. For every pair of vertices $v, w \in \tau$, we write $v \succ w$ ($w \prec v$) if $v$ is a child of $w$ and $v > w$ ($w < v$) if $h_v > h_w$. Furthermore, the notation $v \triangleleft w$ ($w \triangleright v$) is used if there exists a path $\lbrace a_1, \dots, a_n \rbrace \subseteq \vrt(\tau)$ such that $v = a_1 \prec \cdots \prec a_n = w$. For every $v \in \vrt(\tau)$, we define $\tau_v$ as the rooted subtree of $\tau$ obtained by taking $v$ itself, its parent vertex and all the vertices $\lbrace w \colon w \triangleright v \rbrace$, together with the respective edges (by construction, $\tau_v$ belongs to $\trees_{h_v - 1|N}$). Finally, $s_v$ denotes the number of children of $v$.

Each tree is recursively assigned a \emph{value} in the following way:
\begin{defn}[Value of a tree]
\label{def:tree_value}
Let $\tau \in \trees_{h|N}$ and let $w_0 \in \vrt(\tau)$ be its first nonroot vertex. We let
\begin{equation}
\val(\tau)[\psi, \omega] \equiv
\begin{cases}
\sing{V}{N}[\psi, \omega]			& \tau \text{ trivial}, \, h = N \\
\loc \sing{V}{h+1}[\psi, \omega]	& \tau \text{ trivial}, \, h < N \\[3pt]
\displaystyle \frac{1}{s_{w_0}!} \, \TExp{h_{w_0}}( \lbrace \ren^\star \val(\tau_v)[\, \cdot + \psi, \omega] \rbrace_{v \succ w_0} )		& \tau \text{ nontrivial}
\end{cases}
\end{equation}
where the operator $\ren^\star$ acts as
\begin{equation}
\ren^\star \val(\tau_v)[\, \cdot + \psi, \omega] =
\begin{cases}
\val(\tau_v)[\, \cdot + \psi, \omega]		&	\tau_v \text{ trivial} \\
\ren \val(\tau_v)[\, \cdot + \psi, \omega]	&	\tau_v \text{ nontrivial}
\end{cases}
\end{equation}
\end{defn}
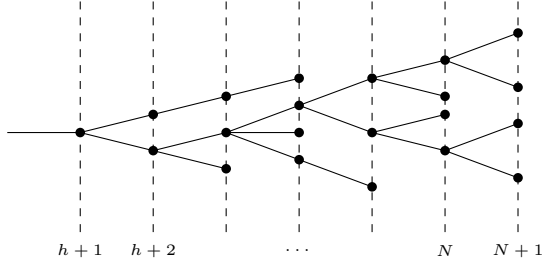
\begin{figure}[t]
\centering
\begin{tikzpicture}[scale=1.2]
\draw	(-0.8, 0) -- (0, 0)
	 	(0, 0) -- (0.8, 0.2) -- (1.6, 0.4) -- (2.4, 0.6)
		(0, 0) -- (0.8, -0.2) -- (1.6, -0.4)
		(0.8, -0.2) -- (1.6, 0)
		(1.6, 0) -- (2.4, 0.3)
		(1.6, 0) -- (2.4, 0)
		(1.6, 0) -- (2.4, -0.3) -- (3.2, -0.6)
		(2.4, 0.3) -- (3.2, 0.6)
		(2.4, 0.3) -- (3.2, 0)
		(3.2, 0.6) -- (4, 0.8)
		(3.2, 0.6) -- (4, 0.4)
		(3.2, 0) -- (4, 0.2)
		(3.2, 0) -- (4, -0.2)
		(4, 0.8) -- (4.8, 1.1)
		(4, 0.8) -- (4.8, 0.5)
		(4, -0.2) -- (4.8, 0.1)
		(4, -0.2) -- (4.8, -0.5);
\fill[black] 	(0,0) circle (1.5pt)
				(0.8, 0.2) circle (1.5pt)
				(0.8, -0.2) circle (1.5pt)
				(1.6, -0.4) circle (1.5pt)
				(1.6, 0) circle (1.5pt)
				(1.6, 0.4) circle (1.5pt)
				(2.4, 0.6) circle (1.5pt)
				(2.4, 0.3) circle (1.5pt)
				(2.4, 0) circle (1.5pt)
				(2.4, -0.3) circle (1.5pt)
				(3.2, 0.6) circle (1.5pt)
				(3.2, 0) circle (1.5pt)
				(3.2, -0.6) circle (1.5pt)
				(4, 0.8) circle (1.5pt)
				(4, 0.4) circle (1.5pt)
				(4, 0.2) circle (1.5pt)
				(4, -0.2) circle (1.5pt)
				(4.8, 1.1) circle (1.5pt)
				(4.8, 0.5) circle (1.5pt)
				(4.8, 0.1) circle (1.5pt)
				(4.8, -0.5) circle (1.5pt);
\draw[dashed]	(0, -1.1) -- (0, 1.5)
				(0.8, -1.1) -- (0.8, 1.5)
				(1.6, -1.1) -- (1.6, 1.5)
				(2.4, -1.1) -- (2.4, 1.5)
				(3.2, -1.1) -- (3.2, 1.5)
				(4, -1.1) -- (4, 1.5)
				(4.8, -1.1) -- (4.8, 1.5);
\node	(h)		at	(0, -1.3)	{\tiny $h+1$};
\node	(h1)		at	(0.8, -1.3)	{\tiny $h+2$};
\node	(dd)		at	(2.4, -1.3)	{\tiny $\dots$};
\node	(N)		at	(4, -1.3)	{\tiny $N$};
\node	(N1)		at	(4.8, -1.3)	{\tiny $N+1$};
\end{tikzpicture}
\caption{A Gallavotti-Nicolò tree. The root vertex is located at the extreme left of the graph and it is not explicitly marked. The endpoints with label $N+1$ or $j \le N$ respectively correspond to $\sing{V}{N}$ or $\loc \sing{V}{j}$ factors.}
\label{fig:tree}
\end{figure}
\noindent
According to this definition, the complicated expansion described above can be synthetically represented as
\begin{equation}
\label{eq:sum_over_trees}
\sing{\mathscr{V}}{h}[\psi, \omega] = \sum_{\tau \in \trees_{h|N}'} \val(\tau)[\psi, \omega].
\end{equation}

The tree expansion can also be defined at the level of the kernels $\sing{W}{h}$ appearing in~\eqref{eq:veff_kernels}. To do this, for every $\tau \in \trees_{h|N}, v \in \vrt(\tau)$ we define
\begin{itemize}
\item A set $P_v = P_v^+ \cup P_v^-$ whose elements are called \emph{external fields} of $v$. If $v$ is an endpoint, its external fields are equal to the fermion fields occurring into the explicit expression of $\val(\tau_v)[\psi, \omega]$; if $v$ is not an endpoint, then $P_v \subseteq \cup_{w \succ v} P_w$. A general collection $\lbrace P_v \rbrace_{v \in \vrt(\tau)}$ is denoted by $\underline{P}$.
\item A set $\Delta_v$ that contains the derivative indices falling on the external fields of $v$. As before, $\underline{\Delta}$ denotes a general collection $\lbrace \Delta_v \rbrace_{v \in \vrt(\tau)}$.
\item A set $I_v \equiv \cup_{w \succ v} P_w \setminus P_v$ whose elements are called \emph{internal fields} of $v$. We say that a vertex $v \in \vrt(\tau)$ is \emph{trivial} (with respect to a given $\underline{P}$) if $I_v = \varnothing$. Moreover, if $v \succ w$, we let $Q_v \equiv P_v \cap I_w$.
\item A set $U_v = U^J_v \cup U^{\eta, +}_v \cup U^{\eta, -}_v$ that contains all the source fields appearing inside the endpoints that are connected with $v$ by an increasing path of vertices.
\end{itemize}
Thanks to~\eqref{eq:renormalization_kernel_action}, the expansion~\eqref{eq:sum_over_trees} becomes
\begin{align}
\label{eq:sum_over_trees_P}
\begin{split}
\sing{\mathscr{V}}{h}[\psi, \omega] & = \sum_{\tau \in \trees_{h|N}'} \sum_{\underline{P}, \underline{\Delta}} \int \dd{\underline{x}(P_{w_0})} \dd{\underline{y}(U_{w_0})} \, W_{\tau, \underline{P}, \underline{\Delta}}(\underline{x}(P_{w_0}), \underline{y}(U_{w_0})) \Psi(P_{w_0}, \Delta_{w_0}) \Src(U_{w_0}), \\
\ren \sing{\mathscr{V}}{h}[\psi, \omega] & =
\begin{multlined}[t]
\sum_{\tau \in \trees_{h|N}'} \sum_{\underline{P}, \underline{\Delta}} \int \dd{\underline{x}(P_{w_0})} \dd{\underline{y}(U_{w_0})} \, (\kren W)_{\tau, \underline{P}, \underline{\Delta}}(\underline{x}(P_{w_0}), \underline{y}(U_{w_0})) \, \times \\
\hspace*{8cm} \times \Psi(P_{w_0}, \Delta'_{w_0}) \Src(U_{w_0}),
\end{multlined}
\end{split}
\end{align}
where $w_0$ is the first nonroot vertex of $\tau$, the monomials $\Psi(P_{w_0}, \Delta_{w_0}), \Src(U_{w_0})$ are constructed as in~\eqref{eq:field_monomials},~\eqref{eq:source_field_monomials} and $\underline{x}(P_{w_0}), \underline{y}(U_{w_0})$ are the spacetime points that the field monomials $\Psi(P_{w_0}, \Delta_{w_0}), \Src(U_{w_0})$ depend on. The set $\Delta'_{w_0}$ differs from $\Delta_{w_0}$ due to the extra derivatives produced the renormalization operator, as showed in~\eqref{eq:ren_kernel_eq}. The first line of~\eqref{eq:sum_over_trees_P} shows that $\sing{\mathscr{V}}{h}$ does indeed admit an expansion of the form~\eqref{eq:veff_kernels} for every scale $h$.

The kernel $W_{\tau, \underline{P}, \underline{\Delta}}$ is recursively defined as
\begin{equation}
\label{eq:W_tau_recursive}
W_{\tau, \underline{P}, \underline{\Delta}}(\underline{x}, \underline{y}) \equiv \frac{1}{s_{w_0}!} \! \int \! \prod_{v \succ w_0} \!\! \dd{\underline{x}(Q_v)} \TExp{h_{w_0}}(\lbrace \Psi(Q_v, \Delta_v') \rbrace_{v \succ w_0}) \! \prod_{v \succ w_0} \!\! \kren_v W_{\tau_v, \underline{P}, \underline{\Delta}}(\underline{x}(P_v), \underline{y}(U_v)),
\end{equation}
where $\kren_v = \kren$ or $\kren_v = \id$ depending on whether a renormalization operator acts on $v$ or not and $\Delta'_v$ is constructed in the same way as $\Delta'_{w_0}$. If $v$ is an endpoint, $W_{\tau_v, \underline{P}, \underline{\Delta}}$ is equal to one among
\begin{subequations}
\begin{align}
\begin{split}
\label{eq:J_endpoints}
\sing{W}{h; \alpha \beta \mu}_{(J)}(x, y, z) & \equiv \delta_N(x - y) \delta_N(x - z) (\gamma^\mu_{J,h})_{\alpha \beta},
\end{split} \\
\begin{split}
\label{eq:lambda_endpoints}
\sing{W}{N; \alpha \beta \sigma \rho}_{(\lambda)}(x_1, x_2, x_3, x_4) & \equiv - \frac{\lambda^2}{2} \delta_N(x_1  - x_2) \delta_N(x_3 - x_4) v(x_1 - x_3) \, (\Upsilon^\mu)_{\alpha \beta} (\Upsilon_\mu)_{\sigma \rho},
\end{split} \\
\begin{split}
\label{eq:eta_endpoints}
\sing{W}{N; \alpha \beta}_{(\eta)}(x, z) & \equiv \delta_{\alpha \beta} \delta_N(x - z),
\end{split}
\end{align}
\end{subequations}
based on the sets $P_v, \Delta_v, U_v$ and on the structure of $\tau$ (the matrix $\gamma^\mu_{J, h}$ is constructed as in~\eqref{eq:renormalized_propagator} with $Z^J_h$ in place of $Z_h$ and $\Upsilon^\mu \equiv \gamma^\mu - \kappa \gamma^\mu\gamma^5$). This suggests that endpoints can be naturally classified as $J$, $\lambda$ or $\eta$ endpoints. By construction, $J$ endpoints can occur at every energy scale, while $\lambda, \eta$ endpoints must lie on scale $N+1$; moreover, definition~\ref{def:tree_value} implies that $W_{\tau, \underline{P}, \underline{\Delta}}$ is \emph{identically vanishing} if $\tau$ contains a $J$ endpoint lying on scale $h$ that is not immediately preceded by a nontrivial vertex with label $h-1$.

\subsection{Feynman diagrams}
The truncated expectation value of a product of fermionic fields admits the representation
\begin{equation}
\label{eq:TExp_feynman}
\TExp{h}(\psi^+_{\alpha_1}(x_1) \cdots \psi^+_{\alpha_n}(x_n) \psi^-_{\beta_1}(y_1) \cdots \psi^-_{\beta_n}(y_n)) = \sum_{\phi \in \mathcal{F}} \epsilon_\phi \prod_{\ell \in \edg(\phi)} \sing{g}{h}_{\alpha_{\ell_1} \beta_{\ell_2}}(x_{\ell_1}, x_{\ell_2}),
\end{equation}
where $\mathcal{F}$ is the set of all the possible connected graphs (called \emph{Feynman diagrams}) obtained by drawing the spacetime points $x_1, \dots, x_n, y_1, \dots, y_n$, reorganizing the fields in pairs of the form $\lbrace \psi^-_{\alpha_a}(x_a), \psi^+_{\beta_b}(y_b) \rbrace$ and drawing an oriented line from $y_b$ to $x_a$ for each of such pairs ($\epsilon_\phi$ is the sign of the permutation needed to convert the product~\eqref{eq:TExp_feynman} into the set of ordered pairs associated with $\phi$). By iteratively applying~\eqref{eq:TExp_feynman} to~\eqref{eq:W_tau_recursive}, we can express $W_{\tau, \underline{P}, \underline{\Delta}}$ as a sum over Feynman diagrams with a \emph{multiscale} structure determined by the tree $\tau$. The procedure to construct a multiscale Feynman diagram associated with a tree $\tau$ goes as follows:
\begin{itemize}
\item Draw an interaction vertex for each endpoint of $\tau$. All the possible interaction vertices are listed in Figure~\ref{fig:interaction_vertices}.
\item If a family of endpoints has a common \emph{nontrivial} parent vertex $v \in \vrt(\tau)$, connect the corresponding family interaction vertices by contracting some of their external legs, respecting the directions of the arrows. A single scale propagator $\sing{g}{h_v}$ is associated with each contracted line. As a result, one obtains a \emph{cluster} for each family of endpoints sharing the same nontrivial parent vertex.
\item Iterate the above construction following the vertex structure of $\tau$. Namely, if the vertices $v_1, \dots, v_n$ associated with a family of clusters and/or interaction vertices have a common nontrivial parent vertex $w$, connect these clusters and/or interaction vertices with a set of lines carrying the single scale propagator $\sing{g}{h_w}$.
\item If the $\ren$ operator acts nontrivially on a vertex $v$, it must act on the corresponding cluster as well.
\end{itemize}
\begin{figure}[t]
\centering
\begingroup
\setlength{\tabcolsep}{0.3cm}
\begin{tabular}{ccl}
\begin{tikzpicture}[baseline={-2pt}]
\begin{feynman}
\draw[dashed, with arrow]	(-1, -1) -- (0, 0);
\draw[dashed, with arrow]	(0, 0) -- (-1, 1);
\draw[dashed, with arrow]	(2, -1) -- (1, 0);
\draw[dashed, with arrow]	(1, 0) -- (2, 1);
\draw[photon]				(0, 0) -- (1, 0);
\end{feynman}
\node	(p1)	at	(-0.15, -0.65)	{$p_1$};
\node	(p2)	at	(1.15, -0.65)	{$p_2$};
\node	(p3)	at	(-0.15, 0.65)	{$p_3$};
\node	(p4)	at	(1.15, 0.65)	{$p_4$};
\node	(i1)	at	(-1, -0.7)		{$\alpha$};
\node	(i2)	at	(-1, 0.65)		{$\beta$};
\node	(i3)	at	(2, -0.7)		{$\sigma$};
\node	(i4)	at	(2, 0.65)		{$\rho$};
\end{tikzpicture}
& $\longrightarrow$ &
$\displaystyle -\lambda^2 \, \displaystyle \frac{1}{L^4} \sum_q \hat{v}(q) \, \delta_{p_1 - p_3 - q} \, \delta_{p_2 - p_4 + q} \, (\Upsilon^\mu)_{\alpha \beta} (\Upsilon_\mu)_{\sigma \rho}$ \\[1.2cm]
\begin{tikzpicture}[baseline={-2pt}]
\begin{feynman}
\draw[dashed, with arrow]		(-1, -1) -- (0, 0);
\draw[dashed, with arrow]		(0, 0) -- (-1, 1);
\draw[photon, densely dotted]	(0, 0) -- (1, 0);
\end{feynman}
\node	(p1)	at	(-0.25, -0.7)	{$p$};
\node	(p2)	at	(-0.25, 0.75)	{$p'$};
\node	(i1)	at	(-1, -0.7)		{$\alpha$};
\node	(i2)	at	(-1, 0.65)		{$\beta$};
\node	(m)		at	(1.2, 0)		{$\mu$};
\node	(q)		at	(0.5, 0.3)		{$q$};
\end{tikzpicture}
& $\longrightarrow$ &
$\displaystyle 
\begin{pmatrix}
0							&	\sigma^\mu_+ Z^{J, +}_h \\
\sigma^\mu_- Z^{J, -}_h		&	0
\end{pmatrix}_{\alpha \beta}
\delta_{p' - p - q} \equiv (\gamma^\mu_J)_{h, \alpha \beta} \, \delta_{p' - p - q}$ \\[1.2cm]
\begin{tikzpicture}[baseline={-2pt}]
\begin{feynman}
\draw[dashed, with arrow]		(-1, 0.5) -- (1, 0.5);
\draw[dashed, with arrow]		(1, -0.5) -- (-1, -0.5);
\fill[black]	(1, 0.5)	circle	(1.5pt);
\fill[black]	(-1, -0.5)	circle	(1.5pt);
\end{feynman}
\node	(p1)		at	(0, 0.8)		{$p$};
\node	(p2)		at	(0, -0.8)		{$p$};
\node	(k1)		at	(1.2, 0.75)		{$k$};
\node	(k2)		at	(1.2, 0.35)		{$\beta$};
\node	(k3)		at	(-1.25, -0.27)	{$k$};
\node	(k4)		at	(-1.25, -0.67)	{$\beta$};
\node	(a1)		at	(1.25, -0.5)		{$\alpha$};
\node	(a2)		at	(-1.25, 0.5)	{$\alpha$};
\end{tikzpicture}
& $\longrightarrow$ &
$\delta_{\alpha \beta} \, \delta_{p - k}$
\end{tabular}
\endgroup
\caption{Interaction vertices and their values in Fourier space. These vertices are respectively associated with $\lambda$ endpoints, $J$ endpoints and $\eta$ endpoints. Note that the value of a $J$ interaction vertex depends on the scale $h$ of the first nontrivial vertex that precedes the corresponding endpoint.}
\label{fig:interaction_vertices}
\end{figure}

\begin{figure}[t]
\centering
\begin{tikzpicture}[baseline]
\draw	(-0.5, 0) -- (4, 0)
		(1.5, 0) -- (4, -1)
		(3, 0) -- (3.5, -0.25)
		(0, 0) -- (4, 0.5)
		(0, 0) -- (4, 1);
\fill[black]		(0, 0)			circle (1.5pt)
				(1.5, 0)			circle (1.5pt)
				(3, 0)			circle (1.5pt)
				(3.5, -0.25)		circle (1.5pt)
				(4, 0)			circle (1.5pt)
				(4, -1)			circle (1.5pt)
				(4, 0.5)			circle (1.5pt)
				(4, 1)			circle (1.5pt);
\draw[dashed]	(0, -1.2) -- (0, 1.2)
				(1.5, -1.2) -- (1.5, 1.2)
				(3, -1.2) -- (3, 1.2)
				(4, -1.2) -- (4, 1.2);
\node	(hs)		at	(0, -1.4)		{\tiny $h^\star$};
\node	(h2)		at	(1.5, -1.4)		{\tiny $h_2$};
\node	(h1)		at	(3, -1.4)		{\tiny $h_1$};
\node	(ep)		at	(4, -1.4)		{\tiny $N+1$};
\node	(l1)		at	(4.2, 0)		{\tiny $\lambda$};
\node	(l2)		at	(4.2, -1)		{\tiny $\lambda$};
\node	(J)		at	(3.7, -0.25)		{\tiny $J$};
\node	(n1)		at	(4.2, 0.5)		{\tiny $\eta$};
\node	(n2)		at	(4.2, 1)		{\tiny $\eta$};
\end{tikzpicture}
\qquad \qquad \qquad \quad
\begin{tikzpicture}[baseline]
\begin{feynman}
\draw[with arrow]	(0, 0) -- (-1, 1);
\draw[with arrow]	(-1, 1) -- (-2, 2);
\draw[with arrow]	(-2, 2) -- (-2.6, 2.6);
\draw[with arrow]	(-2.6, -2.6) -- (-2, -2);
\draw[with arrow]	(-2, -2) -- (-1, -1);
\draw[with arrow]	(-1, -1) -- (0, 0);
\draw[photon]		(-1, -1) -- (-1, 1)
					(-2, -2) -- (-2, 2);
\draw[photon, densely dotted]		(0, 0) -- (0.8, 0);
\end{feynman}
\draw[black, thick]	(-1.15, -1.2) rectangle (0.2, 1.2);
\draw[black, thick]	(-2.15, -2.2) rectangle (0.35, 2.2);
\draw[black, thick]	(-2.75, -2.75) rectangle (0.5, 2.75);
\fill[black]		(-2.6, -2.6) circle (1.5pt)
				(-2.6, 2.6) circle (1.5pt);
\node	(hs1)	at	(-2.15, -2.55)		{\tiny $h^\star$};
\node	(hs2)	at	(-2.15, 2.55)		{\tiny $h^\star$};
\node	(h21)	at	(-1.3, -1.7)		{\tiny $h_2$};
\node	(h22)	at	(-1.3, 1.7)			{\tiny $h_2$};
\node	(h11)	at	(-0.3, -0.7)		{\tiny $h_1$};
\node	(h12)	at	(-0.3, 0.7)			{\tiny $h_1$};
\end{tikzpicture}
\caption{A multiscale Feynman diagram (on the right) together with the corresponding Gallavotti-Nicolò tree (on the left).}
\label{fig:feynman_example}
\end{figure}
\noindent
While computing the value of a diagram, one needs to sum over all the momenta flowing within its internal propagators; in the $L \to +\infty$ limit, these sums are replaced by integrals and the Kronecker deltas $\delta_{p - k}, \delta_{p_1 - p_3 - q}, \cdots$ are replaced by Dirac deltas.

To see how renormalization acts at the level of Feynman diagrams, it is convenient to work directly in Fourier space. As an example, let us consider a kernel $W$ acting on a field monomial of the form $\psi^+_\alpha \psi^-_\beta J_\mu$. Recalling that the highest order localizations occurring in~\eqref{eq:localization_definition} do not contribute to $\loc \sing{\mathscr{V}}{h}[\psi, \omega]$ (see~\eqref{eq:localization_request}), a straightforward application of~\eqref{eq:localization_definition} yields
\begin{equation}
\ren \int W^{\alpha \beta \mu} \psi^+_\alpha \psi^-_\beta J_\mu = \frac{1}{L^8} \sum_{p', p \in (2\pi Z/L)^4} [\hat{W}^{\alpha \beta \mu}(p', p) - \hat{W}^{\alpha \beta \mu}(0, 0)] \hat{\psi}^+_{\alpha p'} \hat{\psi}^-_{\beta p} \hat{J}_{\mu, p' - p}.
\end{equation}
Hence, renormalizing a multiscale Feynman diagram with two external fermionic legs and one $J$ source leg simply amounts to subtract its zero Fourier mode. Similarly, one can see that the renormalization operator acts on a two-legged multiscale Feynman diagram by subtracting its first-order Taylor polynomial centered around zero. Some examples of this will be provided in Section~\ref{sec:g-2_calculation}.

We will use the Feynman diagram representation to compute $\rgyr$ at lowest order in $\lambda^2$. However, to obtain bounds on $\rgyr$ at all orders it is convenient to represent the truncated expectation values in a different way, as discussed below. This is due to the fact that Feynman diagrams have bad combinatorial properties.

\section{Bounds on the renormalized expansion}
\label{sec:renormalized_expansion_bounds}

In order to bound the kernels $W_{\tau, \underline{P}, \underline{\Delta}}$ occurring in~\eqref{eq:sum_over_trees_P}, it is convenient to introduce a suitable \emph{weighted} $L^1$ norm~\cite[Section 4.2.1]{Giuliani_Mastropietro_Rychkov_2021}. For any finite subset $A \subseteq \mathcal{M}_L$ and for every integer $h \le N$, we define the weight function as $\weight_h(A) \equiv \exp (c \sqrt{2^h \, \mathrm{diam}(A)})$, where the notion of diameter is defined with respect to the distance $\abs*{\, \cdot \,}_L$ introduced in Section~\ref{subsec:model_definition} (if $A = \lbrace x, y \rbrace$, $\mathrm{diam}(A)$ reduces to $\abs*{x - y}_L$) and $c$ is some positive constant such that $\int \dd[4]{x} \abs*{\sing{g}{h}(x)} \weight_h(x), \int \dd[4]{x} \abs*{\delta_N(x)} \weight_{N+1}(x) \le +\infty$. It will be clear from Lemma~\ref{lm:propagator_bounds} that this constant exists and it is determined by the choice of the cutoff function $\chi_0$.
\begin{defn}[Weighted norm]
\label{def:norm}
Let $h \in \Z, h \le N$. The $h$-weighted norm of a kernel $\sing{W}{h}(A, B)$ is
\begin{equation}
\label{eq:weighted_norm_alternative}
\norm*{\sing{W}{h}(A, B)}_{h} \equiv \int \dd_0{\underline{x}} \dd{\underline{y}} \abs*{\sing{W}{h}(A, \underline{x}; B, \underline{y})} \, \weight_h( \underline{x}, \underline{y}),
\end{equation}
where the notation $\int \dd_0{\underline{x}} \dd{\underline{y}}$ means that we must integrate over all the spacetime points belonging to $(\underline{x}, \underline{y})$ except one, which is held fixed to zero.
\end{defn}

The effect of the renormalization on kernels is described by the following lemma.
\begin{lm}[Renormalized norm bounds]
\label{lm:taylor_renormalization_gain}
Let $\sing{W}{h}(x_0, \dots, x_n) \equiv \sing{W}{h}(x_0, \underline{x})$ be a smooth, periodic, translationally invariant kernel (that is, $\sing{W}{h}(x_0, x_1, \dots, x_n) = \sing{W}{h}(0, x_1-x_0, \dots, x_n-x_0)$). If $\norm*{\sing{W}{h}}_h < +\infty$, then
\begin{subequations}
\begin{align}
\begin{split}
\label{eq:renormalization_gain_kernels}
\norm*{(\kren_{R, x_0}^0 \sing{W}{h})^{\underline{\alpha}}}_{h'} & \le C_\ren \cdot 2^{-(h+1)R} \norm*{\sing{W}{h}}_h
\end{split}\\
\begin{split}
\label{eq:localization_gain_kernels}
\norm*{(\kloc_{\ell, x_0} \sing{W}{h})^{\underline{\alpha}}}_{h'} & \le C_\ren \cdot 2^{-(h+1)\ell} \norm*{\sing{W}{h}}_h
\end{split}
\end{align}
\end{subequations}
for \emph{every} scale $h' < h$, where $C_\ren$ is a positive constant that does not depend on $h, h'$.
\end{lm}

\begin{proof}
We limit to prove~\eqref{eq:renormalization_gain_kernels} (the proof of~\eqref{eq:localization_gain_kernels} can be done in an entirely similar way). According to~\eqref{eq:ren_taylor_kernel_general}, we have
\begin{multline}
\norm*{(\kren_{R, x_0}^0 W)^{\underline{\alpha}}}_h \le \int \dd{\underline{x}} \binom{R}{\underline{\alpha}} \int_0^1 \dd{t} \frac{(1-t)^{R-1}}{(R-1)!} \int \dd{\underline{y}} \abs*{W(0, \underline{y}) \, \underline{y}^{\underline{\alpha}}_\torus} \, \times \\
\times \abs*{\delta^{(n)}_N(\underline{x} - t \underline{y}_\torus)} \, e^{c \sqrt{2^{h'} \mathrm{diam}(0, \underline{x})}}.
\end{multline}
After performing the change of variables $\underline{x} \mapsto \underline{x} + t \underline{y}_\torus$, the can be rewritten as
\begin{equation}
\label{eq:lm:renormalization_gain_quasi_final}
\norm*{(\kren_{R, x_0}^0 \sing{W}{h})^{\underline{\alpha}}}_{h'} \le C \int_0^1 \dd{t} \int \dd{\underline{x}} \dd{\underline{y}} \abs*{W(0, \underline{y}) \, \underline{y}^{\underline{\alpha}}_\torus \, \delta^{(n)}_N(\underline{x})} \, e^{c \sqrt{2^{h'}\mathrm{diam}(0, \underline{x} + t \underline{y}_\torus)}},
\end{equation}
where $C > 0$ is some $(h, h')$-independent constant. We now use the inequality
\begin{equation}
\label{eq:diameter_decomposition}
\mathrm{diam}(0, \underline{x} + t \underline{y}_\torus) \le \mathrm{diam}(0, t \underline{y}_\torus) + \sum_{\ell = 1}^n \abs*{x_\ell}_L \le \mathrm{diam}(0, \underline{y}) + \sum_{\ell = 1}^n \abs*{x_\ell}_L.
\end{equation}
To prove it, suppose that the maximum distance between any two points among $0, x_1 + t (y_1)_\torus, \dots, x_n + t (y_n)_\torus$ is realized on the pair $\lbrace x_j + t (y_j)_\torus, x_k + t (y_k)_\torus \rbrace$. Then
\begin{align}
\label{eq:diameter_decomposition_intermediate}
\begin{split}
\mathrm{diam}(0, \underline{x} + t \underline{y}_\torus) = \abs*{x_j + t (y_j)_\torus - x_k - t (y_k)_\torus}_L
& \le \abs*{t(y_j)_\torus - t(y_k)_\torus}_L + \abs*{x_j - x_k}_L \\
& \le \abs*{t(y_j)_\torus - t(y_k)_\torus}_L + \sum_{\ell = 1}^n \abs*{x_\ell}_L.
\end{split}
\end{align}
Now, by exploiting the fact that 
\begin{align}
\begin{split}
\abs{ \, \sin \! \left[ \frac{\pi}{L} \left( \frac{tL}{2\pi} \sin \left( \frac{2\pi u}{L} \right) - \frac{tL}{2\pi} \sin \left( \frac{2\pi v}{L} \right) \right) \right] }
& = \abs{ \, \sin \! \left[ t \sin \left( \frac{\pi (u-v)}{L} \right) \cos \left( \frac{\pi (u+v)}{L} \right) \right] } \\
& \le \abs{\, \sin \left( \frac{\pi(u-v)}{L} \right) }
\end{split}
\end{align}
for every $u, v \in \R$ and for every $t \in [0, 1]$, we can write $\abs*{t(y_j)_\torus - t(y_k)_\torus}_L \le \abs{y_j - y_k}_L \le \mathrm{diam}(\underline{y}) \le \mathrm{diam}(0, \underline{y})$, where the last two inequalities follow from the definition of diameter. After plugging this into~\eqref{eq:diameter_decomposition_intermediate}, formula~\eqref{eq:diameter_decomposition} follows immediately; a similar argument applies if the maximum is realized on a pair of the form $\lbrace 0, x_\ell + t (y_\ell)_\torus \rbrace$.

Recalling that $\delta^{(n)}_N(\underline{x}) \equiv \delta^{(1)}_N(x_1) \cdots \delta^{(1)}_N(x_n)$, by using the bound~\eqref{eq:diameter_decomposition} inside~\eqref{eq:lm:renormalization_gain_quasi_final} we obtain
\begin{equation}
\label{eq:lm:renormalization_gain_penultimate}
\norm*{(\kren_{R, x_0}^0 \sing{W}{h})^{\underline{\alpha}}}_{h'} \le C' \left(\norm*{\delta^{(1)}_N}_{h'}\right)^n \int \dd{\underline{y}} \abs*{W(0, \underline{y}) \, \underline{y}^{\underline{\alpha}}_\torus} \, e^{c \sqrt{2^{h'}\mathrm{diam}(0, \underline{y})}}.
\end{equation}
We notice that $\norm*{\delta^{(1)}_N}_{h'} \le C''$ for some $(h', N)$-independent constant $C''$. In fact, $\delta^{(1)}_N$ is equal to the \emph{periodization} of the Schwartz function
\begin{equation}
\label{eq:infinite_volume_delta}
\phi_N(x) \equiv \int \frac{\dd[4]{k}}{(2\pi)^4} \, \chi_{N+1}(k) \, e^{i k \cdot x},
\end{equation}
in the sense that $\delta^{(1)}_N(x) = \sum_{n \in \Z^4} \phi_N(x + nL)$; therefore,
\begin{equation}
\label{eq:delta_estimate}
\begin{split}
\norm*{\delta^{(1)}_N}_{h'} \le \norm*{\delta^{(1)}_N}_{N+1} 
& = \int_{\mathcal{M}_L} \dd[4]{x} \abs{\sum_{n \in \Z^4} \phi_N(x + nL)} \, e^{c\sqrt{2^{N+1} \abs{x}_L}} \\
& \le \sum_{n \in \Z^4} \int_{\mathcal{M}_L} \dd[4]{x} \abs{\phi_N(x + nL)} \, e^{c\sqrt{2^{N+1} \abs{x}_L}} \\
& = \int_{\R^4} \dd[4]{x} \, \abs{\phi_N(x)} \, e^{c\sqrt{2^{N+1} \abs{x}}},
\end{split}
\end{equation}
where in the last line we performed the change of variables $x \mapsto x - nL$ and subsequently exploited the fact that $\abs{x-nL}_L = \abs{x}_L \le \abs{x}$ for every $n \in \Z^4$. It is evident from~\eqref{eq:infinite_volume_delta} that $\phi_N(x) = 2^{4(N+1)} \phi_0(2^{N+1} x)$, so
\begin{equation}
\norm*{\delta^{(1)}_N}_{h'} \le 2^{4(N+1)} \int_{\R^4} \dd[4]{x} \, \abs{\phi_0(2^{N+1} x)} \, e^{c\sqrt{2^{N+1} \abs{x}}}.
\end{equation}
The change of variables $x \mapsto 2^{-{N+1}} x$ then yields
\begin{equation}
\norm*{\delta^{(1)}_N}_{h'} \le 2^{4(N+1)} \cdot 2^{-4(N+1)} \int_{\R^4} \dd[4]{x} \, \abs{\phi_0(x)} \, e^{c\sqrt{\abs{x}}} = \int_{\R^4} \dd[4]{x} \, \abs{\phi_0(x)} \, e^{c\sqrt{\abs{x}}};
\end{equation}
finally, the fact that $\chi_0$ has Gevrey class $2$ implies that the last integral, which is manifestly independent from $h', N$, is finite (more details on this can be found in the proof of Lemma~\ref{lm:propagator_bounds}).

We are left to estimate the $h'$-weighted norm of the function $W(0, \underline{y}) \, \underline{y}^{\underline{\alpha}}_\torus$. Thanks to the bound $\xi \le 5e^{b \, \xi^{ \, 1/2}}/(9b^2)$ (holding for every $\xi, b > 0$), we have
\begin{multline}
\label{eq:lm:renormalization_gain_last_step}
\abs*{\underline{y}^{\underline{\alpha}}_\torus} \le \left[ \mathrm{diam}(0, \underline{y}) \right]^R \le \frac{e^{Rc(2^{h/2} - 2^{h'/2}) R^{-1} \sqrt{\mathrm{diam}(0, \underline{y})}}}{[c(2^{h/2} - 2^{h'/2})R^{-1}]^{2R} \cdot (9/5)^R} \, \le \\
\le 2^{-(h+1)R} \left(\frac{13R^2}{c^2} \right)^R e^{c(2^{h/2} - 2^{h'/2}) \sqrt{\mathrm{diam}(0, \underline{y})}}.
\end{multline}
The first inequality follows from the fact that, by construction, every component $(y_\ell)_\torus^\mu$ satisfies $\abs*{(y_\ell)_\torus^\mu} \le \abs*{y_\ell}_L \le \mathrm{diam}(0, \underline{y})$, whereas the last inequality is due to the elementary bound $(1 - 2^{(h' - h)/2})^{-2} \cdot (10/9) \le 13$. Formula~\eqref{eq:renormalization_gain_kernels} is obtained by plugging~\eqref{eq:lm:renormalization_gain_last_step} into~\eqref{eq:lm:renormalization_gain_penultimate}.
\end{proof}
\begin{rmk}
\label{rmk:constant}
We stress that the constant $(13R^2/c^2)^R$, which is potentially very large, depends on the shape of $\chi_0$. If $\chi_0$ is chosen so that its Gevrey class is close to $1$ and $c$ is sufficiently large, this constant becomes of order $1$. 
\end{rmk}

\begin{lm}[Propagator bounds]
\label{lm:propagator_bounds}
Let $h_0 \in \lbrace h^\star, \dots, N \rbrace$ and let $\mathscr{S}(\, \cdot \,) \equiv (\, \cdot \,) - (\, \cdot \,) \vert_{m_N = 0}$. Suppose that
\begin{equation}
\label{eq:lm:beta_bounds}
\abs*{\beta^s_h}, \, \abs*{2^{-h^\star} \beta^{m, s}_h} \le c \frac{\lambda^2 \cutoff^2}{M^2} 2^{h - N}, \qquad
\abs*{\mathscr{S} \beta^{J, s}_h}, \, \abs*{\mathscr{S} \beta^s_h} \le c \frac{\lambda^2 \cutoff^2}{M^2} 2^{h - N} 2^{h^\star - h}
\end{equation}
for every $h \in \lbrace h_0 + 1, \dots, N \rbrace$ and for some $h$-independent constant $c > 0$. If $c\lambda^2 \cutoff^2/M^2 < \epsilon$ for a sufficiently small $\epsilon \in (0, 1)$, there exists a constant $C_H > 0$ depending on $Z^s_N, m^s_N, \epsilon, \chi_0$ and two functions $\sing{g}{h}_{1, \ell}, \sing{g}{h}_{2, \ell}$ such that $\sing{g}{h}_{1, \ell} \star \sing{g}{h}_{2, \ell} \equiv \partial^\ell \sing{g}{h}$ (where $\star$ denotes the convolution operator) and
\begin{gather}
\begin{split}
\label{eq:propagator_weighted_norm}
\norm*{\partial^\ell \sing{g}{h}}_h \le C_H \cdot 2^{-h + \ell}, \qquad
\norm*{\mathscr{S} \partial^\ell \sing{g}{h}}_h \le C_H \cdot 2^{-h + \ell} \cdot 2^{h^\star - h},
\end{split} \\
\begin{split}
\label{eq:propagator_L_infty_norm}
\norm*{\sing{g}{h}_{1, \ell}}_{\infty} \cdot \norm*{\sing{g}{h}_{2, \ell}}_{\infty} \le C_H \cdot 2^{3h + \ell}, \qquad
\norm*{\mathscr{S} \sing{g}{h}_{1, \ell}}_{\infty} \cdot \norm*{\mathscr{S} \sing{g}{h}_{2, \ell}}_{\infty} \le C_H \cdot 2^{3h + \ell} \cdot 2^{h^\star - h}
\end{split}
\end{gather}
for every $\ell \colon \abs{\ell} \le 3$ and for every $h \in \lbrace h_0 + 1, \dots, N \rbrace$.
\end{lm}

\begin{proof}
We prove the statement for $\ell = 0$, as the other cases can be easily derived from this one. The function $\sing{g}{h}$ is equal to the periodization of
\begin{equation}
\label{eq:propagator_periodization}
\int \frac{\dd[4]{k}}{(2\pi)^4} \, \frac{f_h(k) \, e^{i k \cdot x}}{i \tilde{\gamma}^\mu_h(k) k_\mu + \tilde{\mathsf{m}}_h(k)} = 2^{3h} \int \frac{\dd[4]{k}}{(2\pi)^4} \, \frac{f_0(k) \, e^{i k \cdot 2^h x}}{i \tilde{\gamma}^\mu_h(2^h k) k_\mu + 2^{-h}\tilde{\mathsf{m}}_h(2^h k)} \equiv 2^{3h} \phi_h(2^h x).
\end{equation}
Now let $\ell, s \in \N$ with $\ell$ even, and choose some $x \in \R^4$ such that $x_\mu \ne 0$ for a certain component $x_\mu$. A repeated integration by parts gives
\begin{equation}
\label{eq:lm:integration_by_parts}
\frac{(2^h x_\mu)^s (c^2 2^h x_\mu)^{\ell/2}}{\ell!} \phi_h(2^h x) = \frac{i^{\ell/2s} \, c^\ell}{\ell!} \int \frac{\dd[4]{k}}{(2\pi)^4} \, \partial_{k_\mu}^{\ell/2 + s}[ f_0(k) \cdot u^{-1}(k) ] e^{i k \cdot 2^h x},
\end{equation}
where $u(k) \equiv i \tilde{\gamma}^\mu_h(2^h k) k_\mu + 2^{-h}\tilde{\mathsf{m}}_h(2^h k)$. The functions $\tilde{\gamma}^\mu_h(\, \cdot \,)$, $\tilde{\mathsf{m}}^\mu_h(\, \cdot \,)$ have Gevrey class $2$, because they are defined in terms of $\chi_0$. This, together with the bounds~\eqref{eq:lm:beta_bounds}, implies that $\norm*{\partial^j_{k_\mu} u(k)} \le c_2 C^j (j!)^2$ for every $k \in \mathrm{supp}(f_0)$, where the constants $c_2, C$ only depend on $Z^s_N, m^s_N, \chi_0, \epsilon$. Moreover, we have
\begin{equation}
\label{eq:lm:u-1}
u^{-1}(k) = [-i \tilde{\gamma}^\mu_{h^\star}(2^{h^\star} k) k_\mu + 2^{-h^\star}\tilde{\mathsf{m}}_{h^\star}(2^{h^\star} k)] \cdot [k^2 \tilde{Z}^+_{h^\star}(2^{h^\star} k) \cdot \tilde{Z}^-_{h^\star}(2^{h^\star}k) + 2^{-2h^\star}\tilde{\mathsf{m}}^2_{h^\star}(2^{h^\star} k)]^{-1}.
\end{equation} 
The first factor can be bounded with some scale-independent constant $c_3$. The second factor is a diagonal matrix, so its operator norm is bounded by the maximum among its diagonal entries, which is itself dominated by $\max_{k \in \supp(f_0)} \abs*{k^2 \tilde{Z}^+_h(2^h k) \cdot \tilde{Z}^-_h(2^h k)}^{-1} \le c_4$ for a suitable constant $c_4$. In summary, we have $\norm*{u^{-1}(k)} \le c_3c_4 \equiv c_1$ for all $k \in \supp(f_0)$. According to the above considerations, starting from the elementary identity
\begin{equation}
0 = \partial^j_{k_\mu} (u \cdot u^{-1}) \Rightarrow \partial^j_{k_\mu} (u^{-1}) = - u^{-1} \sum_{n = 1}^j \binom{j}{n} \, \partial^n_{k_\mu} u \cdot \partial^{j - n}_{k_\mu} (u^{-1})
\end{equation}
it is easy to prove that $\norm*{\partial^j_{k_\mu} (u^{-1})(k)} \le c_1 (c_1 c_2 C)^j (j!)^2$ by induction on $j$. Since $f_0$ has Gevrey class $2$, we also have $\norm*{\partial^j_{k_\mu} f_0(k)} \le C_1^{j+1} (j!)^2$ for a suitable constant $C_1$; therefore, the derivatives appearing on the right hand side of~\eqref{eq:lm:integration_by_parts} are bounded by $C_2^{\ell/2 + s + 1} ((\ell/2+s)!)^2$ for every $k \in \mathrm{supp}(f_0)$, where $C_2$ is another constant. As a result, we have
\begin{equation}
\label{eq:gevrey_estimate}
\norm{(2^h x_\mu)^s \cdot \frac{(c^2 2^h x_\mu)^{\ell/2}}{\ell!} \phi_h(2^h x)} \le  2^{3h} \frac{c^\ell C_3^{\ell/2 + s + 1}}{\ell!} \cdot ((s + \ell/2)!)^2.
\end{equation} 
where $C_3$ is another suitable constant. A similar procedure can be followed if $\ell$ is odd. Eventually, if the constant $c$ is properly chosen, the right hand side of~\eqref{eq:gevrey_estimate} can be summed over $\ell$. On the left hand side, this sum reconstructs the stretched exponential $e^{c \sqrt{2^h \abs*{x_\mu}}}$, so we conclude that for every $s \in \N$ there exists a constant $C_s$ such that $e^{c \sqrt{2^h \abs*{x_\mu}}}\norm*{\phi_h(2^h x)} \le (2^{3h} C_s)(2^h \abs*{x_\mu})^{-s}$. A trivial extension of this argument gives the bound
\begin{equation}
\label{eq:xi_bound}
\norm*{2^{3h} \phi_h(2^h x)} \le \frac{2^{3h} C_s}{1 + (2^h \abs*{x})^s} \, e^{-c \sqrt{2^h \abs*{x}}} \qquad \forall x \in \R^4, \forall s \in \N.
\end{equation}
Finally, the desired bound on the weighted norm of $\sing{g}{h}$ follows by combining~\eqref{eq:xi_bound} with the inequality
\begin{equation}
\label{eq:L1_bound}
\int_{\mathcal{M}_L} \dd[4]{x} \abs*{\sing{g}{h}(x)} \weight_h(x) \le C' \int_\R \dd[4]{x} \abs{\phi_h(x)} \, e^{c\sqrt{2^h \abs{x}}},
\end{equation}
which has already been discussed in~\eqref{eq:delta_estimate}.

For what concerns~\eqref{eq:propagator_L_infty_norm}, one writes $\phi_h = \phi_{h, 1} \star \phi_{h, 2}$, where
\begin{align}
\label{eq:propagator_splitting_bounds}
\begin{split}
\phi_{h, 1}(x) 
& \equiv \int \frac{\dd[4]{k}}{(2\pi)^4} \, \frac{ \sqrt{f_h(k)} \, e^{i k \cdot x}}{[i \tilde{\gamma}^\mu_h(k) k_\mu + \tilde{\mathsf{m}}_h(k)]^2}, \\
\quad \phi_{h, 2}(x) 
& \equiv \int \frac{\dd[4]{k}}{(2\pi)^4} \, \sqrt{f_h(k)} \, [i \tilde{\gamma}^\mu_h(k) k_\mu + \tilde{\mathsf{m}}_h(k)] \, e^{i k \cdot x}. 
\end{split}
\end{align}
The same construction presented above can be used to show that $\phi_{h, 1}$ and $\phi_{h, 2}$ satisfy~\eqref{eq:xi_bound} with $2^h$ and $2^{4h}$ in place of $2^{3h}$. The bounds~\eqref{eq:propagator_L_infty_norm} are then obtained by defining $\sing{g}{h}_1$ and $\sing{g}{h}_2$ as the periodizations of $\phi_{h, 1}, \phi_{h, 2}$.

We are left to deal with $\mathscr{S} \sing{g}{h}$ (or $\mathscr{S} \sing{g}{h}_1, \mathscr{S} \sing{g}{h}_2$, which are entirely similar). Recalling that $m^s_h \vert_{m_N = 0} = 0$, it is sufficient to write
\begin{equation}
\mathscr{S} \sing{\hat{g}}{h}(k) = (\sing{\hat{g}}{h} - \sing{\hat{g}}{h} \vert_{m_N = 0})(k) = f_h(k) \frac{[i k_\mu (\tilde{\gamma}^\mu_h \vert_{m_N = 0} - \tilde{\gamma}^\mu_h)(k) + \tilde{\mathsf{m}}_h(k)] \cdot [i k_\mu \tilde{\gamma}^\mu_h \vert_{m_N = 0}(k)]^{-1}}{i \tilde{\gamma}^\mu_h(k) k_\mu + \tilde{\mathsf{m}}_h(k)}
\end{equation}
and proceed as before. The only difference consists in the fact that the numerator carries either a $\abs*{Z^s_h - Z^s_h \vert_{m_N = 0}}$ factor coming from $ik_\mu(\tilde{\gamma}^\mu_h \vert_{m_N = 0} - \tilde{\gamma}^\mu_h)[i k_\mu \tilde{\gamma}^\mu_h \vert_{m_N = 0}(k)]^{-1}$ or a $\abs*{m^s_h}/\abs*{k}$ factor coming from $\tilde{\mathsf{m}}_h(k) [i k_\mu \tilde{\gamma}^\mu_h \vert_{m_N = 0}(k)]^{-1}$. In both cases, the bounds~\eqref{eq:lm:beta_bounds} together with the constraint $\abs*{k} \in [2^{h-1}, 2^{h+1}]$ imply that there is an additional $(c \lambda^2 \cutoff^2/M^2) \cdot 2^{h^\star - h}$ gain, in accordance with~\eqref{eq:propagator_weighted_norm},~\eqref{eq:propagator_L_infty_norm}.
\end{proof}

Lemma~\ref{lm:taylor_renormalization_gain} is the key to prove that the renormalization operator improves the na\"ive dimensional estimate of the kernel on which it acts. Before getting into the details of this, we need to introduce a few more notations. Given tree $\tau \in \trees_{h|N}$ together with a suitable choice of the sets $\underline{P}, \underline{\Delta}$, consider any nontrivial vertex $v \in \vrt(\tau)$. Then
\begin{itemize}
\item $D_v \equiv 4 - 3(\abs*{P_v} + n^\eta_v)/2 - n^J_v - \mathfrak{d}_v$ and $\bar{D}_v \equiv D_v + 2 n_v + n^\eta_v$, where are respectively equal to the number of $\eta, J, \lambda$ endpoints that are connected with $v$ by an increasing path of vertices and $\mathfrak{d}_v$ is the number of derivatives falling on the external fields of $v$ \emph{before} the action of the renormalization operator.
\item $h'_v$ denotes the scale of the first nontrivial vertex that precedes $v$. If such vertex does not exist, $h'_v \equiv h$.
\item If $v$ is either an endpoint or the first nonroot vertex, we let $R_v \equiv 0$. In any other case,
\begin{equation}
R_v \equiv
\begin{cases}
3	&	\quad \abs*{P_v} = 2, \mathfrak{d}_v = 0, n^J_v = 0, n^\eta_v = 0 \\
2 	&	\quad \abs*{P_v} = 2, \mathfrak{d}_v = 1, n^J_v = 0, n^\eta_v = 0 \\
2	&	\quad \abs*{P_v} = 2, \mathfrak{d}_v = 0, n^J_v = 1, n^\eta_v = 0 \\
0	&	\quad \text{otherwise}
\end{cases}
\end{equation}
\end{itemize}

\begin{theorem}[Bounds on the renormalized tree expansion]
\label{th:renormalization_bound}
There exists a pair of constants $C, c > 0$ such that, if $c \lambda^2 \cutoff^2/M^2$ is sufficiently small, then
\begin{align}
\begin{split}
\label{eq:boundedness_beta_functions}
\abs*{\beta^{J, s}_h}, \, \abs*{\beta^s_h}, \, \abs*{2^{-h^\star} \beta^{m, s}_h} & \le c \frac{\lambda^2 \cutoff^2}{M^2} 2^{h - N}
\end{split} \\
\begin{split}
\label{eq:boundedness_beta_functions_zero_mass}
\abs*{\mathscr{S} \beta^{J, s}_h}, \, \abs*{\mathscr{S} \beta^s_h} & \le c \frac{\lambda^2 \cutoff^2}{M^2} 2^{h^\star - h} \, 2^{h - N}
\end{split}
\end{align}
for all $h \in \lbrace h^\star, \dots, N \rbrace$, and
\begin{align}
\begin{split}
\label{eq:renormalization_bound}
\norm*{W_{\tau, \underline{P}, \underline{\Delta}}}_h & \le C^k \cdot \left( \frac{\lambda^2 \cutoff^2}{M^2} \right)^n 2^{hD} \, \prod_v^\circ 2^{(h_v - h'_v)(D_v - R_v)} \prod^\lambda_v 2^{2(h'_v - N)} \prod^\eta_v 2^{-h'_v},
\end{split} \\
\begin{split}
\label{eq:renormalization_bound_zero_mass}
\norm*{\mathscr{S} W_{\tau, \underline{P}, \underline{\Delta}}}_h & \le 2^{h^\star - h_{w_0}} C^k \cdot \left( \frac{\lambda^2 \cutoff^2}{M^2} \right)^n 2^{hD} \, \prod_v^\circ 2^{(h_v - h'_v)(D_v - R_v)} \prod^\lambda_v 2^{2(h'_v - N)} \prod^\eta_v 2^{-h'_v}
\end{split}
\end{align}
for every $\tau \in \trees_{h|N}'$; the bound~\eqref{eq:renormalization_bound} also holds for $h = h^\star - 1, \tau \in \trees_{h^\star - 1|N}$ (see Section~\ref{subsec:last_scale} below). The products $\prod_v^\circ, \prod^\lambda, \prod^\eta$ respectively run over all the nontrivial vertices, the $\lambda$ endpoints and the $\eta$ endpoints of $\tau$, $D$ is equal to $D_{w_0}$ ($w_0$ being the first nonroot vertex of $\tau$) and $k, n$ denote the number of endpoints and the number of $\lambda$ endpoints of $\tau$. Finally, the $\mathscr{S}$ operator acts as $\mathscr{S}(\, \cdot \,) \equiv (\, \cdot \,) - (\, \cdot \,) \vert_{m_N = 0}$.
\end{theorem}

\begin{proof}
We prove the theorem by induction on the scale $h$ that appears in~\eqref{eq:boundedness_beta_functions},~\eqref{eq:boundedness_beta_functions_zero_mass},%
~\eqref{eq:renormalization_bound},~\eqref{eq:renormalization_bound_zero_mass}. If~\eqref{eq:boundedness_beta_functions},~\eqref{eq:boundedness_beta_functions_zero_mass} hold up to scale $h+2$ and $c \lambda^2 \cutoff^2/M^2 < \epsilon$ for some small $\epsilon$, Lemma~\ref{lm:propagator_bounds} ensures that the propagator norms up to scale $h+1$ are bounded as in~\eqref{eq:propagator_weighted_norm},~\eqref{eq:propagator_L_infty_norm} with a constant $C_H$ that only depends on $Z^s_N, m^s_N, \epsilon$ and on the shape of the function $\chi_0$. Also, thanks to the same inductive assumption, there exists a constant $C_E$ that only depends on $Z^{J, s}_N, \epsilon, \chi_0$ such that
\begin{alignat}{3}
\label{eq:th:bounds_endpoints}
& \norm*{\sing{W}{j}_{(J)}}_j \le C_E,
&& \quad \norm*{\sing{W}{N}_{(\lambda)}}_N \le C_E \frac{\lambda^2}{M^2},
&& \quad \norm*{\sing{W}{N}_{(\eta)}}_N \le C_E, \\
\label{eq:th:bounds_endpoints_zero_mass}
& \norm*{\mathscr{S} \sing{W}{j}_{(J)}}_j \le C_E 2^{h^\star - j},
&& \quad \norm*{\mathscr{S} \sing{W}{N}_{(\lambda)}}_N \le C_E 2^{h^\star - N} \frac{\lambda^2}{M^2},
&& \quad \norm*{\mathscr{S} \sing{W}{N}_{(\eta)}}_N \le C_E 2^{h^\star - N}
\end{alignat}
for every $j \ge h+1$.

For technical reasons, we prove the stronger bounds
\begin{subequations}
\begin{gather}
\begin{split}
\label{eq:th:r_bound}
\norm*{W_{\tau, \underline{P}, \underline{\Delta}}}_h \le \frac{C_E^k \lambda^{2n}}{M^{2n}} 2^{h\bar{D}_{w_0}} \, \prod_v^\circ 2^{(h_v - h'_v)(\bar{D}_v - R_v)} \cdot \mathcal{C}_v \mathcal{C}^{\textup{ren}}_v
\end{split} \\
\begin{split}
\label{eq:th:p_bound}
\norm*{\mathscr{S} W_{\tau, \underline{P}, \underline{\Delta}}}_h \le 2^{h^\star - h_{w_0}} \frac{C_E^k \lambda^{2n}}{M^{2n}} 2^{h\bar{D}_{w_0}} \, \prod_v^\circ 2^{(h_v - h'_v)(\bar{D}_v - R_v)} \cdot \mathcal{C}_v \mathcal{C}^{\textup{ren}}_v
\end{split}
\end{gather}
\end{subequations}
in place of~\eqref{eq:renormalization_bound},~\eqref{eq:renormalization_bound_zero_mass}, where
\begin{align*}
\mathcal{C}_v & \equiv (6C_H)^{(\sum_{w \succ v} \abs*{P_w} - \abs{P_v})/2} 2^{s_v}, \qquad \quad
\mathcal{C}^{\textup{ren}}_v \equiv
\begin{cases}
2 C_\ren  	&	\textup{if } \ren \textup{ acts nontrivially on } v	\\
1			&	\textup{otherwise}
\end{cases}
\end{align*}
Here, $C_\ren$ is the same constant introduced in Lemma~\ref{lm:taylor_renormalization_gain} and $C_H, C_E$ are the same constants introduced at the beginning of the proof.

Given any $h \in \lbrace h^\star-1, \dots, N - 2 \rbrace$, consider the induction step
\begin{equation}
\label{eq:induction_step}
\begin{cases}
\eqref{eq:boundedness_beta_functions},\eqref{eq:boundedness_beta_functions_zero_mass} \text{ hold up to scale } h+2 \\
\eqref{eq:th:r_bound},\eqref{eq:th:p_bound} \text{ hold up to scale } h+1 \\
c \lambda^2 \cutoff^2/M^2 < \epsilon
\end{cases}
\Rightarrow
\begin{cases}
\eqref{eq:boundedness_beta_functions},\eqref{eq:boundedness_beta_functions_zero_mass} \text{ hold up to scale } h + 1 \\
\eqref{eq:th:r_bound},\eqref{eq:th:p_bound} \text{ hold up to scale } h \\
c \lambda^2 \cutoff^2/M^2 < \epsilon
\end{cases}
\end{equation}
where $\epsilon \in (0, 1)$ and $c$ is a positive constant whose value will be fixed later. From now on, we inductively assume that the left-hand side of this relation is true for every scale $j \ge h$. Our goal is to prove that the right-hand side holds as well.

Any unnecessary spinor or spacetime index will be suppressed in order to lighten the notation. Such indices assume a \emph{finite} set of values for each field variable, so summing over them yields an unharmful constant that can always be reabsorbed by properly rescaling $C_H, C_E, C_\ren$. Moreover, we restrict our attention to trees that contain at least one nontrivial vertex: the reader may easily check that, due to~\eqref{eq:th:bounds_endpoints} and~\eqref{eq:th:bounds_endpoints_zero_mass}, both~\eqref{eq:th:r_bound} and~\eqref{eq:th:p_bound} are manifestly true if every vertex of $\tau$ is trivial. Finally, since trivial vertices that are different from the first nonroot vertex will not play any active role in the following estimates, we can safely erase them.

\paragraph{Proof of~\eqref{eq:th:r_bound} on scale $h$.} Let $\tau \in \trees'_{h|N}$ and let $w_0 \in \vrt(\tau)$ be its first nonroot vertex. The norm of $W_{\tau, \underline{P}, \underline{\Delta}}$ can be bounded as 
\begin{multline}
\label{eq:th:norm_bound_initial}
\norm*{W_{\tau, \underline{P}, \underline{\Delta}}}_h \le \frac{1}{s_{w_0}!} \int_{x_0 = 0} \dd{\underline{y}(U_{w_0})} \prod_{v \succ w_0} \dd{\underline{x}(Q_v)} \sum_{\lbrace a_v \rbrace} \biggl[ \, \abs*{\TExp{h_{w_0}}( \lbrace \Psi(Q_v, \Delta_{v, \underline{a}}) \rbrace_{v \succ w_0} )} \, \times \\
\times \prod_{v \succ w_0} \abs{\kren_v^{a_v} W_{\tau_v, \underline{P}, \underline{\Delta}}(\underline{x}(P_v), \underline{y}(U_v))} \, \weight_h(\underline{x}(P_{w_0}), \underline{y}(U_{w_0})) \biggr],
\end{multline}
where the $x_0 = 0$ subscript reminds that one spacetime point among $\lbrace \underline{x}(Q_v) \rbrace_{v \succ w_0}, \underline{y}(U_{w_0})$ must be kept fixed to zero. The symbol $\sum_{\lbrace a_v \rbrace}$ denotes a sum over $a_v = 0, 1$ for every $v \succ w_0$ constrained by various requirements (for instance, the $a_v$'s must be compatible with the choice of $\Delta_{w_0}$) and $\Delta_{v, \underline{a}}$ is the set of derivative indices falling on the elements of $Q_v$ after a suitable choice of the quantities $\lbrace a_v \rbrace_{v \succ w_0} \equiv \underline{a}$. Starting from~\eqref{eq:th:norm_bound_initial}, we shall find an explicit bound on the truncated expectation value and to estimate the norms of the renormalized kernels $\lbrace \kren_v^{a_v} W_{\tau_v, \underline{P}, \underline{\Delta}} \rbrace_{v \succ w_0}$ with a combined usage of Lemma~\ref{lm:taylor_renormalization_gain} and the induction hypoteses. 

If $w_0$ is trivial, the truncated expectation value appearing in~\eqref{eq:th:norm_bound_initial} degenerates into $\TExp{h_{w_0}}(\varnothing) = 1$. However, it is easy to see that all the following results still hold in this degenerate case, so we will assume that $w_0$ is not trivial for sake of simplicity. Thanks to the Brydges-Battle-Federbush formula~\cite{Battle_1984, Gawedzki_1985, Lesniewski_1987}, the truncated expectation value can be expanded as
\begin{equation}
\label{eq:BBF}
\TExp{h_{w_0}}( \lbrace \Psi(Q_v, \Delta_{v, \underline{a}}) \rbrace_{v \succ w_0} ) = \sum_T \prod_{\ell \in \edg(T)} \sing{\bar{g}}{h_{w_0}}_\ell \int \dd{\mu}_T(\underline{t}) \det \bar{G}^{h_{w_0}}_T(\underline{t}),
\end{equation}
where the sum $\sum_T$ runs over all the trees anchored to the sets $\lbrace Q_v \rbrace_{v \succ w_0}$, $\dd{\mu}_T(\underline{t})$ is a probability measure on $[0, 1]^{p^2}$ (with $p \equiv \sum_{v \succ w_0} (\abs*{P_v} - \abs*{P_{w_0}})/2 - s_{w_0} + 1$) supported on the set of all the vectors $\underline{t} \equiv (t_{ab})_{a, b = 1, \dots, p}$ such that $t_{ab} = \mathbf{u}_a \cdot \mathbf{u}_b$ for some pair of unit vectors $\mathbf{u}_a, \mathbf{u}_b \in \R^p$. Finally, the $p \times p$ matrix $\bar{G}^{h_{w_0}}_T(\underline{t})$ is defined as
\begin{equation}
[\bar{G}^{h_{w_0}}_T(\underline{t})]_{ab} = t_{ab} \cdot (\sing{\bar{g}}{h_{w_0}})_{\alpha_a \beta_b}(x_a - x_b)
\end{equation}
and a bar over $\sing{\bar{g}}{h_{w_0}}_\ell, \bar{G}^{h_{w_0}}_T(\underline{t})$ reminds that some of the propagators may carry up to three derivatives. A complete proof of~\eqref{eq:BBF} can be found in~\cite[Appendix A.3.2]{Gentile_2001}). Crucially, $\bar{G}^{h_{w_0}}_T(\underline{t})$ is a \emph{Gram matrix} within the support of $\dd{\mu}_T(\underline{t})$, because $\partial^\ell \sing{g}{h_{w_0}}$ is equal to a convolution of the form $\sing{g}{h_{w_0}}_{1, \ell} \star \sing{g}{h_{w_0}}_{2, \ell}$ (see Lemma~\ref{lm:propagator_bounds}) and $t_{ab} = \mathbf{u}_a \cdot \mathbf{u}_b$; therefore, after combining the bounds~\eqref{eq:propagator_weighted_norm},~\eqref{eq:propagator_L_infty_norm} with Gram-Hadamard's inequality~\cite[Theorem A.1]{Gentile_2001}, we have
\begin{equation}
\label{eq:gram_hadamard}
\abs{\int \dd{\mu}_T(\underline{t}) \det \bar{G}^{h_{w_0}}_T(\underline{t})} \le  C_H^p \cdot 2^{h_{w_0}(3p + \delta'_{w_0, \underline{a}, T})},
\end{equation}
where $\delta'_{w_0, \underline{a}, T}$ is the number of derivatives that do not fall along the tree $T$. The constant $C_H$ appears because we know by hypotesis that the bounds~\eqref{eq:boundedness_beta_functions},~\eqref{eq:boundedness_beta_functions_zero_mass} hold up to scale $h+2$ and $c \lambda^2 \cutoff^2/M^2 < \epsilon$, so Lemma~\ref{lm:propagator_bounds} applies up to scale $h+1$.

For every tree $T$ anchored on $\lbrace Q_v \rbrace_{v \succ w_0}$, the basic properties of the diameter imply that
\begin{equation}
\weight_h(\underline{x}(P_{w_0}), \underline{y}(U_{w_0})) \le \prod_{v \succ w_0} \weight_h(\underline{x}(P_v), \underline{y}(U_v)) \prod_{\ell \in \edg(T)} \weight_h(\ell),
\end{equation}
so~\eqref{eq:th:norm_bound_initial} becomes
\begin{multline}
\label{eq:th:norm_estimate_BBF}
\norm*{W_{\tau, \underline{P}, \underline{\Delta}}}_h \le C_H^p \cdot 2^{h_{w_0}(3p + \delta'_{w_0, \underline{a}, T})} \, \sum_T \frac{1}{s_{w_0}!} \int_{x_0 = 0} \dd{\underline{y}(U_{w_0})} \prod_{v \succ w_0} \dd{\underline{x}(Q_v)} \, \times \\
\times \sum_{\lbrace a_v \rbrace} \Biggl( \, \prod_{\ell \in \edg(T)} \abs*{\sing{\bar{g}}{h_{w_0}}_\ell \weight_h(\ell)} \prod_{v \succ w_0} \abs{\kren_v^{a_v} W_{\tau_v, \underline{P}, \underline{\Delta}}(\underline{x}(P_v), \underline{y}(U_v))} \weight_h(\underline{x}(P_v), \underline{y}(U_v)) \Biggr).
\end{multline}
Now consider a leaf of the anchored tree $T$, consisting of some kernel $\kren_v^{a_v} W_{\tau_v, \underline{P}, \underline{\Delta}}$ together with the corresponding weight function, and integrate over all its arguments except for the spacetime point (say, $z$) that connects $\kren_v^{a_v} W_{\tau_v, \underline{P}, \underline{\Delta}}$ with some propagator lying along $T$. This integral produces a $\norm*{\kren_v^{a_v} W_{\tau_v, \underline{P}, \underline{\Delta}}}_h$ factor. When the integral over $z$ is performed, we instead obtain a $\norm*{\sing{\bar{g}}{h_{w_0}}}_h$ factor. After repeating this construction with the remaining leaves of $T$, we are ultimately left with a $\norm*{\kren_v^{a_v} W_{\tau_v, \underline{P}, \underline{\Delta}}}_h$ factor for each $v \succ w_0$ and a $\norm*{\sing{\bar{g}}{h_{w_0}}}_h$ factor for each $\ell \in \vrt(T)$. Since one spacetime point is kept fixed to zero, no overall $L^4$ factors appear. If the weighted norms of the propagators lying along $T$ are bounded as in~\eqref{eq:propagator_weighted_norm}, we eventually obtain
\begin{equation}
\label{eq:th:bound_after_BBF}
\norm*{W_{\tau, \underline{P}, \underline{\Delta}}}_h \le C_H^{(\sum_{v \succ w_0} \abs*{P_v} - \abs*{P_{w_0}})/2} \cdot \left( \sum_T \frac{1}{s_{w_0}!} \right) \sum_{\lbrace a_v \rbrace} \prod_{v \succ w_0} 2^{h_{w_0}\mathscr{D}_{w_0, \underline{a}}} \norm*{\kren_v^{a_v} W_{\tau_v, \underline{P}, \underline{\Delta}}}_h,
\end{equation}
where $\mathscr{D}_{w_0, \underline{a}} \equiv -4(s_{w_0} - 1) + 3(\sum_{v \succ w_0} \abs*{P_v} - \abs*{P_{w_0}})/2 + \delta_{w_0, \underline{a}}$ and $\delta_{w_0,\underline{a}}$ is the number of derivatives falling on the internal propagators of $w_0$. The number of trees anchored to the sets $\lbrace Q_v \rbrace_{v \succ w_0}$ is bounded by $2^{(\sum_{v \succ w_0} \abs*{P_v} - \abs*{P_{w_0}})/2} s_{w_0}!$~\cite[Lemma 2.4]{Mastropietro_Nonpert_Renorm}; since $(2C_H)^{(\sum_{v \succ w_0} \abs*{P_v} - \abs*{P_{w_0}})/2} \le \mathcal{C}_{w_0}$ and $\mathcal{C}^{\textup{ren}}_{w_0} = 1$, we get
\begin{equation}
\label{eq:th:bound_after_BBF_refined}
\norm*{W_{\tau, \underline{P}, \underline{\Delta}}}_h \le (\mathcal{C}_{w_0} \mathcal{C}^{\textup{ren}}_{w_0}) \sum_{\lbrace a_v \rbrace} \prod_{v \succ w_0} 2^{h_{w_0}\mathscr{D}_{w_0, \underline{a}}} \norm*{\kren_v^{a_v} W_{\tau_v, \underline{P}, \underline{\Delta}}}_h.
\end{equation}
By means of Lemma~\ref{lm:taylor_renormalization_gain}, the norms of $\kren_v^{a_v} W_{\tau_v, \underline{P}, \underline{\Delta}}$ can be estimated in terms of the norms of $W_{\tau_v, \underline{P}, \underline{\Delta}}$ and $\mathscr{S} W_{\tau_v, \underline{P}, \underline{\Delta}}$; afterwards, the latter can be bounded using the induction hypoteses~\eqref{eq:th:r_bound},~\eqref{eq:th:p_bound}. The result is
\begin{multline}
\label{eq:th:bound_after_BBF_quasi_final}
\norm*{W_{\tau, \underline{P}, \underline{\Delta}}}_h \le \frac{C^k_E \lambda^{2n}}{M^{2n}} (\mathcal{C}_{w_0} \mathcal{C}_{w_0}^{\textup{ren}}) \sum_{\lbrace a_v \rbrace} \prod_{v \succ w_0} 2^{h_{w_0}\mathscr{D}_{w_0, \underline{a}}} \, (C_\ren^v \mathcal{C}_v) \, 2^{h_v (\bar{D}_v - r^{a_v}_v)} 2^{a_v(h^\star - h_v)} \, \times \\
\times \prod_{z > v}^\circ 2^{(h_z - h'_z)(\bar{D}_z - R_z)} \mathcal{C}_z \mathcal{C}^{\textup{ren}}_z,
\end{multline}
where $r_v^{a_v} = R_v - a_v$ is the number of derivatives produced by the renormalization operator on $v$ and $C_\ren^v$ is either equal to $1$ or to $C_\ren$ depending on whether $\ren$ acts identically on $v$ or not. By construction, $\delta_{w_0, \underline{a}}$ is obtained by taking the number of derivatives falling on the external fields of its children vertices ($\sum_{v \succ w_0} \mathfrak{d}_v$), \emph{plus} the number of extra derivatives produced by the action of $\ren$ on them ($\sum_{v \succ w_0} r^{a_v}_v$), \emph{minus} the number of derivatives falling on the external fields of $w_0$ ($\mathfrak{d}_{w_0}$); synthetically, $\delta_{w_0, \underline{a}} = \sum_{v \succ w_0} (\mathfrak{d}_v + r^{a_v}_v) - \mathfrak{d}_{w_0}$. This easily yields
\begin{equation}
\label{eq:th:D_w_0_sum}
\mathscr{D}_{w_0, \underline{a}} + \sum_{v \succ w_0} (\bar{D}_v - r^{a_v}_v) = \bar{D}_{w_0}.
\end{equation}
Thanks to~\eqref{eq:th:D_w_0_sum} and recalling that $h'_v = h_{w_0} \, \forall v \colon v \succ w_0$, we can reorganize the right hand side of~\eqref{eq:th:bound_after_BBF_quasi_final} as
\begin{multline}
\label{eq:th:bound_quasi_final}
\norm*{W_{\tau, \underline{P}, \underline{\Delta}}}_h \le \frac{C_E^k \lambda^{2n}}{M^{2n}} (\mathcal{C}_{w_0} \mathcal{C}_{w_0}^{\textup{ren}}) \, 2^{h \bar{D}_{w_0}} 2^{(h_{w_0} - h) \bar{D}_{w_0}} \, \times \\
\times \sum_{\lbrace a_v \rbrace} \prod_{v \succ w_0} (C_\ren^v \mathcal{C}_v) \, 2^{(h_v - h'_v)(\bar{D}_v - r_v^{a_v})} 2^{a_v(h^\star - h_v)} \prod_{z > v}^\circ 2^{(h_z - h'_z)(\bar{D}_z - R_z)} \mathcal{C}_z \mathcal{C}^{\textup{ren}}_z.
\end{multline}
Finally,~\eqref{eq:th:r_bound} follows from the fact that $2^{(h_v - h'_v)(\bar{D}_v - r_v^{a_v})} 2^{a_v(h^\star - h_v)} \le 2^{(h_v - h'_v)(\bar{D}_v - R_v)}$ and $\sum_{a_v} C^v_\ren \le \mathcal{C}^{\textup{ren}}_v$. $\square$

\paragraph{Proof of~\eqref{eq:th:p_bound} on scale $h$.} Let $\tau \in \trees'_{h|N}$ as before. We start by writing
\begin{multline}
\label{eq:P_action_explicit}
\mathscr{S} W_{\tau, \underline{P}, \underline{\Delta}} = \int \prod_{v \succ w_0} \dd{\underline{x}(Q_v)} \, \Biggl[\mathscr{S} \left( \TExp{h_{w_0}}(\, \cdots) \right) \prod_{v \succ w_0} \kren_v W_{\tau_v, \underline{P}, \underline{\Delta}}(\underline{x}(P_v), \underline{y}(U_v)) \, + \\
+ \TExp{h_{w_0}}(\, \cdots) \eval_{m_N = 0} \, \mathscr{S} \left( \prod_{v \succ w_0} \kren_v W_{\tau_v, \underline{P}, \underline{\Delta}}(\underline{x}(P_v), \underline{y}(U_v)) \right) \Biggr].
\end{multline}
A careful analysis (see for instance~\cite[Appendix A4.3]{Giuliani_2005}) shows that the action of $\mathscr{S}$ converts $\TExp{h_{w_0}}(\, \cdots)$ into a sum of at most $(\sum_{v \succ w_0} \abs*{P_v} - \abs*{P_{w_0}})^2/4$ terms, each of which has the same structure of $\TExp{h_{w_0}}$ except for containing one $\mathscr{S} \sing{g}{h_{w_0}}$ propagator. According to Lemma~\ref{lm:propagator_bounds}, the occurrence of $\mathscr{S} \sing{g}{h_{w_0}}$ yields an extra $2^{h^\star - h_{w_0}}$ gain, so the first row of~\eqref{eq:P_action_explicit} satisfies the bound~\eqref{eq:th:bound_after_BBF} times an additional $2^{h^\star - h_{w_0}} (\sum_{v \succ w_0} \abs*{P_v} - \abs*{P_{w_0}})^2/4$ factor.

The second row of~\eqref{eq:P_action_explicit} contains a truncated expectation value evaluated at $m_N = 0$ multiplied by a sum of $s \equiv s_{w_0}$ terms of the form $(\kren_{v_1}^{a_{v_1}} W_{\tau_{v_1}}) \cdots (\mathscr{S} \kren_u^{a_u} W_{\tau_u}) \cdots (\kren_{v_s}^{a_{v_s}} W_{\tau_{v_s}})$, where some of the factors $\kren_{v_1}^{a_{v_1}} W_{\tau_{v_1}}, \dots, \kren_{v_s}^{a_{v_s}} W_{\tau_{v_s}}$ may be evaluated at $m_N = 0$ and the indices $a_{v_1}, \dots, a_{v_s}$ can be either equal to $0$ or to $1$. Lemma~\ref{lm:propagator_bounds} can be generalized with little effort in order to deal with the zero-mass single scale propagator $\sing{g}{h} \vert_{m_N = 0}$, so the truncated expectation value evaluated at $m_N = 0$ satisfies the estimates discussed during the proof of~\eqref{eq:th:r_bound}. 

The kernels that are \emph{not} evaluated at $m_N = 0$ satisfy the estimates presented in~\eqref{eq:th:bound_after_BBF_quasi_final}, while
\begin{equation}
\kren_v^{a_v} W_{\tau_v} \vert_{m_N = 0} =
\begin{cases}
\kren_v^0 W_{\tau_v} - \mathscr{S} \kren_v^0 W_{\tau_v} & \qquad \text{if } a_v = 0 \\
0														& \qquad \text{if } a_v = 1
\end{cases}
\end{equation}
It follows from Lemma~\ref{lm:taylor_renormalization_gain} and from the induction hypoteses~\eqref{eq:th:r_bound},~\eqref{eq:th:p_bound} that
\begin{align*}
\norm*{\kren_v^{a_v} W_{\tau_v} \vert_{m_N = 0}}_h 
& \le \id_{a_v = 0} (1 + 2^{h^\star - h_v}) \frac{C_E^k \lambda^{2n}}{M^{2n}} 2^{h_v(\bar{D}_v - R_v)} \, (\mathcal{C}_v C^v_\ren) \! \prod_{z > v}^\circ 2^{(h_z - h'_z)(\bar{D}_z - R_z)} \mathcal{C}_z \mathcal{C}^{\textup{ren}}_z \\
& \le \id_{a_v = 0} \cdot 2 \cdot \frac{C_E^k \lambda^{2n}}{M^{2n}} 2^{h_v(\bar{D}_v - R_v)} \, (\mathcal{C}_v C^v_\ren) \prod_{z > v}^\circ 2^{(h_z - h'_z)(\bar{D}_z - R_z)}\mathcal{C}_z \mathcal{C}^{\textup{ren}}_z.
\end{align*}
Since $\sum_{a_v} 2C_\ren^v \id_{a_v = 0} \le \mathcal{C}^\textup{ren}_v$, the factor $\norm*{\kren^{a_v}_v W_{\tau_v} \vert_{m_N = 0}}_h$ contributes to $\norm*{\mathscr{S} W_{\tau, \underline{P}, \underline{\Delta}}}$ as it were equal to $\norm*{\kren^{a_v}_v W_{\tau_v}}_h$.

It remains to analyze the contribution coming from $\mathscr{S} \kren^{a_u}_u W_{\tau_u}$. The key observation is that $\mathscr{S}^2 = \mathscr{S} \Rightarrow \mathscr{S} \kren^1 = \kren^1$, so
\begin{equation}
\mathscr{S} \kren^{a_u}_u W_{\tau_u} =
\begin{cases}
\mathscr{S} \kren^0_u W_{\tau_u} 											& \qquad \text{if } a_u = 0 \\
\kren^1_u W_{\tau_u} = - \mathscr{S} \kloc_{R_u - 1} W_{\tau_u}		& \qquad \text{if } a_u = 1
\end{cases}
\end{equation}
and due to Lemma~\ref{lm:taylor_renormalization_gain} and the induction hypoteses, it is
\begin{align}
\label{eq:th:zero_mass_quasi_final}
\begin{split}
\norm*{\mathscr{S} \kren^{a_u}_u W_{\tau_u}}_h 
& \le \frac{C_E^k \lambda^{2n}}{M^{2n}} 2^{(h^\star - h_u)} 2^{h_u(\bar{D}_u - r^{a_u}_u)} (\mathcal{C}_u C_\ren^u) \prod_{z > u}^\circ 2^{(h_z - h'_z)(\bar{D}_z - R_z)} \cdot (\mathcal{C}_z \mathcal{C}^\textup{ren}_z) \\
& = 2^{h^\star - h_{w_0}} \frac{C_E^k \lambda^{2n}}{M^{2n}} 2^{(h'_u - h_u)} 2^{h_u(\bar{D}_u - r^{a_u}_u)} (\mathcal{C}_u C_\ren^u) \prod_{z > u}^\circ 2^{(h_z - h'_z)(\bar{D}_z - R_z)} \cdot (\mathcal{C}_z \mathcal{C}^\textup{ren}_z).
\end{split}
\end{align}
We recognize that~\eqref{eq:th:zero_mass_quasi_final} has the same structure of the factors appearing in~\eqref{eq:th:bound_after_BBF_quasi_final}, except for the presence of an extra $2^{h^\star - h_{w_0}}$ gain and a $2^{h'_u - h_u}$ factor in place of $2^{a_u(h^\star - h_u)}$. The same analysis performed after formula~\eqref{eq:th:bound_after_BBF_quasi_final} shows that $\mathscr{S} \kren_u^{a_u} W_{\tau_u}$ contributes to $\norm*{\mathscr{S} W_{\tau, \underline{P}, \underline{\Delta}}}_h$ as it were equal to $2^{h^\star - h_{w_0}} \kren^{a_u}_u W_{\tau_u}$.

In summary, the above considerations imply that $\norm*{\mathscr{S} W_{\tau, \underline{P}, \underline{\Delta}}}_h$ satisfies the same bound as the right hand side of~\eqref{eq:th:bound_after_BBF} multiplied by an extra $2^{h^\star - h_{w_0}}$ gain (due to the presence of the $\mathscr{S}$ operator) and an extra $s_{w_0} + (\sum_{v \succ w_0} \abs*{P_v} - \abs*{P_{w_0}})^2/4$ factor (which takes into account the total number of terms arising from the decomposition~\eqref{eq:P_action_explicit}). All these factors can be reabsorbed inside the constant $\mathcal{C}_{w_0}$, namely
\begin{equation}
(2C_H)^{(\sum_{v \succ w_0} \abs*{P_v} - \abs*{P_{w_0}})/2} \left[ s_{w_0} + \frac{1}{4} \left( \sum_{v \succ w_0} \abs*{P_v} - \abs*{P_{w_0}}  \right)^2 \right] \le \mathcal{C}_{w_0} = \mathcal{C}_{w_0} \mathcal{C}_{w_0}^{\textup{ren}},
\end{equation}
so formula~\eqref{eq:th:p_bound} holds. $\square$

\paragraph{The bounds~\eqref{eq:th:r_bound},~\eqref{eq:th:p_bound} imply~\eqref{eq:renormalization_bound},~\eqref{eq:renormalization_bound_zero_mass}.} Let $b$ be the number of nontrivial vertices occurring in $\tau \in \trees'_{h|N}$. Then
\begin{equation}
\prod^\circ_v \mathcal{C}_v \mathcal{C}^{\textup{ren}}_v \le (\max \lbrace 1, 2 C_\ren \rbrace)^b \cdot (6C_H)^{\sum_v^\circ (\sum_{w \succ v} \abs* {P_w} - \abs*{P_v})} \cdot 2^{\sum^\circ_v s_v}.
\end{equation}
It is easy to see that $\sum^\circ_v s_v = b + k$ and $\sum_v^\circ (\sum_{w \succ v} \abs* {P_w} - \abs*{P_v}) \le \sum^E_v \abs*{P_v}$, where $\sum^E_v$ runs over the endpoints of $\tau$. Each endpoint can have at most four external fields, so $\sum^E_v \abs*{P_v} \le 4k$. Since every nontrivial vertex cannot contain less than $2$ internal fields, $b$ cannot exceed $\sum^E_v \abs*{P_v}/2 \le 2k$. By putting all together, we can write $C_E^k \prod^\circ_v \mathcal{C}_v = C^k$ for a suitable ``universal'' constant $C$ that only depends on $Z^{J, s}_N, Z^s_N, m^s_N, \epsilon, \chi_0$. Finally, the product~\eqref{eq:th:r_bound} can be recast in terms of $D_v$ by applying the basic identity $\sum_v^\circ (h_v - h'_v) n^{\lambda, \eta}_v = -h n^{\lambda, \eta}_{w_0} + \sum_v^{\lambda, \eta} h'_v$ (see for instance~\cite[Section III]{Mastropietro_2024}): the result is precisely~\eqref{eq:renormalization_bound}. A similar manipulation can be done with~\eqref{eq:th:p_bound} as well.

\paragraph{Proof of~\eqref{eq:boundedness_beta_functions},~\eqref{eq:boundedness_beta_functions_zero_mass} on scale $h+1$.} Based on the definition of $\beta^{J, s}_{h+1}$, it is easy to see that
\begin{equation}
\abs*{\beta^{J, s}_{h+1}} \le \sum_\tau \sum_{\underline{P}, \underline{\Delta}} \norm*{W_{\tau, \underline{P}, \underline{\Delta}}}_{h+1},
\end{equation}
where the sum $\sum_\tau \sum_{\underline{P}, \underline{\Delta}}$ is constrained by the fact that $\tau \in \trees_{h|N}'$ is nontrivial and its first nonroot vertex is nontrivial; $\tau$ has one $J$ endpoint, no $\eta$ endpoints and at least one $\lambda$ endpoint; finally, $\underline{P}, \underline{\Delta}$ are chosen so that $W_{\tau, \underline{P}, \underline{\Delta}}$ has two external fermionic legs without any derivatives falling on them. We have just shown that~\eqref{eq:renormalization_bound} holds for every $\tau \in \trees'_{h|N}$, so
\begin{equation}
\label{eq:th:beta_function_trees_bound}
\abs*{\beta^{J, s}_{h+1}} \le \sum_\tau \sum_{\underline{P}, \underline{\Delta}} C^k \frac{\lambda^{2n} \cutoff^{2n}}{M^{2n}} \prod^\circ_v 2^{(h_v - h'_v)(D_v - R_v)} \prod_v^\lambda 2^{2(h'_v - N)}.
\end{equation}
This sum can be estimated by means of the \emph{short memory property}~\cite{Mastropietro_2007, Giuliani_Mastropietro_Porta_2019}. Let $w_r \in \vrt(\tau)$ be a nontrivial vertex that precedes some $\lambda$ endpoint and let $\mathscr{P} \subseteq \vrt(\tau)$ be a path of nontrivial vertices that connects the first nonroot vertex with $w_r$. Given any $\theta \in [0, 2]$, we define a function $\Theta \colon \vrt(\tau) \to \lbrace 0, \theta \rbrace$ such that $\Theta(\mathscr{P}) = \theta$, $\Theta(\vrt(\tau) \setminus \mathscr{P}) = 0$ and subsequently rewrite~\eqref{eq:th:beta_function_trees_bound} as
\begin{equation}
\abs*{\beta^{J, s}_{h+1}} \le \sum_\tau \sum_{\underline{P}, \underline{\Delta}} C^k \frac{\lambda^{2n} \cutoff^{2n}}{M^{2n}} \prod^\circ_v 2^{(h_v - h'_v)(D_v - R_v + \Theta_v)} 2^{(h'_v - h_v) \Theta_v} \prod_v^\lambda 2^{2(h'_v - N)}.
\end{equation}
Since $\prod_v^\circ 2^{(h'_v - h_v) \Theta_v} = 2^{\theta(h+1-h_{w_r})}$ and $\prod_v^\lambda 2^{2(h'_v - N)} \le 2^{\theta(h_{w_r} - N)}$, we have
\begin{equation}
\label{eq:th:beta_tree_theta}
\abs*{\beta^{J, s}_{h+1}} \le 2^{\theta (h+1-N)} \sum_\tau \sum_{\underline{P}, \underline{\Delta}} C^k \frac{\lambda^{2n} \cutoff^{2n}}{M^{2n}} \prod^\circ_v 2^{(h_v - h'_v)(D_v - R_v + \Theta_v)}.
\end{equation}
Let us choose $\theta = 1$. In this case, the combination $D_v - R_v + \Theta_v$ is always strictly negative and it becomes \emph{arbitrarily} negative as $\abs*{P_v}$ increases. In fact, if we write
\begin{align*}
\notag
D_v - R_v + \Theta_v 
& = \left( D_v - R_v + \Theta_v + \frac{\abs*{P_v}}{4} \chi(\abs*{P_v} \ge 8) \right) - \frac{\abs*{P_v}}{4} \chi(\abs*{P_v} \ge 8) \\
& \equiv a_v - \frac{\abs*{P_v}}{4} \chi(\abs*{P_v} \ge 8),
\end{align*} 
then $a_v = - \abs*{a_v} \le -1$ for every choice of $P_v$. 

It can be shown that summing over $\underline{P}$ is essentially the same as summing over $\abs*{P_v} \in \N$ for every nontrivial vertex $v$, apart from some ineffective combinatorial factors (a detailed proof of this can be found in~\cite[Appendix A.6.1]{Gentile_2001}). The sum over $\underline{\Delta}$ is simply controlled by $C_1^n$ for a suitable $C_1 > 0$, because the total number of derivatives per field variable cannot be greater than three. In order to sum over $\tau$, we first sum over all the \emph{unlabelled} trees with $n$ endpoints of type $\lambda$ and one endpoint of type $J$ (since there can be at most $2n$ nontrivial vertices, the number of such trees is bounded by $C_2^n$ for some constant $C_2 > 0$, see~\cite[Lemma 2.1]{Mastropietro_Nonpert_Renorm}), then we sum over all the possible values of the scale differences $\lbrace b_v \equiv h_v - h'_v \rbrace_v$ and we eventually sum over $n$. Recalling that $D_v - R_v + \Theta_v \le -1 - (\abs*{P_v}/4)\chi(\abs*{P_v} \ge 8)$ and every tree carries at most $2n$ nontrivial vertices,~\eqref{eq:th:beta_tree_theta} becomes
\begin{equation}
\label{eq:th:sum_over_P_Delta}
\abs*{\beta^{J, s}_{h+1}} \le 2^{h+1-N} \sum_{n \ge 1} (c')^n \frac{\lambda^{2n} \cutoff^{2n}}{M^{2n}} \left(\sum_{b_v = 1}^{N - h^\star} 2^{- b_v} \sum_{\abs*{P_v} \ge 8} 2^{-\abs*{P_v}/4} \right)^{2n},
\end{equation}
where $c' > 0$ is a multiple of $C$. The sums are now elementary and they can be explicitly perfomed: if $c' \lambda^2 \cutoff^2/M^2 < 1$, the final bound agrees with~\eqref{eq:boundedness_beta_functions}.

The same argument applies to $\beta^s_{h+1}, \mathscr{S} \beta^{J, s}_{h+1}, \mathscr{S} \beta^s_{h+1}$. In the case of $\beta^{m, s}_{h+1}$, we have instead
\begin{equation}
\abs*{\beta^{m, s}_{h+1}} = \abs*{\mathscr{S} \beta^{m, s}_{h+1}} \le \sum_{\underline{P}, \underline{\Delta}} \sum_\tau 2^{h^\star - (h+1)} C^k \frac{\lambda^{2n} \cutoff^{2n}}{M^{2n}} 2^{h+1} \prod^\circ_v 2^{(h_v - h'_v)(D_v - R_v)} \prod_v^\lambda 2^{2(h'_v - N)},
\end{equation}
where we exploited the fact that $\beta^{m, s}_{h+1} \vert_{m_N = 0} = 0$ and we used~\eqref{eq:renormalization_bound_zero_mass} to deal with the action of the $\mathscr{S}$ operator. Here, the $2^{h^\star - (h+1)}$ gain produced by $\mathscr{S}$ combines with the dimensional factor $2^{h+1}$ and the short-memory factor $2^{h+1 -N}$ to give $\abs*{\beta_{h+1}^{m, s}} \le 2^{h^\star + h+1 - N} \cdot (c \lambda^2 \cutoff^2/M^2)$, in agreement with~\eqref{eq:boundedness_beta_functions}.

We conclude that the estimates~\eqref{eq:boundedness_beta_functions},~\eqref{eq:boundedness_beta_functions_zero_mass} hold on scale $h+1$ with a constant $c'$ that only depends on $C$, so the induction step works if we choose $c \equiv c'$ at the beginning of the inductive procedure.

\paragraph{Validity of the induction hypotesis for $h = N-2$.} To complete the proof, we are left to show that the left hand side of~\eqref{eq:induction_step} holds when $h = N - 2$. In first place, we notice that $\beta^{m, s}_N = \beta^s_N = 0$, so the single scale propagator $\sing{g}{N}$ satisfies the bounds~\eqref{eq:propagator_L_infty_norm},~\eqref{eq:propagator_weighted_norm} with the same constant $C_H$ defined at the beginning of the proof. In addition, the kernels $\sing{W}{N}_{(\lambda)}, \sing{W}{N}_{(J)}, \sing{W}{N}_{(\eta)}$ are clearly bounded as in~\eqref{eq:th:bounds_endpoints},~\eqref{eq:th:bounds_endpoints_zero_mass} with the same constant $C_E$ introduced at the beginning of the proof. The same constructions used to prove~\eqref{eq:th:r_bound} and~\eqref{eq:th:p_bound} then imply that~\eqref{eq:renormalization_bound} and~\eqref{eq:renormalization_bound_zero_mass} hold for every $\tau \in \trees'_{N-1|N}$. Since $\beta^{J, s}_N$ can be expanded in terms of nontrivial trees with root on scale $N-1$, the estimates shown in the previous paragraph imply that $\beta^{J, s}_N$ satisfies~\eqref{eq:boundedness_beta_functions},~\eqref{eq:boundedness_beta_functions_zero_mass} with the constant $c' = c$ introduced above. In conclusion, the left hand side of~\eqref{eq:induction_step} holds when $h = N - 2$.
\end{proof}

\begin{rmk}
\label{rmk:gain}
The renormalization operator acts nontrivially \emph{precisely} when $D_v \ge 0$. If we are interested in kernels arising from trees with at most one $J$ endpoint and whose $\eta$ endpoints are contracted on scale $h^\star$, we have the remarkable bound $D_v -R_v \le -2$.
\end{rmk}

\subsection{The flow of the running coupling constants}
\label{subsec:rcc}

\begin{theorem}
\label{th:bare_constants}
If $\lambda^2 \cutoff^2/M^2$ is sufficiently small, it is possible to choose the bare parameters $Z^{J, s}_N, Z^s_N, m^s_N$ so that
\begin{equation}
\label{eq:rcc_fixing}
Z^{J, s}_{h^\star} = Z^s_{h^\star-1} = 1, \qquad m^s_{h^\star-1} = m \qquad \forall s = \pm,
\end{equation}
where $m > 0$ is the physical mass of the fermion. With this choice, the running coupling constants satisfy
\begin{equation}
\label{eq:rcc_bounds}
\abs*{Z^{J, s}_h - 1} \le \frac{c \lambda^2 \cutoff^2}{M^2} 2^{h - N}, \quad \abs*{Z^s_h - 1} \le \frac{c \lambda^2 \cutoff^2}{M^2} 2^{h - N}, \quad \abs*{m^s_h - m} \le m\frac{c \lambda^2 \cutoff^2}{M^2} 2^{h - N}.
\end{equation}
\end{theorem}
\begin{proof}
For the sake of this proof, we exploit the fact that the highest order localization operators occurring in~\eqref{eq:localization_definition} do \emph{not} contribute to $\loc \sing{\mathscr{V}}{h}$, as shown in~\eqref{eq:localization_request}. Thus, we can treat $\ren$ as if it were a purely Taylor renormalization at the price of worsening the estimates~\eqref{eq:th:r_bound},~\eqref{eq:th:p_bound} by letting $R_v \in \lbrace 2, 1, 0 \rbrace$ instead of $R_v \in \lbrace 3, 2, 0 \rbrace$.

In absence of zero-mass renormalizations, the tree expansion naturally allows to treat $\beta_h^s, \beta^{J, s}_h, 2^{-h} \beta^{m, s}_h$ as functions of $\lambda$ and of the finite sequence of \emph{independent variables} $\lbrace Z^s_j, Z^{J, s}_j, m^s_j \rbrace_{j =h+1}^N$. We therefore introduce the vectors
\begin{align*}
\vec{v}_h
& \equiv (Z^+_h, Z^-_h, Z^{J, +}_h, Z^{J, -}_h, 2^{-h^\star} m^+_h, 2^{-h^\star} m^-_h), \\
\vec{\beta}_h & \equiv (\beta^+_h, \beta^-_h, \beta^{J, +}_h, \beta^{J, -}_h, 2^{-h^\star} \beta^{m, +}_h, 2^{-h^\star} \beta^{m, -}_h),
\end{align*}
so that, given any $h = h^\star, \dots, N$, the RG equations read
\begin{equation}
\label{eq:rg_equations}
\vec{v}_{h^\star-1} = \vec{v}_h + \sum_{j = h^\star}^h \vec{\beta}_j(\lambda, \underline{\vec{v}} \,),
\end{equation}
where $\underline{\vec{a}} \equiv ( \vec{a}_h )_{h=h^\star}^N \in \R^{N - h^\star + 1} \otimes \R^6$ and it is understood that $\beta^{J, s}_{h^\star} = 0$. The condition~\eqref{eq:rcc_fixing} can be expressed as $\mathsf{T} \underline{\vec{v}} = \underline{\vec{v}}$, where the map $\mathsf{T}$ acts as
\begin{equation}
(\mathsf{T} \underline{\vec{v}} \,)_h \equiv \vec{w}_h - \! \sum_{j = h^\star}^h \vec{\beta}_j(\lambda, \underline{\vec{v}} \,)
\end{equation}
with $\vec{w}_h \equiv (1, 1, 1, 1, 2^{-h^\star} m, 2^{-h^\star} m) \,\, \forall h$. Note that, by construction, \emph{any} vector $\underline{\vec{v}}$ that satisfies $\mathsf{T} \underline{\vec{v}} = \underline{\vec{v}}$ must also satisfy $\underline{\vec{v}}_{h-1} = \underline{\vec{v}}_h + \vec{\beta}_h(\lambda, \underline{\vec{v}} \,)$.

Let $\mathcal{U} \subseteq \R^{N - h^\star + 1} \otimes \R^6$ be a closed ball of radius $\epsilon > 0$ centered around $\underline{\vec{w}}$, constructed with respect to the norm $\norm*{\underline{\vec{a}}} \equiv \sum_{j = h^\star}^N \sum_{q = 1}^6 \abs*{a_j^q}$. If the maximum gain produced by $\ren$ is equal to $2$, the beta functions can still be bounded by means of the short memory property, but we must choose $\theta < 1$ in order to have $a_v \equiv D_v - R_v + \Theta_v + (\abs*{P_v}/4) \chi(\abs*{P_v} \ge 8) < 0$. In particular, if both $\lambda^2 \cutoff^2/M^2$ and $\norm*{\underline{\vec{v}} - \underline{\vec{w}}}$ are sufficiently small, the same construction adopted in the proof of Theorem~\ref{th:renormalization_bound} yields $\abs*{\beta^q_j(\lambda, \underline{\vec{v}})} \le (\mathrm{const}) \cdot (\lambda^2 \cutoff^2/M^2) \cdot 2^{(h-N)/2}$ for every $q = 1, \dots, 6$ and for every $j \in h^\star, \dots, N$. Provided that $\lambda^2 \cutoff^2/M^2$ is small enough, the map $\mathsf{T}$ sends $\mathcal{U}$ into itself, because
\begin{equation}
\underline{\vec{v}} \in \mathcal{U} \Rightarrow \norm*{\mathsf{T} \underline{\vec{v}} - \underline{\vec{w}}} \le \sum_{j = h^\star}^N \sum_{q = 1}^6 \abs*{\beta^q_j(\lambda, \underline{\vec{v}})} \le 6c \frac{\lambda^2 \cutoff^2}{M^2} \sum_{j = h^\star}^N 2^{(j-N)/2} \le 6C \frac{\lambda^2 \cutoff^2}{M^2} \le \epsilon.
\end{equation}
Moreover, it can be shown that
\begin{equation}
\label{eq:th:contraction}
\norm*{\mathsf{T} \underline{\vec{v}}_{(1)} - \mathsf{T} \underline{\vec{v}}_{(2)}} \le \sum_{j = h^\star}^N \sum_{q = 1}^6 \abs*{\beta^q_j(\lambda, \underline{\vec{v}}_{(1)}) - \beta^q_j(\lambda, \underline{\vec{v}}_{(2)})} \le \norm*{\underline{\vec{v}}_{(1)} - \underline{\vec{v}}_{(2)}} \frac{c\lambda^2 \cutoff^2}{M^2} \sum_{j = h^\star}^N 2^{(j - N)/2},
\end{equation}
see for instance~\cite[Appendix A5]{Giuliani_2005}. The idea is simple: the difference $\beta^q_j(\lambda, \underline{\vec{v}}_{(1)}) - \beta^q_j(\lambda, \underline{\vec{v}}_{(2)})$ can be written in terms of kernels of the form $\mathscr{T}_{1, 2} W_{\tau, \underline{P}, \underline{\Delta}}$, where $\mathscr{T}_{1, 2} (\, \cdot \,) \equiv (\, \cdot \,) \vert_{\underline{\vec{v}}_{(1)}} - (\, \cdot \,)\vert_{\underline{\vec{v}}_{(2)}}$. The $\mathscr{T}_{1, 2}$ operator is a generalization of $\mathscr{S}$, so the whole argument that starts from the decomposition~\eqref{eq:P_action_explicit} can be easily adapted to~\eqref{eq:th:contraction} as well. 

After summing over $j$ in~\eqref{eq:th:contraction}, we obtain $\norm*{\mathsf{T} \underline{\vec{v}}_{(1)} - \mathsf{T} \underline{\vec{v}}_{(2)}} \le c' \lambda^2 \cutoff^2/M^2$, so $\mathsf{T} \vert_{\mathcal{U}}$ is a contraction on $\mathcal{U}$ if $\lambda^2 \cutoff^2/M^2$ is suffciently small. Banach-Caccioppoli's theorem then guarantees the existence of a unique fixed point $\underline{\vec{u}} \colon \mathsf{T} \underline{\vec{u}} = \underline{\vec{u}}$. By construction, the sequence $\lbrace \vec{u}_h \rbrace_{h = h^\star}^N$ represents a RG flow such that $\vec{u}_{h^\star-1} \equiv \vec{u}_{h^\star} + \vec{\beta}_{h^\star}(\lambda, \underline{\vec{u}}) = (1, 1, 1, 1, 2^{-h^\star} m, 2^{-h^\star} m)$: recalling that $\beta^{J, s}_{h^\star} = 0 \Rightarrow Z^{J, s}_{h^\star-1} = Z^{J, s}_{h^\star}$, the renormalization conditions~\eqref{eq:rcc_fixing} are satisfied if we choose the ultraviolet parameters as $(Z^+_N, Z^-_N, \dots, 2^{-h^\star} m^-_N) \equiv (u_N^1, u_N^2, \dots, u_N^6)$. Finally, the bounds~\eqref{eq:rcc_bounds} are an immediate consequence of~\eqref{eq:boundedness_beta_functions} and~\eqref{eq:rg_equations} combined with~\eqref{eq:rcc_fixing}.
\end{proof}

\subsection{Integration of the lowest scale}
\label{subsec:last_scale}
Thanks to the presence of a nonvanishing fermion mass, the norms of the renormalized propagator $\sing{(g')}{\le h^\star}$ (defined as in Section~\ref{subsec:single_scale_integration}) satisfy the bounds~\eqref{eq:propagator_weighted_norm},~\eqref{eq:propagator_L_infty_norm}, so it is possible to interrupt the Renormalization Group flow on scale $h^\star - 1$. In fact, $\sing{(g')}{\le h^\star}$ is equal to the periodization of the function displayed in~\eqref{eq:propagator_periodization} with $h = h^\star$ and with $\chi_0$ in place of $f_0$. Following the same argument used in the proof of Lemma~\ref{lm:propagator_bounds}, it can be easily seen that the ``denominator'' $u(k) = i \tilde{\gamma}^\mu_{h^\star}(2^{h^\star}k) k_\mu + 2^{-h^\star}\tilde{\mathsf{m}}_{h^\star}(2^{h^\star}k)$ satisfies the bound $\norm*{\partial_{k_\mu}^j u} \le c_2 C^j (j!)^2$ for every $j \in \N, k \in \supp(\chi_0)$, as it happens for $\sing{g}{h}$. Moreover, the first factor appearing in~\eqref{eq:lm:u-1} is manifestly controlled by some scale-independent constant $c_3$, while
\begin{equation}
\label{eq:mass_dominance_over_kinetic_term}
\sup_{k \in \supp(\chi_0)}\abs{k^2 \tilde{Z}^+_{h^\star}(2^{h^\star} k) \cdot \tilde{Z}^-_{h^\star}(2^{h^\star}k) + 2^{-2h^\star}[\tilde{m}^s_{h^\star}(2^{h^\star}k)]^2}^{-1} \le \abs{2^{-2h^\star}[\tilde{m}^s_{h^\star}(2^{h^\star}k)]^2}^{-1} \le c_4,
\end{equation}
where $c_4$ is another scale-independent constant. We therefore have $\norm*{u^{-1}(k)} \le c_3 c_4 \equiv c_1$ and the proof of Lemma~\ref{lm:propagator_bounds} can be followed with no further variations.

Once the integration with respect to $\sing{(g')}{\le h^\star}$ is explicitly performed, we are left with the expansion
\begin{equation}
W[\src] = \sum_{\tau \in \trees'_{h^\star - 1|N}} \val(\tau)[\src],
\end{equation}
it being understood that the propagator associated with the lowest scale is $\sing{g}{h^\star} \equiv \sing{(g')}{\le h^\star}$. Note that $\val(\tau)$ does not depend on $\psi$, because every $\psi$ field has been integrated away during the last Renormalization Group step. Since $\sing{g}{h^\star}$ behaves as a single scale propagator, the bound~\eqref{eq:renormalization_bound} still applies to every tree $\tau \in \trees'_{h^\star - 1|N}$.

\subsection{Infinite volume limit}
\label{subsec:thermodynamic_limit}

The bounds obtained in Theorem~\ref{th:renormalization_bound} are uniform with respect to the spacetime volume. Indeed, it can be shown that the sequence of finite-volume kernels $\lbrace W_{\tau, \underline{P}, \underline{\Delta}, L} \rbrace_{L \in \N}$ converges in the $L \to +\infty$ limit with respect to the uniform norm on every \emph{fixed} compact subset of $\R^{4(\abs*{P_{w_0}} + n^J_{w_0} + n^\eta_{w_0})}$, where $w_0$ is the first nonroot vertex of $\tau$. A complete proof of this statement is rather long, so we shall not discuss it here; a detailed treatment can be found in~\cite[Appendix D]{Giuliani_Mastropietro_2010}. Although our model differs from the one analyzed in~\cite{Giuliani_Mastropietro_2010}, the proof can be adapted to the present case with little effort. The main idea consists in noticing that the infinite volume limits of the functions $\delta_N, W_{(\eta)}^N, W_{(J)}^N, W_{(\lambda)}^N, \sing{g}{N}$ exist and they are equal to the inverse Fourier transforms of $\chi_{N+1}(k), \hat{W}_{(\eta)}^N(k), \hat{W}_{(J)}^N(p, p'), \hat{W}_{(\lambda)}^N(k_1, k_2, k_3), \sing{\hat{g}}{N}(k)$, thought as functions of $\R^4$ vectors. As a consequence, the kernels $W_{\tau, \underline{P}, \Delta, L}$ can be defined in the infinite volume limit for every $\tau \in \trees'_{N-1|N}$ (this is realized by taking~\eqref{eq:W_tau_recursive} and replacing every function $\delta_N, W_{(\eta)}^N, \dots$ with its infinite volume limit). The difference between $W_{\tau, \underline{P}, \Delta, L}$ and $W_{\tau, \underline{P}, \Delta}$ can be shown to vanish as $L \to +\infty$ (see~\cite[Equations (D.7), (D.8) et seq.]{Giuliani_Mastropietro_2010}). The existence of the infinite volume limits of the above kernels implies the existence of the infinite volume limits of the running coupling constants $m^s_{N-1}, Z^s_{N-1}, Z^{J, s}_{N-1}$ as well as those of the functions $W_{(\eta)}^{N-1}, W_{(J)}^{N-1}, W_{(\lambda)}^{N-1}, \sing{g}{N-1}$. Relying on the recursive definition~\ref{eq:sum_over_trees_P}, we can iteratively apply the above construction to every kernel $W_{\tau, \underline{P}, \underline{\Delta}, L}$ for every $\tau \in \trees'_{j|N}$, with $j = h^\star-1, \dots, N-1$.

From now on, we will always assume that the infinite volume limit has been taken.

\section{Evaluation of the anomalous gyromagnetic factor}
\label{sec:g-2_calculation}

Based on the results discussed in Section~\ref{sec:RG_analysis} and~\ref{sec:renormalized_expansion_bounds}, we can finally address the problem of evaluating $\rgyr$, which is expressed by (19) in terms of 
$\hat{\Gamma}^\mu(p', p)$ and its derivatives computed at $p=p'=0$.
The bounds~\eqref{eq:renormalization_bound},~\eqref{eq:renormalization_bound_zero_mass} yield to the conclusion that the tree expansions for the correlation functions 
$\hat{\Gamma}^\mu(p', p)$ are absolutely convergent if $\lambda \le \BigO(M/\cutoff)$ and bounded at any order $n$ by $\BigO((\lambda M/\cutoff)^{2n})$.
This bound is however not sufficient to prove our main result, which estabilishes that $\rgyr$ is bounded above and below by a $\BigO(m^2/M^2)$ constant.
In order to get this result we note that, based on Theorem~\ref{th:renormalization_bound} and Theorem~\ref{th:bare_constants}, it is 
\begin{align}
\begin{split}
\label{eq:convergent_expansion_vertex}
\hat{\Gamma}^\mu(p', p) 
& = \gamma^\mu + \sum_{n = 1}^{+\infty} \lambda^{2n} \hat{\Gamma}^\mu_n(p', p; \lambda),
\end{split} \\
\begin{split}
\label{eq:convergent_expansion_S}
\hat{S}(k)
& = \sum_{n = 0}^{+\infty} \lambda^{2n} \hat{S}_n(k; \lambda) \equiv \frac{1}{i \slashed{k} + m} \left(1 + \sum_{n = 1}^{+\infty} \lambda^{2n} \hat{G}_n(k; \lambda) \right), 
\end{split}
\end{align}
where each contribution of order $n$ inlcudes the sum over all the possible trees with $n$ endpoints of type $\lambda$; in particular, the first of the above expressions coincides with~\eqref{eq:vertex_splitting}. The simple form of the $n = 0$ terms is due to the fact that we have imposed the renormalization conditions via Theorem~\ref{th:bare_constants}. 

In this section, we shall prove three properties:
\begin{itemize}
\item If $\abs*{p'}, \abs{p} \le m/2$ and $n \ge 1$, the bound for $\lambda^{2n} \hat{\Gamma}^\mu_n(p', p; \lambda)$ can be improved by a factor $m^2/M^2$, provided that the value of $\lambda$ is sligthly decreased. This is due to the presence of symmetries implemented by our choice of the $\ren$ operator defined in Section~\ref{subsec:localization}.
\item This $m^2/M^2$ improvement is not present in the $n = 0$ term of the expansion~\eqref{eq:convergent_expansion_vertex}; however, the contribution to $\rgyr$ coming from this term is vanishing.
\item We separate the $n = 1$ term from the rest and we show that it gives the same value of the Jackiw-Weinberg formula~\eqref{eq:JW_result} up to subdominant corrections.
\end{itemize}

In what follows, the symbol $\BigO(\lambda f)$ is used to denote any function of $\lambda$ that satisfies $\abs{\BigO(\lambda f)} \le C \cdot \abs{\lambda f}$ for some constant $C > 0$ that \emph{does not depend on $m, M, \cutoff, \lambda$}.

\subsection{Bounds for \texorpdfstring{$\hat{\Gamma}^\mu(p', p)$}{G} and \texorpdfstring{$\hat{S}(k)$}{S}}

\label{subsec:bounds_on_propagator_and_vertex}

By combining Theorem~\ref{th:renormalization_bound} and Theorem~\ref{th:bare_constants}, we exhibit a bound on the two-point function and the amputated vertex function. In particular, we prove that the subdominant part of $\hat{\Gamma}^\mu(p', p)$ is suppressed by a factor $m^2/M^2$.

\begin{theorem}
\label{th:convergent_expansions}
If $\lambda^2 \cutoff^2/M^2 \cdot \log^2(\cutoff/m)$ is sufficiently small, then~\eqref{eq:convergent_expansion_vertex},~\eqref{eq:convergent_expansion_S} hold with
\begin{equation}
\label{eq:bounds_on_S_and_Gamma}
\norm*{m^\ell \partial^\ell \hat{G}_n(k; \lambda)}, \,\, \norm*{m^\ell \partial^\ell \hat{\Gamma}^\mu_n(p', p; \lambda)} \le \frac{m^2}{M^2} \cdot \left(\frac{C_\ell \cutoff^2}{M^2} \right)^{n - 1} \log^{2n} \left( \frac{M}{m} \right) \log^{2n} \left( \frac{\cutoff}{M} \right)
\end{equation}
for every $k, p', p$ such that $\abs{k}, \abs*{p}, \abs*{p'} \le m/2$ and for every $\ell$, where $C_\ell$ is a suitable $\ell$-dependent constant.
\end{theorem}

\begin{proof}
The tree expansion for $\hat{\mathfrak{I}}^\mu(p', p)$ reads
\begin{equation}
\label{eq:full_vertex_tree_expansion}
\hat{\mathfrak{I}}^\mu(p', p) = \sum_{n \ge 0} \, \sum_{\tau_n} \, \sum_{\underline{P}, \underline{\Delta}} \hat{\mathfrak{I}}^\mu_{\tau_n, \underline{P}, \underline{\Delta}}(p', p) \equiv \sum_{n \ge 0} \lambda^{2n} \hat{\mathfrak{I}}^\mu_n(p', p; \lambda),
\end{equation}
where the sum $\sum_{\tau_n}$ is extended over trees $\tau_n \in \trees'_{h^\star-1|N}$ with $n$ endpoints of type $\lambda$, one $J$ endpoint and two $\eta$ endpoints; the sets $\underline{P}, \underline{\Delta}$ are chosen so that the kernel $\hat{\mathfrak{I}}^\mu_{\tau, \underline{P}, \underline{\Delta}}(p', p)$ has no external $\psi$ fields. Due to the compact support of the single scale propagators, the condition $\abs*{p}, \abs*{p'} \le m/2$ implies that the $\eta$ endpoints of any tree contributing to~\eqref{eq:full_vertex_tree_expansion} must be contracted on scale $h^\star$.

If $n \ge 1$, the short memory property can be used to extract an overall $m^2/M^2$ factor from $\hat{\mathfrak{I}}^\mu_{\tau_n, \underline{P}, \underline{\Delta}}(p', p)$. According to Theorem~\ref{th:renormalization_bound}, we have
\begin{equation}
\norm*{\hat{\mathfrak{I}}^\mu_{\tau_n, \underline{P}, \underline{\Delta}}(p', p)} \le \frac{C^n \lambda^{2n} \cutoff^{2n}}{M^{2n}} 2^{-2h^\star} \prod^\circ_v 2^{(h_v - h'_v)(D_v - R_v + \Theta_v)} 2^{(h'_v - h_v)\Theta_v} \prod_v^\lambda 2^{2(h'_v - N)},
\end{equation}
where the function $\Theta \colon \vrt(\tau) \to \lbrace 0, 2 \rbrace$ satisfies $\Theta(\mathscr{P}) = \theta \equiv 2, \,\, \Theta(\vrt(\tau_n) \setminus \mathscr{P}) = 0$ and the path $\mathscr{P}$ connects the first nonroot vertex (which, according to the above considerations, lies on scale $h^\star$) with a nontrivial vertex that precedes some $\lambda$ endpoint. Following the same steps discussed in the proof of Theorem~\ref{th:renormalization_bound}, we obtain
\begin{equation}
\norm*{\hat{\mathfrak{I}}^\mu_{\tau_n, \underline{P}, \underline{\Delta}}(p', p)} \le \frac{C^n \lambda^{2n} \cutoff^{2n}}{M^{2n}} 2^{2(h^\star - N)} \cdot 2^{-2h^\star} \prod^\circ_v 2^{(h_v - h'_v) a_v} \cdot 2^{-(\abs*{P_v}/6) \chi(\abs*{P_v} \ge 8)},
\end{equation}
where $a_v \equiv D_v - R_v + \Theta_v + (\abs*{P_v}/6) \chi(\abs*{P_v} \ge 8)$. This quantity is strictly negative only if $v$ lies outside the path $\mathscr{P}$. If $v \in \mathscr{P}$, $a_v$ may \emph{vanish}, because we are choosing $\theta = 2$ instead of $\theta = 1$. Recalling that $2^{2(h^\star - N)} \lesssim m^2/\cutoff^2 = (m^2/M^2) \cdot (M^2/\cutoff^2)$, we have
\begin{equation}
\norm*{\hat{\mathfrak{I}}^\mu_{\tau_n, \underline{P}, \underline{\Delta}}(p', p)} \le  2^{-2h^\star} \frac{m^2}{M^2} C^n \lambda^{2n} \left(\frac{\cutoff^{2}}{M^2} \right)^{n-1} \prod^\circ_v 2^{-(\abs*{P_v}/6) \chi(\abs*{P_v} \ge 8)} \prod^\circ_{v \notin \mathscr{P}} 2^{-(h_v - h'_v)} \prod^\circ_{v \in \mathscr{P}} (1).
\end{equation}
The sum over $\tau_n, \underline{P}, \underline{\Delta}$ can be performed as shown during the proof of Theorem~\ref{th:renormalization_bound}, namely 
\begin{multline}
\norm*{\hat{\mathfrak{I}}^\mu_n(p', p; \lambda)} \le \lambda^{-2n} \sum_{\tau_n, \underline{P}, \underline{\Delta}} \norm*{\hat{\mathfrak{I}}^\mu_{\tau_n, \underline{P}, \underline{\Delta}}(p', p)} \le  2^{-2h^\star} \frac{m^2}{M^2} \left(\frac{C \cutoff^{2}}{M^2} \right)^{n-1}\, \times \\
\times \prod^\circ_v \sum_{\abs*{P_v} \ge 8} 2^{-\abs*{P_v}/6}\prod^\circ_{v \notin \mathscr{P}} \sum_{b_v = 1}^{N-h^\star} 2^{-b_v} \! \prod^\circ_{v \in \mathscr{P}} \sum_{b_v=1}^{N-h^\star} (1),
\end{multline}
where the constant $C$ has been properly redefined. The first two sums occurring in the second line are bounded by $1$, while 
\begin{equation}
\prod^\circ_{v \in \mathscr{P}} \sum_{b_v=1}^{N-h^\star} (1) = \abs*{N - h^\star -1}^{\abs*{\mathscr{P}}} \le \abs*{N - h^\star -1}^{2n} \lesssim \log^{2n} \left( \frac{\cutoff}{m} \right).
\end{equation}
Since $M > m$ and $\cutoff > M$, we can write $\log(\cutoff/m) \le 2\log(\cutoff/M) \log(M/m)$, so by further redefining the constant $C$ we obtain
\begin{equation}
\label{eq:th:I_estimate}
\norm*{\hat{\mathfrak{I}}^\mu_n(p', p; \lambda)} \le  2^{-2h^\star} \frac{m^2}{M^2} \left(\frac{C \cutoff^{2}}{M^2} \right)^{n-1} \log^{2n} \left( \frac{M}{m} \right) \log^{2n} \left( \frac{\cutoff}{M} \right), \qquad n \ge 1.
\end{equation}
A similar estimate holds for $m^\ell \partial^\ell \hat{\mathfrak{I}}^\mu_n(p', p; \lambda)$ as well. In fact, the function $m^\ell \partial^\ell \hat{\mathfrak{I}}^\mu_n(p', p; \lambda)$ behaves as $m^\ell (\kloc_\ell \, \hat{\mathfrak{I}}^\mu_n)(x, y, z; \lambda)$ in position space, so its bound contains an additional $m^\ell 2^{-h^\star \ell}$ factor (see Lemma~\ref{lm:taylor_renormalization_gain}); since $m \simeq 2^{h^\star}$, this factor is of order $1$.

The argument described above can also be applied to the two-point function, yielding
\begin{equation}
\label{eq:th:S_estimate}
\norm*{\hat{G}_n(k; \lambda)} \le \frac{m^2}{M^2} \left(\frac{C \cutoff^{2}}{M^2} \right)^{n-1} \log^{2n} \left( \frac{M}{m} \right) \log^{2n} \left( \frac{\cutoff}{M} \right), \qquad n \ge 1.
\end{equation}
This, combined with the fact that $\hat{S}_0(k; \lambda) = \sing{g}{h^\star}(k) = (i \slashed{k}+ m)^{-1}$ as long as $\abs{k} \le m/2$ (see Section~\ref{subsec:tree_expansions_calculation} for further details), proves~\eqref{eq:convergent_expansion_S} and~\eqref{eq:bounds_on_S_and_Gamma} for $\hat{S}(k)$. If $\lambda^2 \cutoff^2/M^2 \log^2(\cutoff/m)$ is small enough to make the expansions for $\hat{\mathfrak{I}}^\mu(p', p)$, $\hat{S}(k)$ convergent, we also have
\begin{equation}
\label{eq:S_inverse_expansion}
\hat{S}^{-1}(k) = \sum_{n \ge 0} \lambda^{2n} (\hat{S}^{-1})_n(k; \lambda) \equiv (i \slashed{k} + m) \left( 1 + \sum_{n \ge 1} \lambda^{2n} \hat{B}_n(k; \lambda) \right),
\end{equation}
where the coefficients $\hat{B}_n(k; \lambda)$ are bounded as in~\eqref{eq:th:S_estimate}. Since $\hat{\Gamma}^\mu(p', p) = \hat{S}^{-1}(p') \cdot \hat{\mathfrak{I}}^\mu(p', p) \cdot \hat{S}^{-1}(p)$, a convergent expansion for $\hat{\Gamma}^\mu(p', p)$ is obtained as a Cauchy product; namely, we have $\hat{\Gamma}^\mu(p', p) = \sum_{n \ge 0} \lambda^{2n} \hat{\Gamma}^\mu(p', p; \lambda)$, with
\begin{equation}
\label{eq:amp_vertex_cauchy product}
\hat{\Gamma}^\mu_n(p', p; \lambda) = \sum_{\substack{n_1 + n_2 + n_3 = n \\ n_1, n_2, n_3 \ge 0}} (\hat{S}^{-1})_{n_1} (p'; \lambda) \, \hat{\mathfrak{I}}_{n_2}^\mu(p', p; \lambda) \, (\hat{S}^{-1})_{n_3}(p, \lambda).
\end{equation}
The estimates~\eqref{eq:th:I_estimate},~\eqref{eq:th:S_estimate} then imply that $\hat{\Gamma}^\mu_n(p', p; \lambda)$ satisfies~\eqref{eq:bounds_on_S_and_Gamma} for every $n \ge 1$. If $n = 0$, we have instead
\begin{equation}
\hat{\Gamma}^\mu_0(p', p; \lambda) = (\hat{S}^{-1})_0(p'; \lambda) \cdot \hat{\mathfrak{I}}^\mu_0(p', p; \lambda) \cdot (\hat{S}^{-1})_0(p; \lambda) = \gamma^\mu,
\end{equation}
so $\hat{\Gamma}^\mu(p', p)$ satisfies~\eqref{eq:convergent_expansion_vertex}. Again, the whole argument still works if $\ell \ne 0$.
\end{proof}

\subsection{Tree expansions at lowest order}
\label{subsec:tree_expansions_calculation}

As we are ultimately interested in evaluating the derivatives of $\mathcal{A}(z)$ at $z = 0$, it is sufficient to determine the functions $\hat{S}(k), \hat{\mathfrak{I}}^\mu(p', p)$ within the region $\abs{k}, \abs*{p'}, \abs{p} \le m/2$. Due to the compact support properties of the single scale propagators, the only relevant trees are
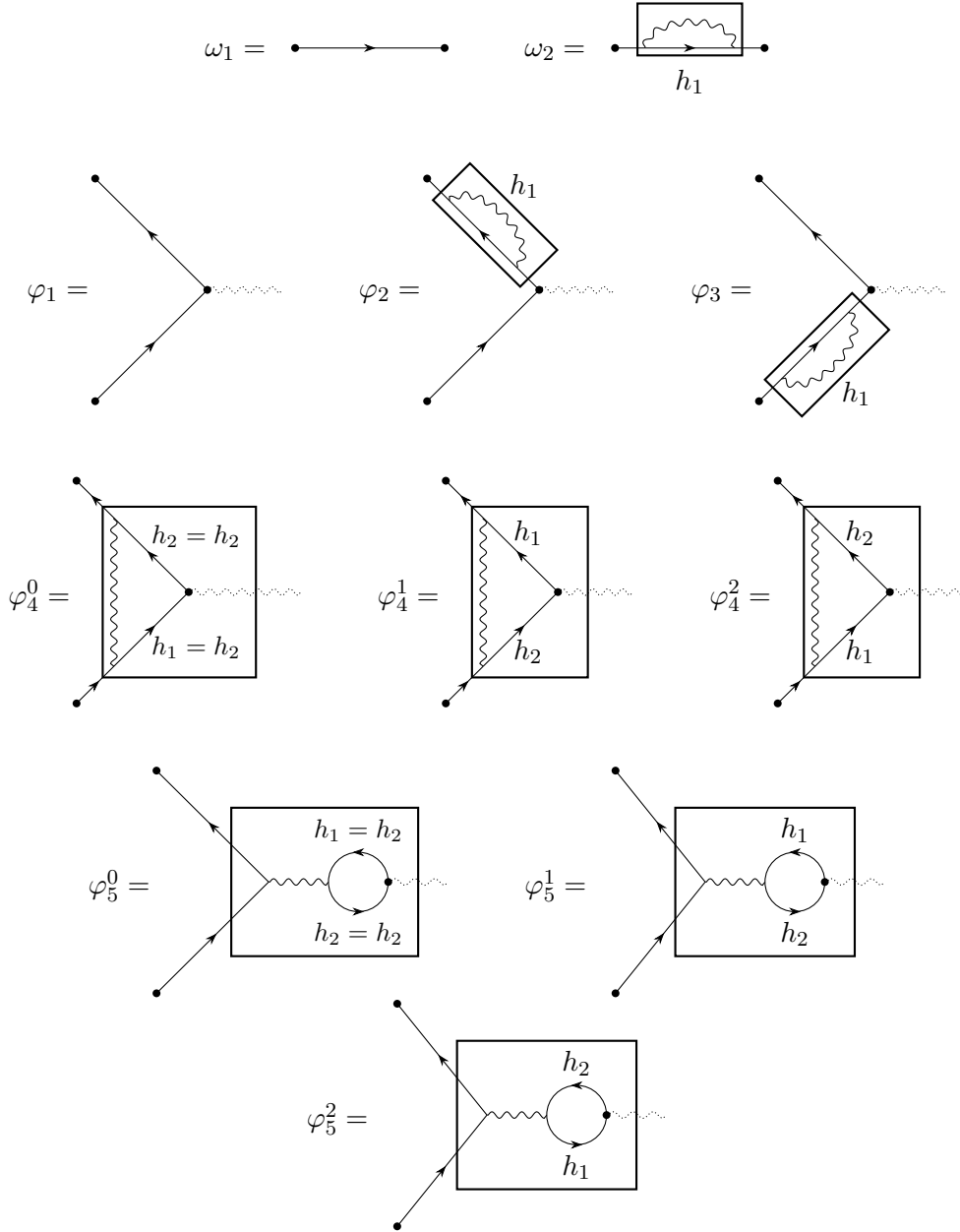
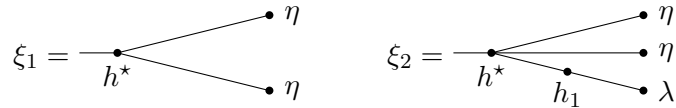
\begin{figure}[p]
\begin{center}
\begin{tikzpicture}[baseline]
\begin{feynman}
\draw[with arrow]	(-1, 0)		--	(1, 0);
\fill[black]		(-1, 0)		circle	(1.5pt)
				(1, 0)		circle	(1.5pt);
\end{feynman}
\node	(ll)	at	(-1.8, -0.05)	{$\omega_1 = $};
\end{tikzpicture}
\qquad
\begin{tikzpicture}[baseline]
\begin{feynman}
\draw[black]	(-1, 0) -- (-0.6, 0)
				(0.6, 0) -- (1, 0);
\draw[with arrow] (-0.6, 0) -- (0.6, 0);
\draw[photon]	(-0.6, 0)	to[out=85, in=85]	(0.6, 0);
\draw[black, thick]	(-0.7, -0.1)	rectangle	(0.7, 0.6);
\end{feynman}
\fill[black]	(-1, 0)	circle (1.5pt)
				(1, 0)	circle	(1.5pt);
\node	(h1)	at	(0, -0.45)	{$h_1$};	
\node	(ll)	at	(-1.8, -0.05)	{$\omega_2 = $};
\end{tikzpicture}
\\[2em]
\begin{tikzpicture}[baseline]
\begin{feynman}
\draw[photon, densely dotted]	(0, 0)	--	(1, 0);
\draw[with arrow]	(0, 0)		--	(-1.5, 1.5);
\draw[with arrow]	(-1.5, -1.5)	--	(0, 0);
\fill[black]	(0, 0)		circle	(1.5pt)
			(-1.5, -1.5)	circle	(1.5pt)
			(-1.5, 1.5)	circle	(1.5pt);
\end{feynman}
\node	(l)		at	(-2, -0.05)	{$\phi_1 = $};
\end{tikzpicture}
\qquad
\begin{tikzpicture}[baseline]
\begin{feynman}
\draw[photon, densely dotted]		(0, 0)	--	(1, 0);
\draw[photon]	(-0.3, 0.3) to [out=45, in=45] (-1.2, 1.2);
\draw[black, thick, rotate around={45:(-0.75,0.75)}] 
	(-0.9, -0.1) rectangle (-0.2 , 1.55);
\draw[with arrow]	(0, 0)		--	(-1.5, 1.5);
\draw[with arrow]	(-1.5, -1.5)	--	(0, 0);
\fill[black]		(0, 0)		circle	(1.5pt)
				(-1.5, -1.5)	circle	(1.5pt)
				(-1.5, 1.5)	circle	(1.5pt);
\end{feynman}
\node	(l)		at	(-2, -0.05)	{$\phi_2 = $};
\node	(h2)		at	(-0.2, 1.4)	{$h_1$};
\end{tikzpicture}
\qquad
\begin{tikzpicture}[baseline]
\begin{feynman}
\draw[photon, densely dotted]	(0, 0)	--	(1, 0);
\draw[photon]	(-0.3, -0.3) to [out=-45, in=-45] (-1.2, -1.2);
\draw[black, thick, rotate around={-45:(-0.75,-0.75)}] 
	(-0.9, 0.1) rectangle (-0.2 , -1.55);
\draw[with arrow]	(0, 0)		--	(-1.5, 1.5);
\draw[with arrow]	(-1.5, -1.5)	--	(0, 0);
\fill[black]		(0, 0)		circle	(1.5pt)
				(-1.5, -1.5)	circle	(1.5pt)
				(-1.5, 1.5)	circle	(1.5pt);
\end{feynman}
\node	(l)		at	(-2, -0.05)		{$\phi_3 = $};
\node	(h2)		at	(-0.2, -1.4)		{$h_1$};
\end{tikzpicture}
\\[2em]
\begin{tikzpicture}[baseline]
\begin{feynman}
\draw[photon, densely dotted]	(0, 0)	--	(1.5, 0);
\draw[photon]		(-1, 1)	--	(-1, -1);
\draw[with arrow]	(0, 0)	--	(-1, 1);
\draw[with arrow]	(-1, -1)	--	(0, 0);
\draw[with arrow]	(-1, 1)	--	(-1.5, 1.5);
\draw[with arrow]	(-1.5, -1.5)	--	(-1, -1);
\fill[black]		(0, 0)		circle	(1.5pt)
				(-1.5, -1.5)	circle	(1.5pt)
				(-1.5, 1.5)	circle	(1.5pt);
\draw[black, thick] (-1.15, -1.15) rectangle (0.9, 1.15);
\end{feynman}
\node	(l)		at	(-2, -0.05)	{$\phi^0_4 = $};
\node	(h1)		at	(0.1, -0.75)	{\small $h_1 = h_2$};
\node	(h2)		at	(0.1, 0.75)	{\small $h_2 = h_2$};
\end{tikzpicture}
\qquad
\begin{tikzpicture}[baseline]
\begin{feynman}
\draw[photon, densely dotted]	(0, 0)	--	(1, 0);
\draw[photon]				(-1, 1)	--	(-1, -1);
\draw[with arrow]			(0, 0)	--	(-1, 1);
\draw[with arrow]			(-1, -1)	--	(0, 0);
\draw[with arrow]			(-1, 1)	--	(-1.5, 1.5);
\draw[with arrow]			(-1.5, -1.5)	--	(-1, -1);
\fill[black]		(0, 0)		circle	(1.5pt)
				(-1.5, -1.5)	circle	(1.5pt)
				(-1.5, 1.5)	circle	(1.5pt);
\draw[black, thick] (-1.15, -1.15) rectangle (0.40, 1.15);
\end{feynman}
\node	(l)		at	(-2, -0.05)		{$\phi^1_4 = $};
\node	(h1)		at	(-0.4, -0.8)		{$h_2$};
\node	(h2)		at	(-0.4, 0.8)		{$h_1$};
\end{tikzpicture}
\qquad
\begin{tikzpicture}[baseline]
\begin{feynman}
\draw[photon, densely dotted]	(0, 0)	--	(1, 0);
\draw[photon]				(-1, 1)	--	(-1, -1);
\draw[with arrow]			(0, 0)	--	(-1, 1);
\draw[with arrow]			(-1, -1)	--	(0, 0);
\draw[with arrow]			(-1, 1)	--	(-1.5, 1.5);
\draw[with arrow]			(-1.5, -1.5)	--	(-1, -1);
\fill[black]		(0, 0)		circle	(1.5pt)
				(-1.5, -1.5)	circle	(1.5pt)
				(-1.5, 1.5)	circle	(1.5pt);
\draw[black, thick] (-1.15, -1.15) rectangle (0.40, 1.15);
\end{feynman}
\node	(l)		at	(-2, -0.05)		{$\phi^2_4 = $};
\node	(h1)		at	(-0.4, -0.8)		{$h_1$};
\node	(h2)		at	(-0.4, 0.8)		{$h_2$};
\end{tikzpicture}
\\[2em]
\begin{tikzpicture}[baseline]
\begin{feynman}
\draw[photon]		(0, 0)	--	(0.8, 0);
\draw[with arrow]	(1.6, 0)	arc	(0:180:0.4);
\draw[with arrow]	(0.8, 0)	arc	(180:360:0.4);
\draw[photon, densely dotted]	(1.6, 0) -- (2.4, 0);
\draw[with arrow]	(0, 0)		--	(-1.5, 1.5);
\draw[with arrow]	(-1.5, -1.5)	--	(0, 0);
\fill[black]		(-1.5, 1.5)	circle	(1.5pt)
				(-1.5, -1.5)	circle	(1.5pt)
				(1.6, 0)		circle	(1.5pt);
\draw[black, thick]	(-0.5, -1) rectangle (2, 1);
\node	(l)		at	(-2, -0.05)	{$\phi^0_5 = $};
\node	(h1)		at	(1.2, 0.7)		{\small $h_1 = h_2$};
\node	(h2)		at	(1.2, -0.7)		{\small $h_2 = h_2$};
\end{feynman}
\end{tikzpicture}
\qquad
\begin{tikzpicture}[baseline]
\begin{feynman}
\draw[photon]		(0, 0)	--	(0.8, 0);
\draw[with arrow]	(1.6, 0)	arc	(0:180:0.4);
\draw[with arrow]	(0.8, 0)	arc	(180:360:0.4);
\draw[photon, densely dotted]	(1.6, 0) -- (2.4, 0);
\draw[with arrow]	(0, 0)		--	(-1.2, 1.5);
\draw[with arrow]	(-1.2, -1.5)	--	(0, 0);
\fill[black]	(-1.2, 1.5)	circle	(1.5pt)
				(-1.2, -1.5)	circle	(1.5pt)
				(1.6, 0)		circle	(1.5pt);
\draw[black, thick]	(-0.4, -1) rectangle (2, 1);
\node	(l)		at	(-2, -0.05)		{$\phi^1_5 = $};
\node	(h1)		at	(1.2, 0.7)		{$h_1$};
\node	(h2)		at	(1.2, -0.7)		{$h_2$};
\end{feynman}
\end{tikzpicture}
\quad
\begin{tikzpicture}[baseline]
\begin{feynman}
\draw[photon]		(0, 0)	--	(0.8, 0);
\draw[with arrow]	(1.6, 0)	arc	(0:180:0.4);
\draw[with arrow]	(0.8, 0)	arc	(180:360:0.4);
\draw[photon, densely dotted]	(1.6, 0) -- (2.4, 0);
\draw[with arrow]	(0, 0)		--	(-1.2, 1.5);
\draw[with arrow]	(-1.2, -1.5)	--	(0, 0);
\fill[black]	(-1.2, 1.5)		circle	(1.5pt)
			(-1.2, -1.5)		circle	(1.5pt)
			(1.6, 0)			circle	(1.5pt);
\draw[black, thick]	(-0.4, -1) rectangle (2, 1);
\node	(l)		at	(-2, -0.05)		{$\phi^2_5 = $};
\node	(h1)		at	(1.2, 0.7)		{$h_2$};
\node	(h2)		at	(1.2, -0.7)		{$h_1$};
\end{feynman}
\end{tikzpicture}
\end{center}
\caption{Nonvanishing multiscale Feynman diagrams arising from $\xi_1, \dots, \tau_4$. Every unlabeled fermion propagator is equal to $\sing{g}{h^\star}$ and a $\ren$ operator acts on every evidenced cluster whose inner propagators live on strictly higher scales than its external legs.}
\label{fig:vertex_diagrams}
\end{figure}
\begin{center}
\begin{tikzpicture}[baseline]
\draw[black]		(-0.5, 0)	-- (0, 0)
				(0, 0)		-- (2, 0.5)
				(0, 0)		-- (2, -0.5);
\fill[black]		(0, 0)		circle	(1.5pt)
				(2, -0.5)	circle	(1.5pt)
				(2, 0.5) 	circle	(1.5pt);
\node	(hs)		at	(0, -0.3)	{$h^\star$};
\node	(e)		at	(2.3, -0.5)	{$\eta$};
\node	(f)		at	(2.3, 0.5)	{$\eta$};	
\node	(t)		at	(-1, -0.05)	{$\xi_1 = $};
\end{tikzpicture}
\qquad
\begin{tikzpicture}[baseline]
\draw[black]		(-0.5, 0)	-- (2, 0)
				(0, 0)		-- (2, 0.5)
				(0, 0)		-- (2, -0.5);
\fill[black]		(0, 0)	circle	(1.5pt)
				(2, 0)		circle	(1.5pt)
				(2, 0.5) 	circle	(1.5pt)
				(2, -0.5) 	circle	(1.5pt)
				(1, -0.25)	circle	(1.5pt);
\node	(hs)		at	(0, -0.3)	{$h^\star$};
\node	(h2)		at	(1, -0.55)	{$h_1$};
\node	(e)		at	(2.3, 0)		{$\eta$};
\node	(l)		at	(2.3, 0.5)	{$\eta$};
\node	(l)		at	(2.3, -0.5)	{$\lambda$};
\node	(t)		at	(-1, -0.05)	{$\xi_2 = $};
\end{tikzpicture}
\end{center}
for the interacting two-point function, and 
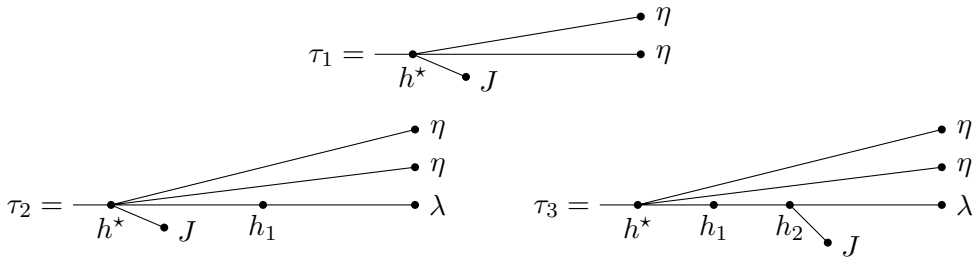
\begin{center}
\begin{tikzpicture}[baseline]
\draw[black]	(-0.5, 0)	-- (3, 0)
				(0, 0)		-- (3, 0.5)
				(0, 0)		-- (0.7, -0.3);
\fill[black]	(0, 0)	circle	(1.5pt)
				(3, 0)	circle	(1.5pt)
				(3, 0.5) circle	(1.5pt)
				(0.7, -0.3)	circle	(1.5pt);
\node	(hs)		at	(0, -0.3)	{$h^\star$};
\node	(J)		at	(1, -0.35)	{$J$};
\node	(e)		at	(3.3, 0)	{$\eta$};
\node	(f)		at	(3.3, 0.5)	{$\eta$};	
\node	(t)		at	(-1, -0.05)	{$\tau_1 = $};
\end{tikzpicture}
\\[3pt]
\begin{tikzpicture}[baseline]
\draw[black]	(-0.5, 0)	-- (4, 0)
				(0, 0)		-- (4, 1)
				(0, 0)		-- (4, 0.5)
				(0, 0)		-- (0.7, -0.3);
\fill[black]	(0, 0)	circle	(1.5pt)
				(4, 0)	circle	(1.5pt)
				(4, 0.5) circle	(1.5pt)
				(4, 1) circle	(1.5pt)
				(2, 0)	circle	(1.5pt)
				(0.7, -0.3)	circle	(1.5pt);
\node	(hs)		at	(0, -0.3)	{$h^\star$};
\node	(h2)		at	(2, -0.3)	{$h_1$};
\node	(J)		at	(1, -0.35)	{$J$};
\node	(e)		at	(4.3, 0)	{$\lambda$};
\node	(f)		at	(4.3, 1)	{$\eta$};	
\node	(l)		at	(4.3, 0.5)	{$\eta$};
\node	(t)		at	(-1, -0.05)	{$\tau_2 = $};
\end{tikzpicture}
\qquad
\begin{tikzpicture}[baseline]
\draw[black]	(-0.5, 0)	-- (4, 0)
				(0, 0)		-- (4, 1)
				(0, 0)		-- (4, 0.5)
				(2, 0)		-- (2.5, -0.5);
\fill[black]	(0, 0)	circle	(1.5pt)
				(4, 0)	circle	(1.5pt)
				(4, 0.5) circle	(1.5pt)
				(4, 1) circle	(1.5pt)
				(1, 0) circle	(1.5pt)
				(2, 0)	circle	(1.5pt)
				(2.5, -0.5)	circle	(1.5pt);
\node	(hs)		at	(0, -0.3)	{$h^\star$};
\node	(h1)		at	(1, -0.3)	{$h_1$};
\node	(h2)		at	(2, -0.3)	{$h_2$};
\node	(J)		at	(2.8, -0.55)	{$J$};
\node	(e)		at	(4.3, 0)	{$\lambda$};
\node	(f)		at	(4.3, 1)	{$\eta$};	
\node	(l)		at	(4.3, 0.5)	{$\eta$};
\node	(t)		at	(-1, -0.05)	{$\tau_3 = $};
\end{tikzpicture}
\end{center}
for the full vertex function. The corresponding Feynman diagrams are shown in Figure~\ref{fig:vertex_diagrams}. It is immediate to see that $\omega_1, \omega_2, \phi_1$ respectively come from $\xi_1, \xi_2, \tau_1$. If $h_1 = h_2 = h^\star$, $\tau_2$ and $\tau_3$ degenerate into the same tree which gives rise to $\phi_2 \vert_{h_1 = h^\star}$, $\phi_3 \vert_{h_1 = h^\star}$, $\phi_4^1 \vert_{h_1 = h_2 = h^\star}$, $\phi_5^1 \vert_{h_1 = h_2 = h^\star}$; in any other case, $\tau_2$ produces $\phi_2, \phi_3$ and $\tau_3$ produces $\phi_4^{1, 2}, \phi_5^{1, 2}$ or $\phi_4^0, \phi_5^0$ depending on whether $h_1 \ne h_2$ or $h_1 = h_2$.

In terms of the above Feynman diagrams, the lowest orders of the tree expansions for $\hat{S}(k)$ and $\hat{\mathfrak{I}}^\mu(p', p)$ are
\begin{gather*}
\begin{split}
\hat{S}_0(k; \lambda) = \hat{S}(k; \omega_1, \lambda), \qquad \lambda^2\hat{S}_1(k; \lambda) = \hat{S}(k; \omega_2, \lambda), \qquad \hat{\mathfrak{I}}^\mu_0(p', p; \lambda) = \hat{\mathfrak{I}}^\mu(p', p; \phi_1, \lambda),
\end{split} \\
\begin{split}
\lambda^2 \hat{\mathfrak{I}}^\mu_1(p', p; \lambda) = \sum_{j=2}^3 \hat{\mathfrak{I}}^\mu(p', p; \phi_j, \lambda) + \sum_{a = 0}^2 \sum_{j=4}^5 \hat{\mathfrak{I}}^\mu(p', p; \phi_j^a, \lambda),
\end{split}
\end{gather*}
where, given any diagram $\phi$, the notations $\hat{S}(k; \phi, \lambda)$, $\hat{\mathfrak{I}}^\mu(p', p; \phi, \lambda)$ stand for the contributions to $\hat{S}(k), \hat{\mathfrak{I}}^\mu(p', p)$ coming from $\phi$. By inverting the series expansion for $\hat{S}(k)$, it is easy to see that
\begin{equation}
\label{eq:S_inverse_coefficients}
(\hat{S}^{-1})_0(k; \lambda) = [\hat{S}_0(k; \lambda)]^{-1}, \qquad (\hat{S}^{-1})_1(k; \lambda) = - [\hat{S}_0(k; \lambda)]^{-1} \, \hat{S}_1(k; \lambda) \, [\hat{S}_0(k; \lambda)]^{-1};
\end{equation}
after plugging the explicit expressions of the Feynman diagrams shown in Figure~\ref{fig:vertex_diagrams} into~\eqref{eq:S_inverse_coefficients} and~\eqref{eq:amp_vertex_cauchy product}, one finds
\begin{align}
\label{eq:amp_vertex_first_order}
\begin{split}
\hat{\Gamma}^\mu_0(p', p; \lambda) 
& = \gamma^\mu, \\
\lambda^2 \hat{\Gamma}^\mu_1(p', p; \lambda) 
& = \sum_{a = 0}^2 \sum_{j = 4}^5 \, [\sing{g}{h^\star}(p')]^{-1} \, \hat{\mathfrak{I}}^\mu(p', p; \phi_j^a, \lambda) \, [\sing{g}{h^\star}(p)]^{-1} \\
& \equiv \sum_{a = 0}^2 \sum_{j = 4}^5 \hat{\Gamma}^\mu(p', p; \phi_j^a, \lambda).
\end{split}
\end{align}
As it happens in ordinary perturbation theory, diagrams $\phi_2, \phi_3$ have no influence on $\hat{\Gamma}^\mu_1(p', p; \lambda)$, because
\begin{equation}
\label{eq:self_energy_insertion}
\sum_{j=2}^3 \hat{\mathfrak{I}}^\mu(p', p; \phi_j, \lambda) = \hat{S}_1(p'; \lambda) \, \hat{\mathfrak{I}}^\mu_0(p', p; \lambda) \, \hat{S}_0(p; \lambda) + \hat{S}_0(p'; \lambda) \, \hat{\mathfrak{I}}^\mu_0(p', p; \lambda) \, \hat{S}_1(p; \lambda)
\end{equation}
and this sum is exactly cancelled by the terms with $(n_1, n_2, n_3) = (1, 0, 0), (0, 0, 1)$ occurring inside~\eqref{eq:amp_vertex_cauchy product}. The particular structure displayed in~\eqref{eq:self_energy_insertion} arises because $\ren$ acts identically on every cluster that can be disconnected by cutting a single fermionic line.

We are now left to explicitly compute the function $\mathcal{A}(z) \equiv [\mathcal{F}(z) - \mathcal{Q}(z)]/\mathcal{Q}(z)$, from which the anomalous gyromagnetic factor will be recovered. Since $\hat{\Gamma}^\mu(p', p)$ satisfies~\eqref{eq:convergent_expansion_vertex}, the difference $\abs*{\mathcal{Q}(z) - 24}$ is controlled by the right hand side of~\eqref{eq:bounds_on_S_and_Gamma} with $n = 1$. Therefore, if $\lambda^2 \cutoff^2/M^2$ is sufficiently small, the function $\mathcal{Q}(z)$ is never vanishing and $\mathcal{A}(z)$ is well-defined. Relying on~\eqref{eq:noncovariant_gfunction}, we can write
\begin{equation}
\mathcal{A}(z) = \mathcal{A}_0(z; \lambda) + \mathcal{A}_2(z; \lambda) + \mathcal{A}_{> 2}(z; \lambda),
\end{equation}
where, according to Theorem~\ref{th:convergent_expansions}, $\abs*{\mathcal{A}_{> 2}(z; \lambda)}$ is controlled by the left hand side of~\eqref{eq:bounds_on_S_and_Gamma} with $n = 2$. 

\subsection{Vanishing of \texorpdfstring{$\mathcal{A}_0(z; \lambda)$}{A0} and of the non-triangular contributions to \texorpdfstring{$\mathcal{A}_2(z; \lambda)$}{A2}}
\label{subsec:non_triangle_vanishing} 

Before getting into the actual calculation of $\rgyr$, we prove that $\mathcal{A}_0(z; \lambda) = 0$, so $\rgyr$ does not receive any contribution from the dominant part of $\hat{\Gamma}^\mu(p', p)$. We also show that $\mathcal{A}_2(z; \lambda)$ only depends on the triangular diagrams $\phi^0_4, \phi^1_4, \phi^2_4$.

Due to the anticommutation properties of the $\gamma^\mu$ matrices, any term proportional to $\gamma^\mu$ that appears within the expansion of $\hat{\Gamma}^\mu(p', p)$ annihilates the difference $\mathcal{F}(z) - \mathcal{Q}(z)$; moreover, every term proportional to $\gamma^5 \gamma^\mu$ annihilates both $\mathcal{F}(z)$ and $\mathcal{Q}(z)$.

It is easy to see that $\hat{\Gamma}^\mu_0(p', p; \lambda)$ and the local parts of $\hat{\Gamma}^\mu(p', p; \phi_5^{0, 1, 2}, \lambda)$ are linear combinations of $\gamma^\mu$ and $\gamma^5 \gamma^\mu$, so they do not contribute to $\mathcal{A}(z)$, nor to $\rgyr$. On the other hand, the \emph{non local} parts of $\hat{\Gamma}^\mu(p', p; \phi_5^{0, 1, 2}, \lambda)$ are linear combinations of integrals of the form
\begin{equation}
\label{eq:loop_integral}
H^\mu(p', p) = \hat{v}(k) \Upsilon_\nu \int \frac{\dd[4]{q}}{(2\pi)^4} \tr[\Upsilon^\nu \sing{g}{h_1}(k + q) \, (\gamma^\mu_J)_{\max(h_1, h_2)} \, \sing{g}{h_2}(q)] \equiv \hat{v}(k) \Upsilon_\nu D^{\nu \mu}(k)
\end{equation}
with $k = p' - p$ and $h^\star \le h_1, h_2 \le N$. The $\mathrm{SO}(4)$ covariance of the theory implies that $D^{\nu \mu}(0)$ is proportional to $\delta^{\nu \mu}$, so  $H^\mu(p_z, p_z) = \hat{v}(k) \Upsilon_\nu D^{\nu \mu}(0)$ is a linear combination of $\gamma^\mu$ and $\gamma^5 \gamma^\mu$. Relying on the parity cancellation $\partial^\alpha D^{\nu \mu}(0) = 0$, it can be seen that $(\partial_{p'} - \partial_p)^\alpha H^\mu(p_z, p_z) = \partial^\alpha \hat{v}(0) \Upsilon_\nu D^{\nu \mu}(0)$ is a linear combination of $\gamma^\mu$ and $\gamma^5 \gamma^\mu$ as well. We conclude that \emph{any} integral of the form~\eqref{eq:loop_integral} annihilates the function $\mathcal{F}(z) - \mathcal{Q}(z)$, so the non local parts of diagrams $\phi_5^{0, 1, 2}$ have no influence on $\rgyr$.

The combined values of $\phi_4^{0, 1, 2}$ will be conveniently represented by the function
\begin{equation}
\sum_{j=0}^2 \hat{\Gamma}^\mu(p', p; \phi_4^j, \lambda) \equiv \hat{\Gamma}^\mu_\triangle(p', p; \lambda),
\end{equation}
called \emph{triangle integral}. Explicitly,
\begin{multline}
\label{eq:triangle_multiscale}
\hat{\Gamma}^\mu_\triangle(p', p; \lambda) = - \sum_{h_1, h_2}^{\star h^\star} \int \frac{\dd[4]{q}}{(2\pi)^4} \, \Upsilon_\nu \, \sing{\hat{g}}{h_1}(p' - q) \, (\gamma^\mu_J)_{\max(h_1, h_2)} \, \sing{\hat{g}}{h_2}(p - q) \, \Upsilon^\nu \, + \\
- \sum_{h_1, h_2 > h^\star} \int \frac{\dd[4]{q}}{(2\pi)^4} \, \ren[ \Upsilon_\nu \, \sing{\hat{g}}{h_1}(p' - q) \, (\gamma^\mu_J)_{\max(h_1, h_2)} \, \sing{\hat{g}}{h_2}(p - q) \, \Upsilon^\nu ],
\end{multline}
where $\sum_{h_1, h_2}^{\star h^\star}$ denotes a sum constrained by the requirement that at least one among $h_1, h_2$ must be equal to $h^\star$. According to the above argument, the function $\mathcal{A}_2(z; \lambda)$ is entirely determined by $\hat{\Gamma}^\mu_\triangle(p', p; \lambda)$.

\subsection{Extraction of the \texorpdfstring{$\lambda$}{lambda}-independent part of \texorpdfstring{$\mathcal{A}_2(z; \lambda)$}{a}}
\label{subsec:lambda_independent_triangle} 

Since the running coupling constants $Z^{J, s}_h, Z^s_h$ depend on $\lambda$ and $h$, it is not immediate to explicitly perform the sum over $h_1, h_2$ inside the triangle integral~\eqref{eq:triangle_multiscale}. For this reason, we split the function $\hat{\Gamma}^\mu_\triangle(p', p; \lambda)$ into the sum of a $\lambda$-independent part (for which the sum over $h_1, h_2$ could be trivially performed) and a small remainder. 

\begin{lm}
\label{lm:sum_over_scales}
Let $\sing{\hat{g}_0}{h}(p) \equiv f_h(p)(i \slashed{p} + m)^{-1}$ and suppose that $\lambda^2 \cutoff^2/M^2 < 1$. For every $\ell \in \N_0$, there exists a positive constant $C_\ell$ such that
\begin{equation}
\label{eq:sum_over_scales_estimate}
\sup_{\abs*{p'}, \abs{p} \le m/2} \norm*{m^\ell \partial^\ell \hat{\Gamma}^\mu_\triangle(p', p; \lambda) - m^\ell \partial^\ell \hat{\Gamma}^\mu_\triangle(p', p)} \le C_\ell \frac{m^2}{M^2} \cdot \frac{\lambda^2 \cutoff^2}{M^2} \log \left( \frac{M}{m} \right) \frac{M}{\cutoff},
\end{equation}
where $\hat{\Gamma}^\mu_\triangle(p', p)$ is the $\lambda$-independent function obtained by replacing $\sing{\hat{g}}{h}(p), (\gamma^\mu_J)_h$ with $\sing{\hat{g}_0}{h}(p), \gamma^\mu$ inside~\eqref{eq:triangle_multiscale}.
\end{lm}

\begin{proof}
For sake of simplicity, we take $\ell = 0$ (if $\ell > 0$, the proof is essentially the same). If $\sing{\hat{g}}{h}(k) \equiv \sing{\hat{g}_0}{h}(k) + \sing{\delta \hat{g}}{h}(k)$ and $(\gamma^\mu_J)_h \equiv \gamma^\mu + (\delta \gamma^\mu_J)_h$, Theorem~\ref{th:bare_constants} implies that
\begin{equation}
\label{eq:estiamtes_remainders}
\norm*{(\delta \gamma^\mu_J)_h} \le (\mathrm{const}) \cdot \frac{\lambda^2 \cutoff^2}{M^2} 2^{h - N}, \qquad
\norm*{\sing{\delta g}{h}}_w \le (\mathrm{const}) \cdot \frac{\lambda^2 \cutoff^2}{M^2} 2^{h - N} \cdot
\begin{cases}
2^{3h}	&	w = \infty \\
2^{-h}	&	w = h
\end{cases}
\end{equation}
Every $\sing{\hat{g}}{h}, (\gamma^\mu_J)_h$ factor appearing in~\eqref{eq:triangle_multiscale} can be decomposed as shown above, so
\begin{equation}
\hat{\Gamma}^\mu_\triangle(p', p; \lambda) - \hat{\Gamma}^\mu_\triangle(p', p)  = \mathsf{R}^\mu(p', p; \lambda),
\end{equation}
where $\mathsf{R}^\mu(p', p; \lambda)$ is given by a sum of terms having the same structure as $\hat{\Gamma}^\mu(p', p; \phi_4^{0, 1, 2}, \lambda)$, except that they contain at least one $\sing{\delta g}{h}/(\delta \gamma^\mu_J)_h$ factor. Now consider one of such terms, say for instance
\begin{align*}
I^\mu(p', p; \lambda) 
& = - \sum_{h^\star < h_1 \le h_2} \int \frac{\dd[4]{q}}{(2\pi)^4} \, \ren[\Upsilon_\nu \, \sing{\delta \hat{g}}{h_1}(p' - q) \, (\gamma^\mu_J)_{h_2} \, \sing{\hat{g}}{h_2}(p - q) \, \Upsilon^\nu \, \hat{v}(q)] \\
& \equiv \sum_{h^\star < h_1 \le h_2} I_{h_1 h_2}(p', p; \lambda).
\end{align*}
Since $\abs*{p'}, \abs{p} \le m/2 \Rightarrow \abs*{p' - q - (p - q)} \le m$ and the single scale propagators are compactly-supported, $h_2 \in \lbrace h_1, h_1 + 1, h_1 + 2 \rbrace \cap [h^\star+1, N]$. This means that, apart from an inessential overall constant, we can perform our estimates by letting $h_2 = h_1$. If $h_1 < h_M \equiv \floor{\log_2 M}$, it is convenient to estimate the boson propagator as $\hat{v}(q) \le 1/M^2$, thus getting
\begin{align}
\label{eq:delta_low_regime}
\begin{split}
\sum_{h^\star < h_1 < h_M} \norm{I_{h_1 h_1}(p', p; \lambda)}
& \le \frac{C}{M^2} \frac{\lambda^2 \cutoff^2}{M^2} \sum_{h^\star < h_1 < h_M} 2^{3h_1 - h_1 + 2(h^\star - h_1)} 2^{h_1 - N} \\
& \le C \frac{m^2}{M^2} \frac{\lambda^2 \cutoff^2}{M^2} \frac{M}{\cutoff} \sum_{h^\star < h_1 < h_M} 2^{h_1 - h_1},
\end{split}
\end{align}
where in the second line we exploited the fact that $2^{h_1 - N} < M/\cutoff$. If $h_1 \ge h_M$, we have instead
\begin{align}
\label{eq:delta_high_regime}
\begin{split}
\sum_{h_M \le h_1 \le N} \norm{I_{h_1 h_1}(p', p; \lambda)}
& \le C \frac{\lambda^2 \cutoff^2}{M^2} \sum_{h_M \le h_1 \le N} 2^{3h_1 - h_1 + 2(h^\star - h_1) - 2h_1} 2^{h_1 - N} \\
& \le C m^2 \frac{\lambda^2 \cutoff^2}{M^2} \sum_{h_M \le h_1 \le N} 2^{-h_1-N},
\end{split}
\end{align}
where we used the fact that $\mathrm{supp} [\sing{\delta \hat{g}}{h_1}(p' - \cdot)] \subseteq \lbrace q \in \R^4 \colon \abs{q} \ge 2^{h_1-1} - m/2 \rbrace$ in order to write $\hat{v}(q) \le 1/q^2 \le (\mathrm{const}) \cdot 2^{-2h_1}$. By combining~\eqref{eq:delta_low_regime} and~\eqref{eq:delta_high_regime}, we obtain
\begin{equation}
\label{eq:regimes_combined}
\norm*{I^\mu(p', p; \lambda)} \le C \frac{m^2}{M^2} \frac{\lambda^2 \cutoff^2}{M^2}\frac{M}{\cutoff} \sum_{h^\star < h_1 < h_M} 1 + C m^2 \frac{\lambda^2 \cutoff^2}{M^2} \sum_{h_M \le h_1 \le N} 2^{-h_1-N}.
\end{equation}
After performing the sum over $h_1$ (and possibly redefining the constant $C$), this becomes
\begin{equation}
\norm*{I^\mu(p', p; \lambda)} \le C \frac{m^2}{M^2} \frac{\lambda^2 \cutoff^2}{M^2} \log \left( \frac{M}{m} \right) \frac{M}{\cutoff},
\end{equation}
so $I^\mu(p', p; \lambda)$ satisfies the bound~\eqref{eq:sum_over_scales_estimate}. The same conclusion holds for any other term contributing to $\mathsf{R}^\mu(p', p; \lambda)$ that contains \emph{one} $\sing{\delta g}{h}/(\delta \gamma^\mu_J)_h$ factor. If $\lambda^2 \cutoff^2/M^2 < 1$, terms containing more than one $\sing{\delta g}{h}/(\delta \gamma^\mu_J)_h$ factor are subdominant with respect to $I^\mu(p', p; \lambda)$, so the whole function $\mathsf{R}^\mu(p', p; \lambda)$ satisfies the bound~\eqref{eq:sum_over_scales_estimate}.
\end{proof}
\begin{rmk}
\label{rmk:optimal_bound}
Within the above proof, the boson propagator has been bounded with $1/M^2$ or $1/q^2$ depending on whether $h_1 < h_M$ or $h_1 \ge h_M$. This estimate is \emph{more refined} than the one presented in Section~\ref{subsec:bounds_on_propagator_and_vertex}, which is only based on the inequality $\hat{v}(q) \le 1/M^2$. This refinement is \emph{fundamental} in order to extract the $M/\cutoff$ factor displayed in~\eqref{eq:sum_over_scales_estimate}. Indeed, the short memory factor $2^{h_1 - N}$ coming from $\sing{\delta g}{h_1}$ produces a $M/\cutoff$ gain only provided that $h_1$ is less than $h_M$; in the opposite regime, the same gain is obtained by exploiting the damping produced by $\hat{v}(q)$ for large values of $\abs{q}$.
\end{rmk}

\subsection{Existence of the \texorpdfstring{$\cutoff \to +\infty$}{ctinf} limit of \texorpdfstring{$\hat{\Gamma}^\mu_\triangle(p', p)$}{g}}
\label{subsec:cutoff_removal_triangle}

Thanks to Lemma~\ref{lm:sum_over_scales}, from now on we can turn our attention on the $\lambda$-independent triangle integral $\hat{\Gamma}^\mu_\triangle(p', p)$. Since the $\sing{g}{h}_0$ propagator scales in the same way as $\sing{g}{h}$, the bound~\eqref{eq:bounds_on_S_and_Gamma} implies that
\begin{equation}
\label{eq:non_optimal_estimate}
\sup_{\abs*{p'}, \abs{p} \le m/2} \norm*{\hat{\Gamma}^\mu_\triangle(p', p)} \le C \cdot \frac{m^2}{M^2} \cdot \log^2 \left( \frac{\cutoff}{m} \right).
\end{equation}
so $\hat{\Gamma}^\mu_\triangle(p', p)$ could apparently \emph{diverge} as $\cutoff$ goes to $+\infty$. However, this is not the case, because the estimate~\eqref{eq:non_optimal_estimate} is not optimal. To improve it, it is necessary to exploit the decay of $\hat{v}(q)$ as $\abs{q} \to +\infty$, as we did in the proof of Lemma~\ref{lm:sum_over_scales}.
\begin{lm}
\label{lm:optimal_estimate}
The function $\hat{\Gamma}^\mu_{\triangle, \infty}(p', p) \equiv \lim_{\cutoff \to +\infty} \hat{\Gamma}^\mu_\triangle(p', p)$ is well-defined on $\mathcal{U} = \lbrace (p', p) \in \R^4 \times \R^4 \colon \abs*{p'}, \abs{p} \le m/2 \rbrace$ and it is smooth in the interior of $\mathcal{U}$. Moreover, for any $\ell \in \N_0$ there exists a positive constant $C_\ell$ such that
\begin{equation}
\label{eq:lm:optimal_estimate_convergence}
\norm*{m^\ell \partial^\ell \hat{\Gamma}^\mu_\triangle - m^\ell \partial^\ell \hat{\Gamma}^\mu_{\triangle, \infty}}_\mathcal{U} \le C_\ell \frac{m^2}{M^2} \frac{M^2}{\cutoff^2},
\end{equation}
where $\norm*{\, \cdot \,}_\mathcal{U}$ is the uniform norm on $\mathcal{U}$.
\end{lm}

\begin{proof}
Let $\hat{\Gamma}^\mu_{\triangle, N}(p', p)$ denote the triangle integral in presence of a finite cutoff $\cutoff = 2^N$. To prove that $\hat{\Gamma}^\mu_{\triangle, \infty}(p', p)$ exists, we show that the sequence of smooth functions $\lbrace \hat{\Gamma}^\mu_{\triangle, N} \rbrace_{N \in \N}$ has the Cauchy property with respect to $\norm*{\, \cdot \,}_\mathcal{U}$. 

For sake of simplicity, let us consider the contribution to $\hat{\Gamma}^\mu_{\triangle, N}(p', p)$ coming from the $\lambda$-independent version of diagram $\phi_4^1$, namely
\begin{align*}
I_N^\mu(p', p) 
& = - \sum_{h_1 = h^\star + 1}^N \sum_{h_2 = h_1}^N \int \frac{\dd[4]{q}}{(2\pi)^4} \, \ren[\Upsilon_\nu \, \sing{\hat{g}_0}{h_1}(p' - q) \, \gamma^\mu \, \sing{\hat{g}_0}{h_2}(p - q) \, \Upsilon^\nu \, \hat{v}(q)] \\
& = \sum_{h_1 = h^\star + 1}^N \sum_{h_2 = h_1}^N I^\mu_{h_1 h_2}(p', p).
\end{align*}
As we already noticed in the proof of Lemma~\ref{lm:sum_over_scales}, $h_2$ cannot exceed $h_1 + 2$, so we can let $h_1 = h_2$ at the price of introducing an inessential overall constant $c$. Given any $N' > N$, it is easy to see that
\begin{equation}
\norm*{I_N^\mu - I^\mu_{N'}}_\mathcal{U} \le c \sum_{h_1 = N-1}^{N'} \norm*{I^\mu_{h_1 h_1}}_\mathcal{U}.
\end{equation}
If $N \ge \floor{\log_2 M} \equiv h_M$, $h_1$ is greater than $h_M-1$, so it is convenient to estimate the boson propagator as $\hat{v}(q) \le 1/q^2$. In absence of $\sing{\delta g}{h_1}$ factors, the same procedure shown in~\eqref{eq:delta_high_regime} yields $\norm*{I^\mu_{h_1 h_1}(p', p)} \le m^2 \cdot 2^{-2h_1}$, so
\begin{equation}
\label{eq:cauchy_C0}
\norm*{I_N^\mu - I^\mu_{N'}}_\mathcal{U} \le C_1 m^2 \sum_{h_1 = N-1}^{N'} 2^{-2h_1} \le \frac{C_2 m^2}{2^{2N}}.
\end{equation}
The right hand side of this inequality becomes arbitrarily small if $N$ is large enough. Consequently, $\lbrace I^\mu_N \rbrace_{N \in \N}$ is a Cauchy sequence that converges uniformly to a continuous function $I^\mu_{\infty}(p', p)$ defined on $\mathcal{U}$. The same argument can be applied to the contributions coming from $\phi_4^0$ and $\phi_4^2$; in particular, $\hat{\Gamma}^\mu_{\triangle, \infty}(p', p)$ exists and it is continuous. Since the estimate~\eqref{eq:cauchy_C0} holds for the $\ell$-th derivatives of $m^\ell \hat{I}^\mu_N(p', p)$ as well (the $m^\ell$ factor balances the $2^{-\ell h^\star}$ factor produced by a $\ell$-th derivative), the sequence $\lbrace \hat{\Gamma}^\mu_{\triangle, N} \rbrace_{N \in \N}$ converges with respect to \emph{any} $C^\ell(\mathcal{U})$ norm for every $\ell \ge 0$. This implies that $\hat{\Gamma}^\mu_{\triangle, \infty}(p', p)$ is smooth and the $\ell$-th derivatives of $\hat{\Gamma}^\mu_{\triangle, N}(p', p)$ converge uniformly to the $\ell$-th derivatives of $\hat{\Gamma}^\mu_{\triangle, \infty}(p', p)$.

The difference between $I^\mu_N(p', p)$ and $I^\mu_\infty(p', p)$ can be estimated as
\begin{equation}
\norm*{I^\mu_N- I^\mu_\infty}_\mathcal{U} \le C \sum_{h_1 = N}^{+\infty} \norm*{I^\mu_{h_1 h_1}}_\mathcal{U}.
\end{equation}
Here $h_1$ is much larger than $h_M$, so we use the inequality $\hat{v}(q) \le 1/q^2$ to write $\norm*{I^\mu_{h_1 h_1}}_\mathcal{U} \le m^2 \cdot 2^{-2h_1}$. The bound~\eqref{eq:lm:optimal_estimate_convergence} then follows by summing over $h_1$ (again, this can be extended the whole triangle integral, for any value of $\ell$).
\end{proof}

\subsection{Evaluation of \texorpdfstring{$\mathcal{A}_2(z)$}{gyr}}
\label{subsec:A_evaluation_first_order}

Thanks to Lemmas~\ref{lm:sum_over_scales} and~\ref{lm:optimal_estimate}, we can write 
\begin{equation}
\label{eq:A_2_decomposition}
\mathcal{A}_2(z; \lambda) = \mathcal{A}_2(z) + R(z; \lambda),
\end{equation} 
where the function $\mathcal{A}_2(z)$ has the same form as $\mathcal{A}_2(z; \lambda)$ with $\hat{\Gamma}^\mu_{\triangle, \infty}(p', p)$ in place of $\hat{\Gamma}^\mu_2(p', p)$ and
\begin{equation}
\label{eq:R_estimate_A_2}
\abs{R(z; \lambda)} \le C_1 \cdot \frac{m^2}{M^2} \cdot \frac{\lambda^4 \cutoff^2}{M^2} \log \left( \frac{M}{m} \right) \frac{M}{\cutoff} + C_2 \cdot \lambda^2 \frac{m^2}{M^2} \frac{M^2}{\cutoff^2}.
\end{equation}
The explicit form of $\hat{\Gamma}^\mu_{\triangle, \infty}(p', p)$ is
\begin{align}
\label{eq:lambda_independent_triangle_limit}
\begin{split}
\hat{\Gamma}^\mu_{\triangle, \infty}(p', p) 
& = \lim_{\cutoff \to +\infty} \left(- \int \frac{\dd[4]{q}}{(2\pi)^4} \, \Upsilon_\nu \sing{\hat{g}}{\le N}_0(p' - q) \, \gamma^\mu \, \sing{\hat{g}}{\le N}_0(p - q) \Upsilon^\nu - L^\mu \right) \\
& \equiv \lim_{\cutoff \to +\infty} ( \hat{\mathscr{G}}^\mu_{\triangle, \cutoff}(p', p) - L^\mu),
\end{split}
\end{align}
where $L^\mu$ denotes the local term that comes from the action of the renormalization operator (see~\eqref{eq:triangle_multiscale}, recalling that $\ren = 1 - \loc$). Moreover,
\begin{multline}
\label{eq:gfunction_first_order}
\mathcal{A}_2(z) = \frac{\lambda^2}{24} \biggl( z \epsilon_{a b c} \tr[\gamma^5 (\partial_{p'} - \partial_p)^a \hat{\Gamma}^b_{\triangle, \infty}(p_z, p_z) \gamma^c]\, + \\
+ 2 \tr[\gamma_a \hat{\Gamma}^a_{\triangle, \infty}(p_z, p_z)] - 6 \tr[(1 + \gamma^0) \hat{\Gamma}^0_{\triangle, \infty}(p_z, p_z)] \biggr).
\end{multline}
It is easy to see that $L^\mu$ does not contribute to $\rgyr$, because it is equal to a linear combination of $\gamma^\mu$ and $\gamma^5 \gamma^\mu$ (the absence of such contributions has already been discussed in Section~\ref{subsec:non_triangle_vanishing}). Therefore, we can rewrite~\eqref{eq:gfunction_first_order} in terms of $\hat{\mathscr{G}}^\mu_{\triangle, \cutoff}(p', p)$ alone,
\begin{multline}
\label{eq:gfunction_first_order_no_local_part}
\mathcal{A}_2(z) = \lim_{\cutoff \to +\infty} \frac{\lambda^2}{24} \biggl( z \epsilon_{a b c} \tr[\gamma^5 (\partial_{p'} - \partial_p)^a \hat{\mathscr{G}}^b_{\triangle, \cutoff}(p_z, p_z) \gamma^c]\, + \\
+ 2 \tr[\gamma_a \hat{\mathscr{G}}^a_{\triangle, \cutoff}(p_z, p_z)] - 6 \tr[(1 + \gamma^0) \hat{\mathscr{G}}^0_{\triangle, \cutoff}(p_z, p_z)] \biggr),
\end{multline}
and then plug~\eqref{eq:lambda_independent_triangle_limit} into~\eqref{eq:gfunction_first_order_no_local_part}. A straightforward calculation yields
\begin{align}
\label{eq:A_cutoff}
\begin{split}
\mathcal{A}_2(z) 
& = 4 \lambda^2 \lim_{\cutoff \to +\infty} \int \frac{\dd[4]{q}}{(2\pi)^4} \frac{T_z(q) \chi(q) \chi^2(q - p_z)}{[q^2 + M^2][(p_z - q)^2 + m^2]^2}, \\
T_z(q) 
& \equiv (\kappa^2 + 1)[2 z^2 - 3 z q^0 + (q^0)^2 - q_a q^a/3] + 2im(\kappa^2 - 1)(z - q^0),
\end{split}
\end{align}
where a sum over the spatial index $a \in \lbrace 1, 2, 3 \rbrace$ is understood. This can be further simplified thanks to the identity
\begin{equation}
\label{eq:feynman_parametrization}
\frac{1}{ab^2} = 2 \int_0^1 \dd{x} \frac{x}{[a(1 - x) + bx]^3} \quad \forall a, b \in \C \setminus \lbrace 0 \rbrace,
\end{equation}
known as \emph{Feynman's parametrization}~\cite[Appendix B.1.1]{Schwartz_2013}: by letting $a \equiv q^2 + M^2, b \equiv (p_z - q)^2 + m^2$ and subsequently performing the variable change $q \mapsto q + p_z x$,~\eqref{eq:A_cutoff} becomes
\begin{equation}
\label{eq:limit_shifted_cutoffs}
\mathcal{A}_2(z) = 8 \lambda^2 \lim_{\cutoff \to +\infty} \int_0^1 \dd{x} \int \frac{\dd[4]{q}}{(2\pi)^4} \, \frac{x T_z(q + x p_z)}{[q^2 + \Delta^2(x, z)]^3} \, \chi(q + x p_z) \chi^2(q + (x - 1)p_z),
\end{equation}
where $\Delta^2(x, z) \equiv (1 - x)M^2 + z^2x(1-x) + m^2 x$. Due to the fact that the polynomial $T_z(q)$ contains the quadratic terms $(q^0)^2 - q_a q^a/3$ (see~\eqref{eq:A_cutoff}), the function $f_{x, z}(q) \equiv x T_z(q + x p_z)[q^2 + \Delta^2(x, z)]^{-3}$ decays as $1/q^4$ when $\abs{q} \to +\infty$; hence, although the limit~\eqref{eq:limit_shifted_cutoffs} exists, it \emph{cannot} be recklessly brought under the integral sign. This problem is circumvented by means of a subtle cancellation arising from rotational symmetry, as discussed below.

Fix an arbitrary $c \in \R^4$ and consider the difference $\chi(q + c) - \chi(q)$. Then
\begin{equation}
\label{eq:chi_estimate}
\abs{\chi(q + c) - \chi(q)} \le \abs{c} \cdot \norm*{\partial \chi}_\infty \, \id_\mathcal{B}(q) \le (\mathrm{const}) \cdot \frac{\id_\mathcal{B}(q)}{\cutoff},
\end{equation}
where $\mathcal{B} \equiv \lbrace q \in \R^4 \colon \cutoff - \abs{c} \le \abs{q} \le 2 \cutoff + \abs{c} \rbrace$ and the relation $\norm{\partial \chi}_\infty = 2^{-N} \norm{\partial \chi_0}_\infty$ has been exploited. Now, since $\abs{f_{x, z}(q)} \asymp 1/q^4$, we have
\begin{equation}
\label{eq:chi_estimate_function}
\int \frac{\dd[4]{q}}{(2\pi)^4} \, \abs{f_{x, z}(q)} \, \abs{\chi(q + c) - \chi(q)} \lesssim \frac{(\mathrm{const})}{\cutoff} \int_\cutoff^{2 \cutoff} \frac{q^3\dd{q}}{q^4} \lesssim \frac{1}{\cutoff} \overset{\cutoff \to +\infty}{\longrightarrow} 0.
\end{equation} 
Based on~\eqref{eq:chi_estimate_function}, the limit~\eqref{eq:limit_shifted_cutoffs} can be \emph{equivalently} computed with $\chi^3(q)$ in place of $\chi(q + x p_z) \chi^2(q + (x - 1)p_z)$, namely
\begin{equation}
\label{eq:limit_after_centering}
\mathcal{A}_2(z) = 8 \lambda^2 \lim_{\cutoff \to +\infty} \int_0^1 \dd{x} \int \frac{\dd[4]{q}}{(2\pi)^4} \, \frac{x T_z(q + x p_z)}{[q^2 + \Delta^2(x, z)]^3} \, \chi^3(q).
\end{equation} 
Thanks to the rotational symmetry of the function $\chi^3(q)[q^2 + \Delta^2(x, z)]^{-3}$, we have
\begin{equation}
\int \frac{\dd[4]{q}}{(2\pi)^4} \frac{(q^0)^2 - q_a q^a/3}{[q^2 + \Delta^2(x, z)]^3} \, \chi^3(q) = 0,
\end{equation} 
so the dangerous quadratic terms occurring in the numerator of~\eqref{eq:limit_after_centering} \emph{cancel out} exactly. We stress that this argument could not be applied to~\eqref{eq:limit_shifted_cutoffs}, because the function that multiplies $f_{x, z}(q)$ inside~\eqref{eq:limit_shifted_cutoffs} is not rotationally invariant.

In absence of the above quadratic terms, Lebesgue's dominated convergence theorem holds. If the limit is brought under the integral sign, an elementary integration gives
\begin{equation}
\label{eq:gfunction_pert}
\mathcal{A}_2(z) = \frac{\lambda^2}{4\pi^2} \int_0^1 \dd{x} \, \frac{z^2 (\kappa^2 + 1)x(x^2 - 3x + 2) + 2imz(\kappa^2 - 1) x(1-x)}{(1 - x)M^2 + z^2x(1-x) + m^2 x}.
\end{equation}

The function $\mathcal{A}_2(z)$ admits a unique holomorphic extension $\tilde{\mathcal{A}}_2(z)$ defined inside an open disk of radius $M/2$ centered around $z = 0$; in particular, this disk contains the point $z = im$. This is true because the denominator of~\eqref{eq:gfunction_pert} cannot vanish as long as $x \in [0, 1]$ and $\abs*{z} \le M/2$. According to Cauchy's integral formula, we have
\begin{equation}
\label{eq:cauchy_estimates}
\abs*{m^\ell \tilde{\mathcal{A}}_2^{(\ell)}(0)} = \abs{\ell! \, m^\ell \oint_{\abs{z} = M/2} \frac{\dd{z}}{2\pi i} \frac{\tilde{\mathcal{A}}_2(z)}{z^{\ell + 1}}} \le \ell! \left( \frac{2m}{M} \right)^\ell \max_{\theta \in [0, 2\pi]} \abs*{\tilde{\mathcal{A}}_2((M/2)e^{i \theta})},
\end{equation}
so
\begin{equation}
\abs*{m^\ell \tilde{\mathcal{A}}_2^{(\ell)}(0)} \le C \lambda^2 \, \ell! \left( \frac{2m}{M} \right)^\ell.
\end{equation}
Since $K = 4$ and $M > 10m$, this implies that
\begin{equation}
\abs{\sum_{\ell = 0}^K \frac{(im)^\ell}{\ell!} \mathcal{A}^{(\ell)}_2(0) - \tilde{\mathcal{A}}_2(im)} \le C \lambda^2 \frac{m^4}{M^4}.
\end{equation}
Finally, an easy calculation shows that
\begin{equation}
\label{eq:A_expansion}
\tilde{\mathcal{A}}_2(im) = \frac{m^2}{M^2} \, \frac{\lambda^2}{4\pi^2} \frac{1- 5 \kappa^2}{3} + \BigO \left( \frac{\lambda^2 m^4}{M^4} \right) = \bar{\gyr}_{\textsc{z}, 1}(1 + \BigO(m^2/M^2));
\end{equation} 
consequently, there exists a $(\lambda, M, m, \cutoff)$-independent constant $C_0$ such that
\begin{equation}
\abs{\sum_{\ell = 0}^K \frac{(im)^\ell}{\ell!} \mathcal{A}^{(\ell)}_2(0) - \bar{\gyr}_{\textsc{z}, 1}} \le C_1 \lambda^2 \frac{m^4}{M^4}.
\end{equation}

\subsection{Proof of Theorem~\ref{th:main_theorem}}
\label{subsec:main_theorem_proof}

According to definition~\eqref{eq:g-2:def}, we can write
\begin{equation}
\label{eq:holomorphic_estimate}
\rgyr = \sum_{\ell = 0}^K \frac{(im)^\ell}{\ell!} \, [\mathcal{A}^{(\ell)}_0(0; \lambda) + \mathcal{A}^{(\ell)}_2(0; \lambda) + \mathcal{A}^{(\ell)}_{>2}(0; \lambda)].
\end{equation}
We know from Section~\ref{subsec:non_triangle_vanishing} that $\mathcal{A}_0(z; \lambda)$ is identically zero. Moreover, relying on the fact that $\mathcal{A}(z) = [\mathcal{F}(z) - \mathcal{Q}(z)]/\mathcal{Q}(z)$ with $\mathcal{Q}(z) = 24 + \sum_{n \ge 1} \lambda^{2n} \mathcal{Q}_n(z; \lambda)$, Theorem~\ref{th:convergent_expansions} implies that $m^\ell \abs*{\mathcal{A}^{(\ell)}_{>2}(0; \lambda)}$ is bounded by the right hand side of~\eqref{eq:bounds_on_S_and_Gamma} with $n = 2$ and a $\ell$-dependent overall constant. Finally, $m^\ell \mathcal{A}^{(\ell)}_2(z; 0)$ can be decomposed as in~\eqref{eq:A_2_decomposition} and the remainder $m^\ell R^{(\ell)}(z; \lambda)$ is estimated as in~\eqref{eq:R_estimate_A_2} with a $\ell$-dependent overall constant. Based on~\eqref{eq:A_expansion}, we obtain $\rgyr = \bar{\gyr}_{\textsc{z}, 1}(1 + R_\lambda)$, where
\begin{equation}
\abs*{R_\lambda} \le C_1 \frac{M^2}{\cutoff^2} + C_2 \frac{m^2}{M^2} + C_3 \frac{\lambda^2\cutoff^2}{M^2} \cdot \log^4 \left(\frac{M}{m}\right) \log^4 \left( \frac{\cutoff}{M} \right) + C'_3 \frac{\lambda^2\cutoff^2}{M^2} \log \left( \frac{M}{m} \right) \frac{M}{\cutoff}.
\end{equation}
Note that $C'_3 (\lambda^2 \cutoff^2/M^2) \log (M/m) \cdot (M/\cutoff)$ is subdominant with respect to the remainder multiplied by $C_3$, so it can be reabsorbed inside it. This completes the proof of Theorem~\ref{th:main_theorem}.

\section*{Acknowledgments}

We acknowledge support from the MUR, PRIN 2022 project MaIQuFi cod.\ 20223J85K3. This work has been carried out under the auspices of the GNFM of INdAM.  

\appendix

\section{The regularized anomalous gyromagnetic factor}
\label{app:g-2_definition}
In this section we provide some details about the definition of regularized anomalous gyromagnetic factor given by~\eqref{eq:g-2:def}. In the context of perturbative Quantum Field Theory, the anomalous gyromagnetic factor is defined in Minkowski spacetime starting from the probability amplitude for a process in which a single fermion is scattered by a weak, external electromagnetic field $A_\mu(x)$. This amplitude reads
\begin{equation}
\label{eq:scattering_amplitude}
- \frac{1}{2\sqrt{p_0 \, p'_0}} \bar{u}_{p' \xi'} \, \hat{\Gamma}^\mu_\mink(p', p) \, u_{p \xi} \, \hat{A}_\mu(k) \eval_{(p')^2 = p^2 = - m^2}
\end{equation}
within the first order in $A_\mu$, where $\hat{\Gamma}^\mu_\mink(p', p)$ is the Minkowskian amputated vertex function (see e.\ g.~\cite[Equation (6.30)]{Peskin_1995}). Here, $k = p' - p$, $p^2 \equiv -(p_0)^2 + \abs{\bm{p}}^2$ and $\lbrace u_{p \xi} \rbrace_{\xi = \pm}$ are two linearly independent solutions of the equation $(i p_\mu \gamma^\mu_\mink + \mnorm{p}) u_{p \xi} = 0$, where $\gamma^\mu_\mink$ are the Minkowskian gamma matrices, $\mnorm{p} \equiv \sqrt{-p^2}$ and $p^2 \equiv -(p^0)^2 + \abs{\mathbf{p}}^2$. Explicitly,
\begin{equation}
\C^4 \ni u_{p \xi} = 
\begin{pmatrix}
\sqrt{-p^0 + i\mathbf{p} \cdot \bm{\sigma}} \, e_\xi \\
\sqrt{-p^0 - i\mathbf{p} \cdot \bm{\sigma}} \, e_\xi
\end{pmatrix}
\end{equation}
where $\xi \in \lbrace -1/2, 1/2 \rbrace$ labels the helicity states of the fermion and $\lbrace e_{-1/2}, e_{1/2} \rbrace \subseteq \C^2$ is the canonical basis of $\C^2$. The constraint $p^2=-m^2$ occurring in~\eqref{eq:scattering_amplitude} is called \emph{mass shell} condition. 
Lorentz symmetry constraints the vertex matrix element to take the form~\cite{Nowakowski_2005}
\begin{equation}
\label{eq:vertex_parametrization}
\bar{u}_{p' \xi'} \, \hat{\Gamma}^\mu_\mink \, u_{p \xi} =
\bar{u}_{p' \xi'} \, 
\biggl[ \gamma^\mu_\mink [F + \gamma^5 F_5]- \frac{i(p'+p)^\mu}{\mnorm{p'}+\mnorm{p}} [G + \gamma^5 G_5]
+\frac{(p'-p)^\mu}{\mnorm{p'}+\mnorm{p}} [H + \gamma^5 H_5] \biggr] u_{p \xi},
\end{equation}
where $F, F_5, \dots, H_5$ are called \emph{form fators} and they depend on $\mnorm{p'}, \mnorm{p}, \mnorm{k}$. In the non interacting case ($\lambda=0$), one has  $F=1$, $F_5, G, G_5, H, H_5 = 0$ and $\hat{\Gamma}^\mu_\mink(p', p) = \gamma^\mu_\mink$. 

If $\mathbf{k} = \mathbf{p}' - \mathbf{p}$ is small with respect to $m$, the parametrization~\eqref{eq:vertex_parametrization} implies that~\eqref{eq:scattering_amplitude} takes the form
\begin{equation}
\label{eq:low_energy_scattering}
(-i) [F + G](m, m, 0) \cdot \biggl[ \delta_{\xi' \xi} \hat{A}^0(\mathbf{k}) +\frac{2F(m, m, 0)}{[F + G](m, m, 0)} \cdot \left(-\frac{1}{2m} \mathbf{S}_{\xi' \xi} \cdot \hat{\mathbf{B}}(\mathbf{k}) \right) \biggr] + \cdots,
\end{equation}
where $\mathbf{S} = (\sigma_1, \sigma_2, \sigma_3)/2$ are the standard spin-$1/2$ matrices and $\hat{\mathbf{B}}(\mathbf{k}) \equiv i \mathbf{k} \times \hat{\mathbf{A}}(\mathbf{k})$ is the Fourier transform of the external magnetic field. The neglected terms, denoted by $(\, \cdots)$, are either higher orders in $\mathbf{p}, \mathbf{p}'$ or terms that do not contribute to the gyromagnetic factor~\cite{Nowakowski_2005}. To derive~\eqref{eq:low_energy_scattering}, one uses the basic relation
\begin{equation}
\bar{u}_{p' \xi'} \gamma^\mu_\mink \, u_{p \xi}=
 \bar{u}_{p' \xi'} \left[ \frac{(p' + p)^\mu}{2m} - \frac{[\gamma^\mu_\mink, \gamma^\nu_\mink]}{4m} \, k_\nu \right] u_{p \xi},
\end{equation}
which, once plugged into~\eqref{eq:vertex_parametrization}, shows that a term proportianal to $\mathbf{S}_{\xi' \xi} \cdot \hat{\mathbf{B}}$ arises even in the non-interacting case. 
Recalling that the gyromagnetic factor is the coefficient $g$ appearing inside the relation $\bm{\mu} = -g (e/2m) \mathbf{S}$, the usual perturbative definition of anomalous gyromagnetic factor can be read off from~\eqref{eq:low_energy_scattering} (see~\cite[Equation (6.37)]{Peskin_1995}):
\begin{equation}
\label{eq:g_form_factors}
\gyr = \frac{-G(m, m, 0)}{[F + G](m, m, 0)}.
\end{equation}
The form factors can be expressed in terms of the amputated vertex function as
\begin{align}
\label{eq:F_factor}
F(\mnorm{p}, \mnorm{p}, 0) & = \frac{i\epsilon_{\alpha \mu \sigma \beta}}{48 \mnorm{p}^2} p^\sigma \, (\partial_{p'} - \partial_p)^\alpha \tr[\gamma^5 \mproj{p'} \, \hat{\Gamma}_\mink^\mu(p', p) \, \mproj{p} \gamma^\beta_\mink] \eval_{p, p} \\
\label{eq:F_G_factor}
(F + G)(\mnorm{p}, \mnorm{p}, 0) & = -\frac{i}{4 \mnorm{p}^2} \tr[\mproj{p} \, p_\mu \hat{\Gamma}^\mu_\mink(p, p)]
\end{align}
where $\mproj{p} = -i \slashed{p} + \mnorm{p}$ and $\epsilon_{\alpha \mu \sigma \beta}$ is the Levi-Civita symbol. A proof of these relations will be given below.

As a consequence of the Lorentz covariance of the formal continuum theory, both $F$ and $G$ depend on $p$ only through 
$\mnorm{p}$, so we can evaluate them at $p = (\mnorm{p}, 0, 0, 0)$ without loss of generality. With this choice, a combination of~\eqref{eq:g_form_factors} and~\eqref{eq:F_factor} gives a \emph{formal} characterization of $\gyr$ in terms of the Minkowskian amputated vertex function; namely, a straightforward computation shows that $\gyr = \mathcal{A}_\mink(m)$, where
\begin{equation}
\mathcal{A}_\mink(s) \equiv \frac{i s \epsilon_{a b c} \tr[\gamma^5 (1 - i \gamma^0_\mink)(\partial_{p'} - \partial_p)^a \hat{\Gamma}^b_\mink(p_s, p_s) \gamma^c_\mink] + 2 \tr[ (\gamma_\mink)_a \hat{\Gamma}^a_\mink(p_s, p_s) ]}{-6i \tr[(1 - i \gamma^0_\mink) \hat{\Gamma}^0_\mink(p_s, p_s)]} - 1
\end{equation}
and $p_s \equiv (s, 0, 0, 0)$. Recalling that $\gamma^0_\mink = i \gamma^0$ and $\gamma^a_\mink = \gamma^a \,\, \forall a = 1, 2, 3$, the Euclidean function $\mathcal{A}(z)$ defined in~\eqref{eq:noncovariant_gfunction} is obtained by applying the formal Wick rotation
\begin{equation}
s \mapsto -iz, \quad \hat{\Gamma}^0_\mink(p_s, p_s) \mapsto i\hat{\Gamma}^0(p_z, p_z), \quad \hat{\Gamma}^a_\mink(p_s, p_s) \mapsto\hat{\Gamma}^a(p_z, p_z) \,\,\, \forall a = 1, 2, 3
\end{equation} 
on $\mathcal{A}_\mink(s)$. Our non perturbative definition of $\rgyr$ consists in the $K$-th MacLaurin polynomial of $\mathcal{A}(z)$ evaluated at $z = im$, as the rest is subdominant in $m^2/M^2$.

We finally give a derivation of~\eqref{eq:F_factor} and~\eqref{eq:F_G_factor}. As a preliminary step, it is convenient to multiply both sides of~\eqref{eq:vertex_parametrization} by $u_{p' \xi'}$ on the left and by $\bar{u}_{p \xi}$ on the right and subsequently sum over $\xi', \xi \in \lbrace -1/2, 1/2 \rbrace$. Thanks to the identity $\sum_{\xi=\pm 1/2} u_{p \xi} \bar{u}_{p \xi} = -i\slashed{p} + \mnorm{p} \equiv \mproj{p}$, this converts~\eqref{eq:vertex_parametrization} into 
\begin{equation}
\label{eq:form_factors_matrix}
\mproj{p'} \, \hat{\Gamma}^\mu_\mink(p', p) \, \mproj{p} = \mproj{p'} \, X^\mu_\mink(p', p) \, \mproj{p},
\end{equation}
where $X^\mu_\mink(p', p)$ denotes the matrix appearing at the right hand side of~\eqref{eq:vertex_parametrization}.

To prove~\eqref{eq:F_G_factor}, we take the trace of both sides of~\eqref{eq:form_factors_matrix} and evaluate them at $p = p'$, thus getting
\begin{equation}
\label{eq:F_G_trace}
\tr[\mproj{p} \, \hat{\Gamma}^\mu_\mink(p, p) \, \mproj{p}] = \tr \left[ \, \mproj{p} \! \left( \gamma^\mu_\mink \, F(\mnorm{p}, \mnorm{p} , 0) - \frac{ip^\mu}{\mnorm{p}} G(\mnorm{p}, \mnorm{p}, 0) \right) \mproj{p} \right].
\end{equation}
Here, $H, H_5$ are absent because they are multiplied by $(p - p)^\mu = 0$ and $G_5$ and $F_5$ are wiped away because $\tr[ \mproj{p} \, \gamma^5 \gamma^\mu_\mink \, \mproj{p} ] = \tr[ \mproj{p} \, \gamma^5 \, \mproj{p} ] = 0$. By exploiting the trivial matrix identity $\mproj{p} \gamma^\mu_\mink \mproj{p} = - \mproj{p} \, i p^\mu \, \mproj{p}/\mnorm{p}$ and the cyclicity of the trace, we can rewrite~\eqref{eq:F_G_trace} as
\begin{equation}
\label{eq:F_G_normalization_quasi_final}
\tr[\hat{\Gamma}^\mu_\mink(p, p) \, \mproj{p}^2] = - \frac{i p^\mu}{\mnorm{p}} (F + G)(\mnorm{p}, \mnorm{p}, 0) \tr[ \mproj{p}^2 ];
\end{equation}
knowing that $p^2 = -\mnorm{p}^2$, $\mproj{p}^2 = 2 \mnorm{p} \mproj{p}$ and $\tr[\mproj{p}^2] = 8 \mnorm{p}^2$, formula \eqref{eq:F_G_factor} follows after contracting both sides of~\eqref{eq:F_G_normalization_quasi_final} with $p_\mu$.

We now prove~\eqref{eq:F_factor}. At first, we verify that the right hand side of this equality cannot depend on any form factor different from $F$, nor on the derivatives of $F$ itself. For instance, we see from~\eqref{eq:form_factors_matrix} that $G$ contributes with the term
\begin{equation}
\frac{i\epsilon_{\alpha \mu \sigma \beta}}{48 \mnorm{p}^2} p^\sigma (\partial_{p'} - \partial_p)^\alpha \tr[\gamma^5 \left( \mproj{p'} \, \frac{-i(p' + p)^\mu}{\mnorm{p'} + \mnorm{p}} \, \mproj{p} \right) \gamma^\beta_\mink] \eval_{p, p} G(\mnorm{p}, \mnorm{p}, 0),
\end{equation}
which \emph{vanishes} because it can be written as a sum of traces of the form $\tr[\gamma^5 \gamma^{a_1}_\mink \cdots \gamma^{a_n}_\mink]$ with $n < 4$. An immediate extension of this argument leads to the conclusion that the right hand side of~\eqref{eq:F_factor} does not receive any contribution from $G, H, \partial F, \partial G, \partial H$. The form factors $F_5, G_5, H_5$ and their derivatives are instead cancelled due to the complete antisymmetry of the Levi-Civita symbol. To see why this is the case, let us consider the term proportional to $F_5$, which is
\begin{equation}
\label{eq:F5_contribution}
\frac{i\epsilon_{\alpha \mu \sigma \beta}}{48 \mnorm{p}^2} p^\sigma (\partial_{p'} - \partial_p)^\alpha \tr[\gamma^5 \, \mproj{p'} \, \gamma^\mu_\mink \gamma^5 \, \mproj{p} \, \gamma^\beta_\mink] \eval_{p, p} F_5(\mnorm{p}, \mnorm{p}, 0).
\end{equation}
Thanks to the properties of the $\gamma^5$ matrix~\cite[Appendix ``Traces'']{Weinberg_1995}, we can write
\begin{align}
\label{eq:axial_form_factor_vanishing}
\begin{split}
(\partial_{p'} - \partial_p)^\alpha \tr[\gamma^5 \, \mproj{p'} \, \gamma^\mu_\mink \gamma^5 \, \mproj{p} \, \gamma^\beta_\mink] \eval_{p, p} 
& = (\partial_{p'} - \partial_p)^\alpha \tr[(i \slashed{p}' + \abs*{p'}_\mink) \gamma^\mu_\mink (i \slashed{p} - \abs*{p}_\mink) \gamma^\beta_\mink] \eval_{p, p}\\
& = i\tr[\gamma^\alpha_\mink \, \gamma^\mu_\mink \, \slashed{p} \, \gamma^\beta_\mink] -i\tr[\slashed{p} \, \gamma^\mu_\mink \, \gamma^\alpha_\mink \, \gamma^\beta_\mink].
\end{split}
\end{align}
If $\eta = \mathrm{diag} (-1, 1, 1, 1)$ denotes the Lorentzian metric tensor, the trace of a product of $4$ gamma matrices satisfies $\tr \, [\gamma^\alpha_\mink \, \gamma^\beta_\mink \, \gamma^\mu_\mink \, \gamma^\nu_\mink] = 4(\eta^{\alpha \beta} \eta^{\mu \nu} - \eta^{\alpha \mu} \eta^{\beta \nu} + \eta^{\alpha \nu} \eta^{\beta \mu})$, so both the traces appearing in the last line of~\eqref{eq:axial_form_factor_vanishing} vanish when they are contracted with the totally antisymmetric tensor $\epsilon_{\alpha \mu \sigma \beta}$; this causes the whole term~\eqref{eq:F5_contribution} to vanish. A similar conclusion can be drawn for $G_5, H_5, \partial F_5, \partial G_5, \partial H_5$.

According to the above considerations, the right hand side of~\eqref{eq:F_factor} consists in a \emph{single} term which is proportional to $F(\mnorm{p}, \mnorm{p}, 0)$, namely
\begin{equation}
\label{eq:F_factor_RHS_surviving}
\frac{i\epsilon_{\alpha \mu \sigma \beta}}{48 \mnorm{p}^2} p^\sigma \, (\partial_{p'} - \partial_p)^\alpha \tr[\gamma^5 \mproj{p'} \, \gamma^\mu_\mink \, \mproj{p} \gamma^\beta_\mink] \eval_{p, p} F(\mnorm{p}, \mnorm{p}, 0);
\end{equation}
since $\tr[\gamma^5 \mproj{p'} \, \gamma^\mu_\mink \, \mproj{p} \gamma^\beta_\mink] = 4i \epsilon^{\tau \mu \nu \beta} (p')_\tau p_\nu$, this is also equal to
\begin{equation}
\frac{i\epsilon_{\alpha \mu \sigma \beta}}{48 \mnorm{p}^2} p^\sigma \cdot 4i \epsilon^{\tau \mu \nu \beta} \left( (\partial_{p'} - \partial_p)^\alpha (p')_\tau p_\nu \eval_{p, p} \right) \,  F(\mnorm{p}, \mnorm{p}, 0).
\end{equation}
After computing the derivatives and exploiting the identity $\epsilon_{\alpha \mu \sigma \beta} \epsilon^{\alpha \mu \nu \beta} = 6 \delta^\nu_\sigma$, the right hand side of~\eqref{eq:F_factor} reduces to
\begin{equation}
\frac{i p^2}{48 \mnorm{p}^2} \cdot (8i \cdot 6) F(\mnorm{p}, \mnorm{p}, 0) = F(\mnorm{p}, \mnorm{p}, 0),
\end{equation}
so~\eqref{eq:F_factor} holds.

\section{Proof of Lemma~\texorpdfstring{\ref{lm:localization_request}}{eq}}
\label{app:localization_request_proof}

Preliminarily, we show that $\beta^{m, s}_h \vert_{m_N = 0} = m^s_h \vert_{m_N = 0} = 0$ on all scales. The key observation is that the effective potential $\sing{\mathscr{V}}{N-1}[\psi, \omega] \vert_{m_N = 0}$ is manifestly invariant under the global chiral $\mathrm{U}(1)$ transformation $\hat{\psi}^\epsilon_{k s} \mapsto e^{i\epsilon \alpha_s} \hat{\psi}^\epsilon_{k s}, \, \hat{\eta}^\epsilon_{k s} \mapsto e^{i \epsilon \alpha_s} \hat{\eta}^\epsilon_{k s}$, where $\alpha_+$ and $\alpha_-$ are two independent phases. The term $-\beta^{m, s}_{N-1} \vert_{m_N = 0} (\psi^+_s \psi^-_{-s})(x)$ is not invariant under this transformation, so it must be $\beta^{m, s}_{N-1} \vert_{m_N = 0} = 0$ and therefore $m^s_{N-2} \vert_{m_N = 0} = 0$ (recall that $m^s_{N-1} = m^s_N$). This implies that $\sing{\mathscr{V}}{N-2}[\psi, \omega] \vert_{m_N = 0}$ is still invariant under the global chiral $\mathrm{U}(1)$ transformation introduced above, so both $\beta^{m, s}_{N-2} \vert_{m_N = 0}$ and $m^s_{N-3} \vert_{m_N = 0}$ vanish. By iterating this argument, one infers that $\beta^{m, s}_h \vert_{m_N = 0} = m^s_{h-1} \vert_{m_N = 0} = 0$ for every $h = h^\star, \dots, N$.

Now we turn our attention on~\eqref{eq:localization_request}. Since the theory is translationally invariant, the first line of~\eqref{eq:localization_definition} is equal to
\begin{equation}
\int_{\mathcal{M}_L^2} \!\! \dd[4]{x} \dd[4]{y} \loc \left[ W^{\alpha \beta}(x,y) \, \psi^+_\alpha(x) \psi^-_\beta(y) \right] =
\int_{\mathcal{M}_L} \! \dd[4]{x} (\psi^+_\alpha [b^{\alpha \beta} + c^{\mu \alpha \beta} \partial_\mu + d^{\mu \nu \alpha \beta} \partial_\mu \partial_\nu] \psi^-_\beta)(x)
\end{equation}
with
\begin{gather*}
\begin{split}
b^{\alpha \beta} = \int_{\mathcal{M}_L} \dd[4]{z} W^{\alpha \beta}(z), \! \quad c^{\mu \alpha \beta} = \int_{\mathcal{M}_L} \dd[4]{z} (z_\torus)^\mu W^{\alpha \beta}(z),
\end{split} \\
\begin{split}
d^{\mu \nu \alpha \beta} = \frac{1}{2} \int_{\mathcal{M}_L} \dd[4]{z} (z_\torus)^\mu (z_\torus)^\nu \bar{W}^{\alpha \beta}(z),
\end{split}
\end{gather*}
where a sum over repeated indices is understood the scale labels are temporarily suppressed, $W^{\alpha \beta}(z)$ stands for $W^{\alpha \beta}(z, 0)$ and the kernel $\bar{W}^{\alpha \beta}(z)$ must be evaluated at $m_N = 0$ when the spinor indices $\alpha, \beta$ have different chiralities. If the coefficients $b^{\alpha \beta}, c^{\mu \alpha \beta}, d^{\mu \nu \alpha \beta}$ are thought as $4 \times 4$ complex matrices whose entries are labelled by $\alpha, \beta$, we can decompose them along the standard basis $\lbrace \id, \gamma^\mu, \gamma^\mu \gamma^5, [\gamma^\mu, \gamma^\nu]/2 \rbrace$:
\begin{align}
\label{eq:gamma_matrices_base_decomposition}
\begin{split}
b & = b_\id \id + b_5 \gamma^5 + (b_v)_\rho \gamma^\rho + (b_a)_\rho \gamma^\rho \gamma^5 + (b_t)_{\theta \rho} [\gamma^\theta, \gamma^\rho]/2, \\
c^\mu & = (c_\id)^\mu \id + (c_5)^\mu \gamma^5 + (c_v)_\rho^\mu \gamma^\rho + (c_a)_\rho^\mu \gamma^\rho \gamma^5 + (c_t)_{\theta \rho}^\mu [\gamma^\theta, \gamma^\rho]/2, \\
d^{\mu \nu} & = (d_\id)^{\mu \nu} \id + (d_5)^{\mu \nu} \gamma^5 + (d_v)_\rho^{\mu \nu} \gamma^\rho + (d_a)_\rho^{\mu \nu} \gamma^\rho \gamma^5 + (d_t)_{\theta \rho}^{\mu \nu} [\gamma^\theta, \gamma^\rho]/2.
\end{split}
\end{align}
Even if the spacetime volume is finite, the theory is invariant under the action of a discrete group $G \subseteq \mathrm{SO}(4)$ generated by all the $\pi/2$ rotations around each pair of axes. The precise meaning of this sentence is the following: if we consider the maps
\begin{align*}
S(\vec{\theta}, \vec{\xi}) 
& \equiv 
\begin{pmatrix}
\displaystyle \exp \left( \frac{i \vec{\sigma} \cdot (\vec{\theta} + \vec{\xi})}{2} \right)	&	0 \\[3pt]
\displaystyle
0									& \displaystyle \exp \left( \frac{i \vec{\sigma} \cdot (\vec{\theta} - \vec{\xi})}{2} \right)
\end{pmatrix}, \\
U(\vec{\theta}, \vec{\xi}) 
& \equiv \exp[\vec{\theta} \cdot \vec{L} + \vec{\xi} \cdot \vec{K}],
\end{align*}
where $\vec{\sigma}$ are the Pauli matrices and the $4 \times 4$ matrices $\vec{L}, \vec{K}$ are defined as
\begin{align*}
L_1 = 
\begin{pmatrix}
0	&		&		&		\\	
	&	0	&	0	&	0	\\	
	&	0	&	0	&	-1	\\
	&	0	&	1	&	0	\\		
\end{pmatrix}, \qquad
L_2 = 
\begin{pmatrix}
0	&		&		&		\\	
	&	0	&	0	&	1	\\	
	&	0	&	0	&	0	\\
	&	-1	&	0	&	0	\\		
\end{pmatrix}, \qquad
L_3 = 
\begin{pmatrix}
0	&		&		&		\\	
	&	0	&	-1	&	0	\\	
	&	1	&	0	&	0	\\
	&	0	&	0	&	0	\\		
\end{pmatrix}, \\[4pt]
K_1 = 
\begin{pmatrix}
0	&	-1	&		&		\\	
1	&	0	&	0	&	0	\\
	&	0	&	0	&	0	\\		
	&	0	&	0	&	0	\\	
\end{pmatrix}, \qquad
K_2 = 
\begin{pmatrix}
0	&		&	1	&		\\	
	&	0	&	0	&	0	\\	
-1	&	0	&	0	&	0	\\		
	&	0	&	0	&	0	\\	
\end{pmatrix}, \qquad
K_3 = 
\begin{pmatrix}
0	&		&		&	-1	\\	
	&	0	&	0	&	0	\\	
	&	0	&	0	&	0	\\		
1	&	0	&	0	&	0	\\	
\end{pmatrix},
\end{align*}
the generating functional~\eqref{eq:generating_functional_fermionic} is invariant under the transformations
\begin{gather*}
\hat{\psi}^-_k \mapsto S(\vec{\theta}, \vec{\xi}) \, \hat{\psi}^-_{U(\vec{\theta}, \vec{\xi})k}, \qquad \hat{\psi}^+_k \mapsto \hat{\psi}^+_{U(\vec{\theta}, \vec{\xi})k} \, S^{-1}(\vec{\theta}, \vec{\xi}), \\
\hat{\eta}^-_k \mapsto S(\vec{\theta}, \vec{\xi}) \, \hat{\eta}^-_{U(\vec{\theta}, \vec{\xi})k}, \qquad \hat{\eta}^+_k \mapsto \hat{\eta}^+_{U(\vec{\theta}, \vec{\xi})k} \, S^{-1}(\vec{\theta}, \vec{\xi}), \\
\begin{split}
\hat{J}_k \mapsto U^{-1}(\vec{\theta}, \vec{\xi}) \hat{J}_{U(\vec{\theta}, \vec{\xi})k}
\end{split}
\end{gather*}
\emph{only} provided that $\vec{\theta}, \vec{\xi} \in (\pi \Z/2)^3$.  This constraint on $\vec{\theta}, \vec{\xi}$ is due to the finiteness of the spacetime volume: if $L < +\infty$, the generating functional is left invariant by those $\mathrm{SO}(4)$ transformations that preserve the shape of the cubic spacetime $\mathcal{M}_L = [-L/2, L/2]^4$ and the structure of the discrete Fourier space $\mathcal{D}_{N, L} = (2\pi \Z/L)^4 \cap \supp \chi_N$.
As $\vec{\theta}, \vec{\xi}$ span $(\pi \Z/2)^3$, the matrices $U(\vec{\theta}, \vec{\xi})$ constitute a four-dimensional, complex representation of the finite group $G$ denoted by $\mathcal{U}$. This representation is also \emph{irreducible}: in fact, the whole canonical basis of $\C^4$ is recovered by acting on the vector $e_0 \equiv (1, 0, 0, 0)$ with the matrices $\id_4, \exp(\pi K_1/2), \exp(-\pi K_2/2), \exp(\pi K_3/2) \in \mathcal{U}(G)$, so the only invariant subspaces of $\mathcal{U}$ are $\lbrace 0 \rbrace \subset \C^4$ and $\C^4$ itself.

Depending on the number of their spacetime indices, the coefficients $(b_v)_\rho, (b_a)_\rho$, $(b_t)_{\theta \rho}, (c_\id)^\mu, (c_5)^\mu, (c_v)_\rho^\mu, \dots$ can be thought as tensors belonging to the representation spaces of $\mathcal{U}, \mathcal{U}^{\otimes 2}, \mathcal{U}^{\otimes 3}, \dots$. The $G$-symmetry of the theory implies that each of these tensors must belong to a one-dimensional invariant subspace of the corresponding representation. To see a concrete example of this, consider the term
\begin{equation}
\int \dd[4]{x} [\psi^+ \, (b_v)_\rho \gamma^\rho \, \psi^-](x).
\end{equation} 
If this term is invariant under the above transformations, it must be
\begin{equation}
S^{-1}(\vec{\theta}, \vec{\xi}) \, (b_v)_\rho \gamma^\rho \, S(\vec{\theta}, \vec{\xi}) = (b_v)_\rho \gamma^\rho
\end{equation}
for every $\vec{\theta}, \vec{\xi} \in (\pi\Z/2)^3$. Thanks to the identity $S^{-1}(\vec{\theta}, \vec{\xi}) \, \gamma^\rho \, S(\vec{\theta}, \vec{\xi}) = U(\vec{\theta}, \vec{\xi})^{\rho}_{\,\, \mu} \, \gamma^\mu$, this can be rewritten as
\begin{equation}
(b_v)_\rho = U(\vec{\theta}, \vec{\xi})^{\rho}_{\,\, \mu} \, (b_v)^\mu \qquad \forall \vec{\theta}, \vec{\xi} \in (\pi\Z/2)^3,
\end{equation}
showing that the vector $((b_v)_\rho)_{\rho = 0, \dots, 3} \in \C^4$ belongs to a one-dimensional invariant subspace of the representation $\mathcal{U}$. Similar considerations apply to the other coefficients as well.

Now consider the element of $G$ represented by $\exp\,\![(\pi, 0, 0) \cdot (\vec{L} + \vec{K})] = - \id_4$. The invariance of $(b_v)_\rho$, $(b_a)_\rho$, $(c_\id)^\mu$, $(c_5)^\mu$, $(c_t)_{\theta \rho}^\mu$, $(d_v)_\rho^{\mu \nu}$, $(d_a)_\rho^{\mu \nu}$ under the action of this element implies that $(b_v)_\rho = - (b_v)_\rho, (b_a)_\rho = - (b_a)_\rho$ and so on; therefore, all the coefficients carrying an \emph{odd} number of spacetime indices are identically vanishing.  The fact that $\mathcal{U}$ is irreducible implies that all the coefficients with two spacetime indices are proportional to the Kronecker delta. Indeed, let us consider $(b_t)_{\theta \rho}$ for concreteness. The invariance argument used above yields
\begin{equation}
\label{eq:rank_2_tensor_invariance}
\mathcal{U}_\mu^{\,\, \theta}(g) \, \mathcal{U}_\nu^{\,\, \rho}(g) (b_t)_{\theta \rho} = (b_t)_{\mu \nu} \qquad \forall g \in G,
\end{equation}
or, in a more compact matrix form,
\begin{equation}
\label{eq:rank_2_tensor_invariance_matrix}
\mathcal{U}(g) \cdot b_t \cdot \mathcal{U}(g)^T = b_t \qquad \forall g \in G.
\end{equation}
Recalling that $\mathcal{U}(g) \mathcal{U}(g)^T = \id_4$ ($\mathcal{U}$ is orthogonal, because $\vec{L}, \vec{K}$ are skew-symmetric), we can multiply both sides of~\eqref{eq:rank_2_tensor_invariance_matrix} by $\mathcal{U}(g)$ on the right, thus getting $[b_t, \mathcal{U}(g)] = 0 \, \forall g \in G$. Then, since $\mathcal{U}$ is irreducible, Schur's lemma guarantees that $(b_t)_{\theta \rho} = b_t \, \delta_{\theta \rho}$ for some $b_t \in \mathbb{C}$. The same argument holds for $(c_v)_\rho^\mu, (c_a)_\rho^\mu, (d_\id)^{\mu \nu}, (d_5)^{\mu \nu}$ as well. 

Based on the above considerations, the decomposition~\eqref{eq:gamma_matrices_base_decomposition} becomes
\begin{equation}
\label{eq:lm:coefficients_matrix_form_final}
b = b_\id \id + b_5 \gamma^5, \qquad c^\mu = (c_v + c_a \gamma^5) \gamma^\mu, \qquad d^{\mu \nu} = (d_\id \id + d_5 \gamma^5) \delta^{\mu \nu} + (d_t)^{\mu \nu}_{\theta \rho} [\gamma^\theta, \gamma^\rho]/2.
\end{equation}
Since
\begin{equation}
b_\id \id + b_5 =
\begin{pmatrix}
b + b_5	&	0	\\
0		&	b - b_5
\end{pmatrix},
\qquad
(c_v + c_a \gamma^5) \gamma^\mu =
\begin{pmatrix}
0								&	(c_v + c_a) \sigma^\mu_+ \\
(c_v - c_a) \sigma^\mu_-		&	0
\end{pmatrix},
\end{equation}
the matrix structure of $b, c^\mu$ clearly agrees with~\eqref{eq:localization_request}. According to~\eqref{eq:lm:coefficients_matrix_form_final}, we have
\begin{equation}
\label{eq:d_structure}
d^{\mu \nu} =
\begin{pmatrix}
\delta^{\mu \nu}(d_\id + d_5) + (d_t)^{\mu \nu}_{\theta \rho} \, \sigma_+^{[\theta} \sigma_-^{\rho]}/2	&	0 \\
0	&	\delta^{\mu \nu}(d_\id - d_5) + (d_t)^{\mu \nu}_{\theta \rho} \, \sigma_-^{[\theta} \sigma_+^{\rho]}/2
\end{pmatrix},
\end{equation} 
where the notation $\sigma_\pm^{[\theta} \sigma_\mp^{\rho]}$ stands for the antisymmetrization $\sigma_\pm^{\theta} \sigma_\mp^{\rho} -\sigma_\pm^{\rho} \sigma_\mp^{\theta}$. Formula~\eqref{eq:d_structure} shows that the only potentially nonvanishing entries of the matrix $d^{\mu \nu}$ are labelled by spinor indices $\alpha, \beta$ with opposite chiralities, so these entries must be evaluated at $m_N = 0$. The same argument used to prove that $\beta^{m, s}_{h-1} \vert_{m_N = 0} = 0$ then implies that $d^{\mu \nu} = 0$, so we conclude that the action of $\loc$ on $W^{\alpha \beta}$ agrees with~\eqref{eq:localization_request}. An entirely similar procedure can be followed with the remaining lines of~\eqref{eq:localization_definition}.

\section{Some estimates at higher orders}
\label{app:higher_orders_estimates}

Although the bound provided by Theorem~\ref{th:convergent_expansions} is based on the tree expansion, it clearly holds for each renormalized multiscale Feynman diagram that contributes to $\hat{S}(k)$ or $\hat{\Gamma}^\mu(p', p)$, at least when $\abs{k}, \abs*{p'}, \abs{p}$ are small enough. To concretely appreciate the subtle mechanism that allows to extract the overall $m^2/M^2$ factor, it is instructive to analyze a few concrete examples and work out the bounds~\eqref{eq:bounds_on_S_and_Gamma} ``by hand''. In the standard perturbative QFT framework, it is practically impossible to give an \emph{a priori} estimate on the size of a Feynman diagram; yet, the non perturbative formalism adopted in this paper allows to exhibit a meaningful bound on any multiscale renormalized Feynman diagram, without even writing down its explicit expression.

\paragraph{A fourth order diagram.} Let us consider the diagram displayed in~\cite[Figure 1, row IV]{Karplus_Kroll_1950}, which contributes to the fourth order correction to the anomalous gyromagnetic factor in ordinary perturbative QED$_4$. In our model, there are plenty of renormalized multiscale Feynman diagrams with the same structure as this. An example is given by figure~\ref{fig:fourth_order_diagram}. Here, black boxes have been used to evidence the clusters associated with the vertices of the tree: proceeding from the innermost to the outermost cluster, we encounter one propagator on scale $h_1$, two propagators on scale $h_2$, one propagator on scale $h_3$ and two propagators on scale $h^\star$.
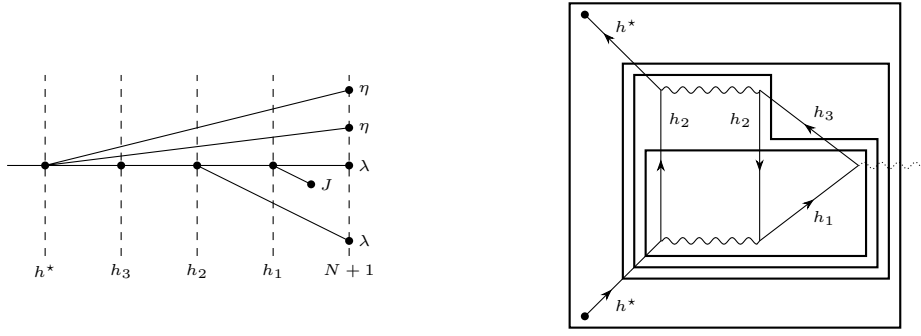
\begin{figure}[t]
\begin{center}
\begin{tikzpicture}[baseline]
\draw	(-0.5, 0) -- (4, 0)
		(2, 0) -- (4, -1)
		(3, 0) -- (3.5, -0.25)
		(0, 0) -- (4, 0.5)
		(0, 0) -- (4, 1);
\fill[black]		(0, 0)			circle (1.5pt)
				(1, 0)			circle (1.5pt)
				(2, 0)			circle (1.5pt)
				(3, 0)			circle (1.5pt)
				(3.5, -0.25)		circle (1.5pt)
				(4, 0)			circle (1.5pt)
				(4, -1)			circle (1.5pt)
				(4, 0.5)			circle (1.5pt)
				(4, 1)			circle (1.5pt);
\draw[dashed]	(0, -1.2) -- (0, 1.2)
				(1, -1.2) -- (1, 1.2)
				(2, -1.2) -- (2, 1.2)
				(3, -1.2) -- (3, 1.2)
				(4, -1.2) -- (4, 1.2);
\node	(hs)		at	(0, -1.4)		{\tiny $h^\star$};
\node	(h1)		at	(1, -1.4)		{\tiny $h_3$};
\node	(h2)		at	(2, -1.4)		{\tiny $h_2$};
\node	(h3)		at	(3, -1.4)		{\tiny $h_1$};
\node	(ep)		at	(4, -1.4)		{\tiny $N+1$};
\node	(l1)		at	(4.2, 0)			{\tiny $\lambda$};
\node	(l2)		at	(4.2, -1)		{\tiny $\lambda$};
\node	(J)		at	(3.7, -0.25)		{\tiny $J$};
\node	(n1)		at	(4.2, 0.5)		{\tiny $\eta$};
\node	(n2)		at	(4.2, 1)			{\tiny $\eta$};
\end{tikzpicture}
\qquad \qquad \qquad
\begin{tikzpicture}[baseline]
\begin{feynman}
\draw[with arrow]	(0.5, 1) -- (0.5, -1);
\draw[with arrow]	(0.5, -1) -- (1.8, 0);
\draw[with arrow]	(1.8, 0) -- (0.5, 1);
\draw[photon]		(-0.8, 1) -- (0.5, 1)
					(-0.8, -1) -- (0.5, -1);
\draw[with arrow]	(-0.8, -1) -- (-0.8, 1);
\draw[with arrow]	(-1.8, -2) -- (-1.2, -1.4);
\draw				(-1.2, -1.4) -- (-0.8, -1);
\draw				(-0.8, 1) -- (-1.2, 1.4);
\draw[with arrow]	(-1.2, 1.4) -- (-1.8, 2);
\fill[black]			(-1.8, -2)		circle	(1.5pt)
					(-1.8, 2)		circle	(1.5pt);
\draw[photon, densely dotted]		(1.8, 0) -- (2.7, 0);
\end{feynman}
\draw[black, thick]	(-1, -1.2)		rectangle	(1.9, 0.2);
\draw[black, thick]	(-1.3, -1.5) 	rectangle	(2.2, 1.35);
\draw[black, thick]	(-1.15, -1.35) -- (-1.15, 1.2) -- 
					(0.65, 1.2) -- (0.65, 0.35) --
					(2.05, 0.35) -- (2.05, -1.35) -- cycle;
\draw[black, thick]	(-2, -2.15)	rectangle	(2.35, 2.15);
\node	(h1)		at	(1.35, -0.7)		{\tiny $h_1$};
\node	(h2)		at	(-0.55, 0.6)		{\tiny $h_2$};
\node	(h21)	at	(0.25, 0.6)		{\tiny $h_2$};
\node	(h3)		at	(1.35, 0.7)		{\tiny $h_3$};
\node	(hs1)	at	(-1.25, -1.85)	{\tiny $h^\star$};
\node	(hs2)	at	(-1.25, 1.85)	{\tiny $h^\star$};
\end{tikzpicture}
\end{center}
\caption{Fourth order diagram with the same structure as~\cite[Figure 1, row IV]{Karplus_Kroll_1950}, together with corresponding tree.}
\label{fig:fourth_order_diagram}
\end{figure}
According to definition~\eqref{eq:localization_definition}, the $\ren$ operator acts nontrivially on the penultimate cluster. The $(h^\star-1)$-weighted norm of the diagram can be estimated by using the same methods discussed in the proof of Theorem~\ref{th:renormalization_bound}; namely, after constructing a spanning tree $T_h$ for each cluster with scale $h$, we bound the propagators lying along $T_h$ with their weighted $L^1$ norm and the other ones with their $L^\infty$ norm. Moreover, the action of $\ren$ on the cluster with scale $h_3$ results in an extra $2^{(h^\star - h_3) \cdot 2}$ factor, so we get
\begin{equation}
\norm*{W}_{h^\star-1} \le \underbrace{\left(\frac{C \lambda^2}{M^2}\right)^2}_{\text{Endpoints}} \cdot \,\, \underbrace{2^{-h^\star - h^\star}}_{h^\star \text{ cluster}} \,\, \cdot \underbrace{2^{3h_3}}_{h_3 \text{ cluster}} \cdot \,\, \underbrace{2^{3h_2 - h_2}}_{h_2 \text{ cluster}} \,\, \cdot \underbrace{2^{-h_1}}_{h_1 \text{ cluster}} \cdot \,\, \underbrace{2^{2(h^\star - h_3)}}_{\ren}
\end{equation} 
and this can be easily rewritten as
\begin{equation}
\label{eq:diagram_estimate}
\norm*{W}_{h^\star-1} \le \left(\frac{C \lambda^2 \cutoff^2}{M^2}\right)^2 \cdot 2^{2(h_1 - N) + 2(h_2 - N) -2h^\star} \cdot 2^{(h_3 - h^\star) \cdot (0 - 2) + (h_2 - h_3) \cdot (-3 - 0) + (h_1 - h_2) \cdot (-3 - 0)},  
\end{equation}
thus reproducing the structure displayed in~\eqref{eq:renormalization_bound} (for each cluster, we evidenced the difference $D_v - R_v$, where $v \in \vrt(\tau)$ is the corresponding vertex on the tree). The short memory property is recovered by multiplying~\eqref{eq:diagram_estimate} by $1 = 2^{2(h^\star - N)} \cdot 2^{2(N - h^\star)}$ and redistributing the factor $2^{2(N - h^\star)}$ among the various clusters as $2^{2(N - h^\star)} = 2^{2(N - h_1) + 2(h_1 - h_2) + 2(h_2 - h_3) + 2(h_3 - h^\star)}$; the result is
\begin{equation}
\norm*{W}_{h^\star-1} \le \left(\frac{C \lambda^2 \cutoff^2}{M^2}\right)^2 \cdot 2^{2(h^\star - N)} \cdot 2^{2(h_2 - N) -2h^\star} \cdot 2^{(h_3 - h^\star) \cdot 0 + (h_2 - h_3) \cdot (-1) + (h_1 - h_2) \cdot (-1)}.
\end{equation}
Finally, recalling that $2^{h^\star} \simeq m$, we use the bound $2^{2(h_2 - N)} \le 1$ and rearrange the remaining factors as
\begin{equation}
\norm*{W}_{h^\star-1} \le \frac{C^2 \lambda^4 \cutoff^2}{M^2} \cdot \frac{m^2}{M^2} \cdot 2^{-2h^\star} \cdot 2^{(h_3 - h^\star) \cdot 0 + (h_2 - h_3) \cdot (-1) + (h_1 - h_2) \cdot (-1)}.
\end{equation}
As expected, an overall $m^2/M^2$ factor appears. Also, the sum over the scale differences $h_2 - h_3, h_1 - h_2$ can be safely performed due to the damping factors $2^{-(h_2 - h_3)}, 2^{-(h_1 - h_2)}$, whereas the undamped sum over $h_3 - h^\star$ is bounded by $\abs*{N - h^\star + 1} \simeq \log(\cutoff/m) \le \log^4(\cutoff/m)$. Since the amputation procedure cancels the $2^{-2h^\star}$ factor, we proved that our (amputated) fourth order diagram satisfies the bound~\eqref{eq:bounds_on_S_and_Gamma}. It is worth noting that this estimate is not optimal, because it neglects the gain produced by $2^{2(h_2 - N)}$ and it overcounts the number of $\log(\cutoff/m)$ factors. Nevertheless, this analysis makes it evident that some diagrams (or, more generally, some trees) are more relevant than others. In principle, the precise number of $\log(\cutoff/m)$ factors and the overall damping produced by the $\lambda$ endpoints can be worked out case by case.

\paragraph{A fourth order diagram with two nested renormalizations.} As a second example, we consider the diagram depicted in Figure~\ref{fig:fourth_order_diagram_nested}, which has the same structure as the one displayed in~\cite[Figure 1, row II]{Karplus_Kroll_1950}.
\begin{figure}[t]
\begin{center}
\begin{tikzpicture}[baseline]
\draw	(-0.5, 0) -- (4, 0)
		(1.5, 0) -- (4, -1)
		(3, 0) -- (3.5, -0.25)
		(0, 0) -- (4, 0.5)
		(0, 0) -- (4, 1);
\fill[black]		(0, 0)			circle (1.5pt)
				(1.5, 0)			circle (1.5pt)
				(3, 0)			circle (1.5pt)
				(3.5, -0.25)		circle (1.5pt)
				(4, 0)			circle (1.5pt)
				(4, -1)			circle (1.5pt)
				(4, 0.5)			circle (1.5pt)
				(4, 1)			circle (1.5pt);
\draw[dashed]	(0, -1.2) -- (0, 1.2)
				(1.5, -1.2) -- (1.5, 1.2)
				(3, -1.2) -- (3, 1.2)
				(4, -1.2) -- (4, 1.2);
\node	(hs)		at	(0, -1.4)		{\tiny $h^\star$};
\node	(h2)		at	(1.5, -1.4)		{\tiny $h_2$};
\node	(h1)		at	(3, -1.4)		{\tiny $h_1$};
\node	(ep)		at	(4, -1.4)		{\tiny $N+1$};
\node	(l1)		at	(4.2, 0)			{\tiny $\lambda$};
\node	(l2)		at	(4.2, -1)		{\tiny $\lambda$};
\node	(J)		at	(3.7, -0.25)		{\tiny $J$};
\node	(n1)		at	(4.2, 0.5)		{\tiny $\eta$};
\node	(n2)		at	(4.2, 1)			{\tiny $\eta$};
\end{tikzpicture}
\qquad \qquad \qquad \quad
\begin{tikzpicture}[baseline]
\begin{feynman}
\draw[with arrow]	(0, 0) -- (-1, 1);
\draw[with arrow]	(-1, 1) -- (-2, 2);
\draw[with arrow]	(-2, 2) -- (-2.6, 2.6);
\draw[with arrow]	(-2.6, -2.6) -- (-2, -2);
\draw[with arrow]	(-2, -2) -- (-1, -1);
\draw[with arrow]	(-1, -1) -- (0, 0);
\draw[photon]		(-1, -1) -- (-1, 1)
					(-2, -2) -- (-2, 2);
\draw[photon, densely dotted]		(0, 0) -- (0.8, 0);
\end{feynman}
\draw[black, thick]	(-1.15, -1.2) rectangle (0.2, 1.2);
\draw[black, thick]	(-2.15, -2.2) rectangle (0.35, 2.2);
\draw[black, thick]	(-2.75, -2.75) rectangle (0.5, 2.75);
\fill[black]		(-2.6, -2.6) circle (1.5pt)
				(-2.6, 2.6) circle (1.5pt);
\node	(hs1)	at	(-2.15, -2.55)		{\tiny $h^\star$};
\node	(hs2)	at	(-2.15, 2.55)		{\tiny $h^\star$};
\node	(h21)	at	(-1.3, -1.7)		{\tiny $h_2$};
\node	(h22)	at	(-1.3, 1.7)			{\tiny $h_2$};
\node	(h11)	at	(-0.3, -0.7)		{\tiny $h_1$};
\node	(h12)	at	(-0.3, 0.7)			{\tiny $h_1$};
\end{tikzpicture}
\end{center}
\caption{A fourth order diagram with two nested renormalization operators acting nontrivially.}
\label{fig:fourth_order_diagram_nested}
\end{figure}
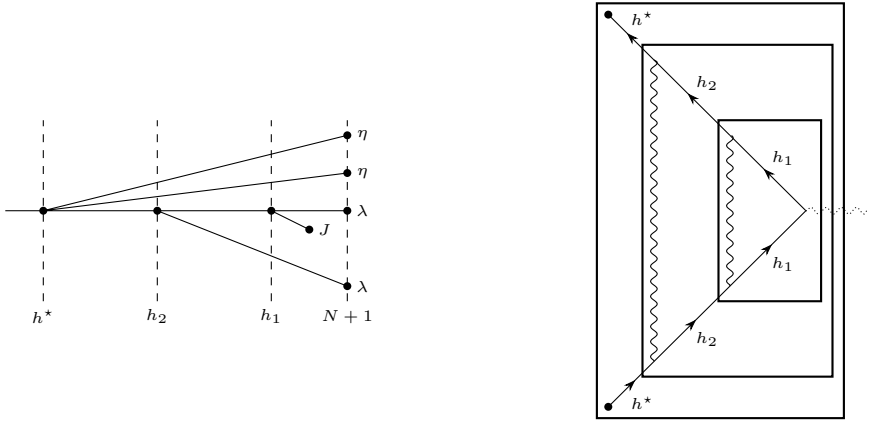
Now the renormalization operator acts nontrivially on two clusters, respectively lying on scale $h_1$ and on scale $h_2$. Proceeding as before, we obtain
\begin{equation}
\norm*{W}_{h^\star-1} \le \underbrace{\left(\frac{C \lambda^2}{M^2}\right)^2}_{\text{Endpoints}} \cdot \,\, \underbrace{2^{-h^\star - h^\star}}_{h^\star \text{ cluster}} \,\, \cdot \,\, \underbrace{2^{3h_2 - h_2}}_{h_2 \text{ cluster}} \,\, \cdot \,\, \underbrace{2^{3h_1 - h_1}}_{h_1 \text{ cluster}} \,\, \cdot \,\, \underbrace{2^{2(h^\star - h_2) + 2(h_2 - h_1)}}_{\ren}.
\end{equation}
Again, this can be rearranged as
\begin{equation}
\label{eq:second_diagram_clusters_bound}
\norm*{W}_{h^\star-1} \le \left(\frac{C \lambda^2 \cutoff^2}{M^2}\right)^2 \cdot 2^{2(h_1 - N) + 2(h_2 - N) -2h^\star} \cdot 2^{(h_2 - h^\star) \cdot (0 - 2) + (h_1 - h_2) \cdot (0 - 2)}
\end{equation}
and the short memory property is evidenced by inserting a $1 = 2^{2(h^\star - N)} \cdot 2^{2(N - h^\star)}$ factor and subsequently writing $N - h^\star$ as $(N - h_1) + (h_1 - h_2) + (h_2 - h^\star)$. Since $2^{2(h_2 - N)} \le 1, 2^{h^\star} \simeq m$, the short memory property transforms~\eqref{eq:second_diagram_clusters_bound} into
\begin{equation}
\norm*{W}_{h^\star-1} \le \frac{C \lambda^4 \cutoff^2}{M^2} \frac{m^2}{M^2} \cdot 2^{-2h^\star} \cdot 2^{(h_2 - h^\star) \cdot 0 + (h_1 - h_2) \cdot 0}.
\end{equation}
The sum over the scale differences is not damped, so it produces a $\abs*{N - h^\star + 1}^2 \simeq \log^2(\cutoff/m) \le \log^4(\cutoff/m)$ factor: once the amputation is taken into account, we conclude that this diagram satisfies~\eqref{eq:bounds_on_S_and_Gamma}.

\paragraph{A tenth order diagram.} We now consider the diagram represented in Figure~\ref{fig:tenth_order_diagram}, that contributes to $\rgyr$ at tenth order in $\lambda$. It should be noted that this diagram lies at the current \emph{limit} of the ordinary perturbative calculations in Minkowskian QED$_4$: it is extremely demanding (and perhaps impossible) to express its contribution to $\rgyr$ in a closed form~\cite{Aoyama_2012}.
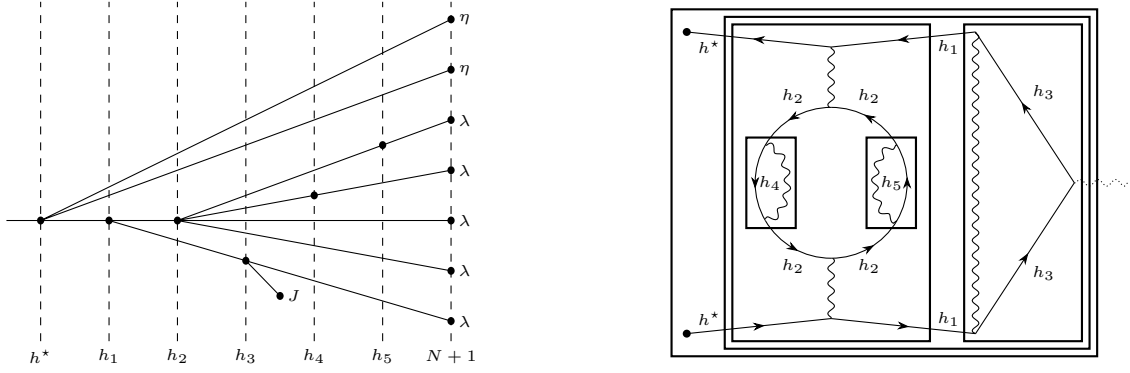
\begin{figure}
\begin{center}
\begin{tikzpicture}[baseline, xscale=0.9]
\draw	(-0.5, 0) -- (6, 0)
		(1, 0) -- (6, -1.333)
		(2, 0) -- (6, -0.667)
		(2, 0) -- (6, 0.667)
		(2, 0) -- (6, 1.333)
		(0, 0) -- (6, 2)
		(0, 0) -- (6, 2.667)
		(3, -0.533) -- (3.5, -1);
\fill[black]		(0, 0)			circle (1.5pt)
				(1, 0)			circle (1.5pt)
				(2, 0)			circle (1.5pt)
				(3, -0.533)		circle (1.5pt)
				(4, 0.333)		circle (1.5pt)
				(5, 1)			circle (1.5pt)
				(6, 0)			circle (1.5pt)
				(6, -1.333)		circle (1.5pt)
				(6, -0.667)		circle (1.5pt)
				(6, 0.667)		circle (1.5pt)
				(6, 1.333)		circle (1.5pt)
				(6, 2)			circle (1.5pt)
				(6, 2.667)		circle (1.5pt)
				(3.5, -1)		circle (1.5pt);
\draw[dashed]	(0, -1.6) -- (0, 2.9)
				(1, -1.6) -- (1, 2.9)
				(2, -1.6) -- (2, 2.9)
				(3, -1.6) -- (3, 2.9)
				(4, -1.6) -- (4, 2.9)
				(5, -1.6) -- (5, 2.9)
				(6, -1.6) -- (6, 2.9);
\node	(hs)		at	(0, -1.8)		{\tiny $h^\star$};
\node	(h1)		at	(1, -1.8)		{\tiny $h_1$};
\node	(h2)		at	(2, -1.8)		{\tiny $h_2$};
\node	(h3)		at	(3, -1.8)		{\tiny $h_3$};
\node	(h4)		at	(4, -1.8)		{\tiny $h_4$};
\node	(h5)		at	(5, -1.8)		{\tiny $h_5$};
\node	(h6)		at	(6, -1.8)		{\tiny $N+1$};
\node	(l0)		at	(6.2, 0.667)		{\tiny $\lambda$};
\node	(l1)		at	(6.2, 0)			{\tiny $\lambda$};
\node	(l2)		at	(6.2, -0.667)	{\tiny $\lambda$};
\node	(l3)		at	(6.2, 1.333)		{\tiny $\lambda$};
\node	(l4)		at	(6.2, -1.333)	{\tiny $\lambda$};
\node	(l5)		at	(6.2, 2)			{\tiny $\eta$};
\node	(l6)		at	(6.2, 2.667)		{\tiny $\eta$};
\node	(J)		at	(3.7, -1)		{\tiny $J$};
\end{tikzpicture}
\qquad \qquad \qquad
\begin{tikzpicture}[baseline=-0.5cm]
\begin{feynman}
\draw[with arrow]	(0.3, -2) -- (1.6, 0);
\draw[with arrow]	(1.6, 0) -- (0.3, 2);
\draw[with arrow]	(0.3, 2) -- (-1.6, 1.8);
\draw[with arrow]	(-1.6, -1.8) -- (0.3, -2);
\draw[with arrow]	(-1.6, 1.8) -- (-3.5, 2);
\draw[with arrow]	(-3.5, -2) -- (-1.6, -1.8);
\coordinate	(A)	at	($(-1.6, 0) +(150:1)$);
\coordinate	(B)	at	($(-1.6, 0) +(210:1)$);
\coordinate	(C)	at	($(-1.6, 0) +(30:1)$);
\coordinate	(D)	at	($(-1.6, 0) +(-30:1)$);
\coordinate	(d1)	at	($(-1.6, 0) +(180:1)$);
\coordinate	(d2)	at	($(-1.6, 0) +(0:1)$);
\draw[photon]		(-1.6, 1.8) -- (-1.6, 1)
					(-1.6, -1.8) -- (-1.6, -1)
					(0.3, 2) -- (0.3, -2);
\draw[photon]		(A) to[out=-5, in=5] (B)
					(C) to[out=185, in=175] (D);
\draw[photon, densely dotted]		(1.6, 0) -- (2.3, 0);
\draw[with arrow]	(-1.6, 0) ++(90:1) arc (90:150:1);
\draw[with arrow]	(-1.6, 0) ++(150:1) arc (150:210:1);
\draw[with arrow]	(-1.6, 0) ++(210:1) arc (210:270:1);
\draw[with arrow]	(-1.6, 0) ++(270:1) arc (270:330:1);
\draw[with arrow]	(-1.6, 0) ++(-30:1) arc (-30:30:1);
\draw[with arrow]	(-1.6, 0) ++(30:1) arc (30:90:1);
\draw[black, thick]	($(A) +(-0.25, 0.1)$) rectangle ($(B) +(0.4, -0.1)$);
\draw[black, thick]	($(C) +(-0.4, 0.1)$) rectangle ($(D) +(0.25, -0.1)$);
\draw[black, thick]	($(d1) +(-0.3, 2.1)$) rectangle ($(d2) +(0.3, -2.1)$);
\draw[black, thick]  (0.15, -2.1) rectangle (1.7, 2.1);
\draw[black, thick]	($(d1) +(-0.4, 2.2)$) rectangle (1.8, -2.2);
\draw[black, thick]	(-3.7, -2.3) rectangle (1.9, 2.3);
\fill[black]			(-3.5, -2)	circle	(1.5pt)
					(-3.5, 2)	circle	(1.5pt);
\end{feynman}
\node	(l1)		at	(1.2, -1.2)		{\tiny $h_3$};
\node	(l2)		at	(1.2, 1.2)		{\tiny $h_3$};
\node	(l3)		at	(-3.2, -1.8)		{\tiny $h^\star$};
\node	(l4)		at	(-3.2, 1.8)		{\tiny $h^\star$};
\node	(l5)		at	(-0.05, -1.8)	{\tiny $h_1$};
\node	(l6)		at	(-0.05, 1.8)		{\tiny $h_1$};
\node	(l7)		at	(-1.1, -1.1)		{\tiny $h_2$};
\node	(l8)		at	(-1.1, 1.15)		{\tiny $h_2$};
\node	(l9)		at	(-2.1, -1.1)		{\tiny $h_2$};
\node	(l10)	at	(-2.1, 1.15)		{\tiny $h_2$};
\node	(l11)	at	($(d1) +(0.2, 0)$)	{\tiny $h_4$};
\node	(l12)	at	($(d2) +(-0.2, 0)$)	{\tiny $h_5$};
\end{tikzpicture}
\end{center}
\caption{A tenth order diagram with several nested renormalizations.}
\label{fig:tenth_order_diagram}
\end{figure}

The same approach followed in the previous cases yields
\begin{multline}
\norm*{W}_{h^\star} \le \underbrace{\left(\frac{C \lambda^2}{M^2}\right)^5}_{\text{Endpoints}} \cdot \,\, \underbrace{2^{-h^\star - h^\star}}_{h^\star \text{ cluster}} \,\, \cdot \underbrace{2^{3h_1 - h_1}}_{h_1 \text{ cluster}} \cdot \,\, \underbrace{2^{3h_2 + 3 \cdot (-h_2)}}_{h_2 \text{ cluster}} \,\, \cdot \,\, \underbrace{2^{3h_3 -h_3}}_{h_3 \text{ cluster}} \,\, \cdot \underbrace{2^{3h_4}}_{h_4 \text{ cluster}} \, \times \\
\times \underbrace{2^{3h_5}}_{h_5 \text{ cluster}} \,\, \cdot \,\, \underbrace{2^{2(h^\star - h_1) + 2(h_1 - h_3) + 3(h_2 - h_4) + 3(h_2 - h_5)}}_{\ren}
\end{multline}
which is equivalent to
\begin{multline}
\norm*{W}_{h^\star} \le \left(\frac{C \lambda^2 \cutoff^2}{M^2}\right)^5 \,\, \cdot \,\, 2^{2(h_2 - N) + 2(h_2 - N) + 2(h_3 - N) + 2(h_4 - N) + 2(h_5 - N) - 2h^\star} \, \times \\
\times 2^{(h_1 - h^\star)(0 - 2) + (h_2 - h_1) (-2 - 0) + (h_3 - h_1)(0 - 2) + (h_4 - h_2)(1 - 3) + (h_5 - h_2)(1 - 3)}.
\end{multline}
Once again, we see that the estimate provided by Theorem~\ref{th:renormalization_bound} is correctly reproduced. The short memory factor can be extracted as before, by multiplying the right hand side by $1 = 2^{2(h^\star - N)} \cdot 2^{2(N - h^\star)}$. 

\newpage

\printbibliography

\end{document}